\documentclass[a4paper]{article}

\usepackage[top=2.5cm,bottom=2.5cm, left=2.5cm,right=2.5cm]{geometry}

\usepackage{amsmath,amssymb,amsthm,mathtools,dsfont}
\usepackage{tikz,graphicx}
\usepackage{tkz-euclide}
\usepackage{pict2e}
\usepackage{todonotes}
\theoremstyle{plain}
\usepackage[bottom]{footmisc}
\usepackage{cases}

\newtheorem*{theorem*}{Theorem}
\newtheorem{theorem}{Theorem}[section] 

\newtheorem{proposition}[theorem]{Proposition}
\newtheorem{corollary}[theorem]{Corollary}

\theoremstyle{definition}

\newtheorem{remark}[theorem]{Remark}

\usepackage{tikz}
\usepackage{pict2e}
\usetikzlibrary{arrows,calc,chains, positioning, shapes.geometric,shapes.symbols,decorations.markings,arrows.meta}
\usetikzlibrary{patterns,angles,quotes}
\usetikzlibrary{decorations.markings}
\tikzset{middlearrow/.style={
			decoration={markings,
				mark= at position 0.6 with {\arrow{#1}} ,
			},
			postaction={decorate}
		}
	}
\tikzset{->-/.style={decoration={
				markings,
				mark=at position #1 with {\arrow{latex}}},postaction={decorate}}}
	
	\tikzset{-<-/.style={decoration={
				markings,
				mark=at position #1 with {\arrowreversed{latex}}},postaction={decorate}}}

\numberwithin{equation}{section}

\def\bigO{{\mathcal O}}

\tikzset{
	master/.style={
		execute at end picture={
			\coordinate (lower right) at (current bounding box.south east);
			\coordinate (upper left) at (current bounding box.north west);
		}
	},
	slave/.style={
		execute at end picture={
			\pgfresetboundingbox
			\path (upper left) rectangle (lower right);
		}
	}
}

\tikzset{middlearrow/.style={
		decoration={markings,
			mark= at position 0.6 with {\arrow{#1}} ,
		},
		postaction={decorate}
	}
}

\newcommand{\re}{\text{\upshape Re\,}}
\newcommand{\im}{\text{\upshape Im\,}}

\newcommand{\diag}{\text{\upshape diag\,}}

\let\oldbibliography\thebibliography
\renewcommand{\thebibliography}[1]{\oldbibliography{#1}
\setlength{\itemsep}{-3pt}}

\allowdisplaybreaks

\usepackage[colorlinks=true]{hyperref}
\hypersetup{urlcolor=blue, citecolor=red, linkcolor=blue}

\title{On the generating function of the Pearcey process}
\date{}
\author{Christophe Charlier and Philippe Moreillon}
\date{\small Department of Mathematics, KTH Royal Institute of Technology, \\
\small 100 44 Stockholm, Sweden. \\ \smallskip
\small E-mail: cchar@kth.se, phmoreil@kth.se}

\begin{document}
\maketitle
\begin{abstract}
The Pearcey process is a universal point process in random matrix theory. In this paper, we study the generating function of the Pearcey process on any number $m$ of intervals. We derive an integral representation for it in terms of a Hamiltonian that is related to a system of $6m+2$ coupled nonlinear equations. We also obtain asymptotics for the generating function as the size of the intervals get large, up to and including the constant term. This work generalizes some recent results of Dai, Xu and Zhang, which correspond to $m=1$.
\end{abstract}

\medskip

\noindent
{\small{\sc AMS Subject Classification (2020)}: 41A60, 60B20, 30E25.}

\noindent
{\small{\sc Keywords}: Pearcey point process, Generating function asymptotics, Hamiltonian, Riemann--Hilbert problems.}

\section{Introduction and statement of results}



In random matrix theory, the universality conjecture  asserts that the microscopic behavior of the eigenvalues of large random matrices is similar for many different models. More precisely, it is expected that the local eigenvalue statistics only depend on the symmetry class of the matrix ensemble and on the nature of the point around which these statistics are considered \cite{ErdosYauUniversality, ForresterBook, KuijlaarsUniversality}. The Pearcey process is one of the canonical point processes from the theory of random matrices: it models the asymptotic behavior of the eigenvalues near the points of the spectrum where the density of states admits a cusp-like singularity. This process appears in Gaussian random matrices with an external source \cite{BreHik1, BreHik2, BleherKuijlaarsIII}, Hermitian random matrices with independent, not necessarily identically distributed entries \cite{ErdosKrugerSchroder}, and in general Wishart matrices with correlated entries \cite{HHN2}. The Pearcey process is also universal in a sense that goes beyond random matrix theory: it appears in certain models of skew plane partitions \cite{OkounkovReshetikhin} and of Brownian motions \cite{AdlerOrantinMoerbeke, AdlerMoerbeke, BleherKuijlaarsIII, TradWidPearcey, GeuZhang}.


\medskip The Pearcey process is the determinantal point process on $\mathbb{R}$ associated with the kernel
\begin{align}
K^{\mathrm{Pe}}_{\rho}(x,y) & = \frac{\mathcal{P}(x)\mathcal{Q}''(y)-\mathcal{P}'(x)\mathcal{Q}'(y) + \mathcal{P}''(x)\mathcal{Q}(y) - \rho \mathcal{P}(x) \mathcal{Q}(y)}{x-y}, \label{def of the pearcey Kernel}
\end{align}
where $\rho \in \mathbb{R}$,
\begin{align*}
\mathcal{P}(x) = \frac{1}{2\pi} \int_{-\infty}^{\infty}e^{-\frac{1}{4}t^{4}-\frac{\rho}{2}t^{2}+itx}dt, \qquad \mathcal{Q}(y) = \frac{1}{2\pi} \int_{\Sigma}e^{\frac{1}{4}t^{4}+\frac{\rho}{2}t^{2}+ity}dt,
\end{align*}
and $\Sigma = (e^{\frac{\pi i}{4}}\infty,0)\cup(0,e^{\frac{3\pi i}{4}}\infty) \cup (e^{-\frac{3\pi i}{4}}\infty,0)\cup(0,e^{-\frac{\pi i}{4}}\infty)$. As $\rho$ increases, the point configurations with few points near $0$ are increasingly likely to occur, and when $\rho$ tends to $+\infty$, the Pearcey process ``factorizes" (in the sense of gap probabilities) into two independent Airy processes \cite{BertolaCafasso}.  Important progress on the large gap asymptotics for any fixed $\rho \in \mathbb{R}$ have only recently been obtained in \cite{DXZ2020}. 

\medskip This paper is inspired by the work \cite{DXZ2020 thinning} of Dai, Xu and Zhang, and is concerned with the moment generating function of the Pearcey process. Let $X$ be a locally finite random point configuration distributed according to the Pearcey process, and let $N(x) := \#\{\xi \in X: \xi \in (-x,x)\}$ be the associated counting function. We are interested in the $m$-point generating function
\begin{align}\label{F as expectation intro}
F(r\vec{x},\vec{u}) := \mathbb{E}\bigg[ \prod_{j=1}^{m}e^{u_{j}N(rx_{j})} \bigg],
\end{align}
where 
\begin{align*}
r>0, \qquad m \in \mathbb{N}_{>0}, \qquad \vec{u}=(u_{1},\ldots,u_{m}) \in \mathbb{R}^{m}, \qquad  \vec{x}=(x_{1},\ldots,x_{m}) \in \mathbb{R}^{+,m}_{\mathrm{ord}},
\end{align*}
and $\mathbb{R}^{+,m}_{\mathrm{ord}}:=\{(x_{1},\ldots,x_{m}):0<x_{1}<\ldots<x_{m}<+\infty\}$. 

\medskip Our main results are stated in Theorems \ref{thm: hamiltonian representation of the Pearcey determinant} and \ref{thm: asymp of fredholm determinant} below and can be summarized as follows:
\begin{itemize}
\item Theorem \ref{thm: hamiltonian representation of the Pearcey determinant} establishes an integral representation for $F(r\vec{x},\vec{u})$ in terms of a Hamiltonian related a system of $6m+2$ coupled differential equations. This system of equations admits at least one solution, which we derive from the Lax pair of a Riemann-Hilbert (RH) problem. The asymptotic properties of this solution are stated in Theorem \ref{thm:asymp of pkj as r to inf}.
\item Theorem \ref{thm: asymp of fredholm determinant} is concerned with the asymptotic properties of the generating function as the size of the intervals get large. Specifically, Theorem \ref{thm: asymp of fredholm determinant} gives a precise asymptotic formula, up to and including the constant term, for $F(r\vec{x},\vec{u})$ as $r \to +\infty$.
\end{itemize}
For $m=1$, Theorems \ref{thm:asymp of pkj as r to inf}, \ref{thm: hamiltonian representation of the Pearcey determinant} and \ref{thm: asymp of fredholm determinant} have previously been obtained in \cite{DXZ2020 thinning}. The general case $m \geq 2$ allows to capture the correlation structure of the Pearcey process, see Corollary \ref{coro:correlation structure}, and to analyze the joint fluctuations of the counting function, see Corollary \ref{coro: central limit theorem}. We also expect the case $m=2$ of Theorem \ref{thm: asymp of fredholm determinant} to play an important role in future studies of the rigidity of the Pearcey process (we comment more on that at the end of this section).


\medskip The relevant system of $6m+2$ coupled differential equations depends on unknown functions which are denoted 
\begin{align}\label{unknown functions intro}
p_{0}(r), \; q_{0}(r), \; p_{j,k}(r), \; q_{j,k}(r), \qquad k=1,2,3, \quad j=1,\ldots,m,
\end{align} 
and is as follows:
\begin{align}\label{big system of ODEs}
& \begin{cases}
p_{0}'(r) = -\sqrt{2}\sum_{j=1}^{m}x_{j}p_{j,3}(r)q_{j,2}(r),\\
q_{0}'(r) = \sqrt{2}\sum_{j=1}^{m}x_{j}p_{j,2}(r)q_{j,1}(r),\\
q_{j,1}'(r) = \frac{2}{r}S_{11}(r)q_{j,1}(r) + x_{j}q_{j,2}(r) + \frac{2}{r}S_{31}(r)q_{j,3}(r), \\
q_{j,2}'(r) =  \sqrt{2}p_{0}(r)x_{j}q_{j,1}(r) + \frac{2}{r}S_{22}(r)q_{j,2}(r) + x_{j}q_{j,3}(r), \\
q_{j,3}'(r) = (rx_{j}^{2} + \frac{2}{r}S_{13}(r))q_{j,1}(r) + \sqrt{2}q_{0}(r)x_{j}q_{j,2}(r) + \frac{2}{r}S_{33}(r)q_{j,3}(r), \\
p_{j,1}'(r) = -\frac{2}{r}S_{11}(r)p_{j,1}(r) - \sqrt{2}p_{0}(r)x_{j}p_{j,2}(r) - (rx_{j}^{2} + \frac{2}{r}S_{13}(r))p_{j,3}(r), \\
p_{j,2}'(r) = -x_{j}p_{j,1}(r) - \frac{2}{r}S_{22}(r)p_{j,2}(r) - \sqrt{2}q_{0}(r)x_{j}p_{j,3}(r), \\
p_{j,3}'(r) = -\frac{2}{r}S_{31}(r)p_{j,1}(r) - x_{j}p_{j,2}(r) - \frac{2}{r}S_{33}(r)p_{j,3}(r),
\end{cases}
\end{align}
where $j=1,\ldots,m$, and
\begin{align*}
S_{kl}(r) = \sum_{j=1}^{m} p_{j,k}(r)q_{j,l}(r), \qquad k,l=1,2,3.
\end{align*}
Furthermore, we require the functions \eqref{unknown functions intro} to satisfy the following $m$ relations
\begin{align}\label{sum relation between qj and pj}
\sum_{k=1}^{3} p_{j,k}(r)q_{j,k}(r) = 0, \qquad j=1,\ldots,m.
\end{align}
Let 
\begin{align*}
H(r)=H(r;p_{0},q_{0},\{p_{j,1},q_{j,1},p_{j,2},q_{j,2},p_{j,3},q_{j,3}\}_{j=1}^{m})
\end{align*}
be defined by
\begin{align}
H(r) = & \; \sqrt{2}p_{0}(r)\sum_{j=1}^{m}x_{j}p_{j,2}(r)q_{j,1}(r) + \sqrt{2}q_{0}(r)\sum_{j=1}^{m}x_{j}p_{j,3}(r)q_{j,2}(r) \nonumber \\
& + \sum_{j=1}^{m} x_{j}p_{j,1}(r)q_{j,2}(r) + \sum_{j=1}^{m} x_{j}p_{j,2}(r)q_{j,3}(r) +\sum_{j=1}^{m}rx_{j}^{2}p_{j,3}(r)q_{j,1}(r) \nonumber \\
& + \frac{1}{2r}\bigg(\Big( S_{11}(r) - S_{22}(r) + S_{33}(r) \Big)^{2} \nonumber \\
& - 2 \sum_{k=1}^{m}\sum_{\ell=1}^{m}\Big(p_{k,1}(r)p_{\ell,3}(r)-p_{k,3}(r)p_{\ell,1}(r)\Big)\Big(q_{k,1}(r)q_{\ell,3}(r)-q_{k,3}(r)q_{\ell,1}(r)\Big) \bigg). \label{def of Hamiltonian}
\end{align}
It is readily checked that $q_{0}'(r)=\frac{\partial H}{\partial p_{0}}(r)$, $p_{0}'(r)=-\frac{\partial H}{\partial q_{0}}(r)$ and
\begin{align}\label{Hamiltonian relation}
q_{j,k}'(r) = \frac{\partial H}{\partial p_{j,k}}(r), \qquad p_{j,k}'(r) = -\frac{\partial H}{\partial q_{j,k}}(r), \qquad j=1,\ldots,m, \; k=1,2,3,
\end{align}
and therefore $H$ is a Hamiltonian for the system \eqref{big system of ODEs}--\eqref{sum relation between qj and pj}.
\begin{theorem}\label{thm:asymp of pkj as r to inf}
Let $\rho \in \mathbb{R}$,
\begin{align*}
r>0, \qquad m \in \mathbb{N}_{>0}, \qquad \vec{u}=(u_{1},\ldots,u_{m}) \in \mathbb{R}^{m}, \quad \mbox{ and } \quad  \vec{x}=(x_{1},\ldots,x_{m}) \in \mathbb{R}^{+,m}_{\mathrm{ord}}.
\end{align*}
There exists at least one solution $(p_{0},q_{0},\{p_{j,1},q_{j,1},p_{j,2},q_{j,2},p_{j,3},q_{j,3}\}_{j=1}^{m})$ to the system of equations \eqref{big system of ODEs} and \eqref{sum relation between qj and pj} satisfying the following asymptotics. As $r \to + \infty$, we have
\begin{subequations}\label{all asymp as r to inf}
\begin{align}
& p_{0}(r) = \frac{\sqrt{3}}{2\sqrt{2}\pi}\sum_{\ell=1}^{m}u_{\ell}x_{\ell}^{\frac{2}{3}}r^{\frac{2}{3}} + \frac{1}{\sqrt{2}}\bigg( \frac{\rho^{3}}{54}+\frac{\rho}{2} \bigg) + \bigO(r^{-\frac{2}{3}}), \label{p0 large r in thm} \\
& p_{j,1}(r) = -\frac{1}{3\pi i} e^{\frac{1}{2}\theta_{3}(rx_{j})}  (rx_{j})^{\frac{1}{3}}\frac{u_{j}}{\mathcal{A}_{j}} \nonumber \\
& \hspace{1.5cm} \times \bigg( \cos(\vartheta_{j}(r)-\tfrac{\pi}{3}) + \frac{\sqrt{3}}{2\pi} \sum_{\ell=1}^{m} u_{\ell}\frac{x_{\ell}^{2/3}}{x_{j}^{2/3}}\cos(\vartheta_{j}(r)+\tfrac{\pi}{3}) \bigg)(1+\bigO(r^{-\frac{2}{3}})) \label{p1 large r in thm} \\
& p_{j,2}(r) = \frac{1}{3\pi i} e^{\frac{1}{2}\theta_{3}(rx_{j})}  \frac{u_{j}}{\mathcal{A}_{j}}  \cos(\vartheta_{j}(r))  (1+\bigO(r^{-\frac{2}{3}})) \label{p2 large r in thm} \\
& p_{j,3}(r) = -\frac{1}{3\pi i} e^{\frac{1}{2}\theta_{3}(rx_{j})}  (rx_{j})^{-\frac{1}{3}}\frac{u_{j}}{\mathcal{A}_{j}}  \cos(\vartheta_{j}(r)+\tfrac{\pi}{3}) (1+\bigO(r^{-\frac{2}{3}})) \label{p3 large r in thm} \\
& q_{0}(r) = -\frac{\sqrt{3}}{2\sqrt{2}\pi} \sum_{\ell=1}^{m}u_{\ell}x_{\ell}^{\frac{2}{3}}r^{\frac{2}{3}} + \frac{1}{\sqrt{2}}\bigg( -\frac{\rho^{3}}{54}+\frac{\rho}{2} \bigg) + \bigO(r^{-\frac{2}{3}}), \label{q0 large r in thm} \\
& q_{j,1}(r) = 2ie^{-\frac{1}{2}\theta_{3}(r x_{j})}(r x_{j})^{-\frac{1}{3}} \mathcal{A}_{j}\sin(\vartheta_{j}(r)-\tfrac{\pi}{3})(1+\bigO(r^{-\frac{2}{3}})), \label{q1 large r in thm} \\
& q_{j,2}(r) = -2ie^{-\frac{1}{2}\theta_{3}(r x_{j})} \mathcal{A}_{j}\sin(\vartheta_{j}(r))(1+\bigO(r^{-\frac{2}{3}})), \label{q2 large r in thm} \\
& q_{j,3}(r) = 2ie^{-\frac{1}{2}\theta_{3}(r x_{j})}(r x_{j})^{\frac{1}{3}} \mathcal{A}_{j} \nonumber\\
& \hspace{1.5cm}\times \bigg(\sin(\vartheta_{j}(r)+\tfrac{\pi}{3})-\frac{\sqrt{3}}{2\pi}\sum_{\ell=1}^{m}u_{\ell} \frac{x_{\ell}^{2/3}}{x_{j}^{2/3}}\sin(\vartheta_{j}(r)-\tfrac{\pi}{3})\bigg)(1+\bigO(r^{-\frac{2}{3}})), \label{q3 large r in thm}
\end{align}
\end{subequations}
where $\theta_{3}(r)=\frac{3}{4}r^{\frac{4}{3}}+\frac{\rho}{2}r^{\frac{2}{3}}$, 
\begin{align}
& \mathcal{A}_{j} = |\Gamma(1-\tfrac{u_{j}}{2\pi i})|\exp \bigg( -\frac{u_{j}}{3}-\sum_{k=j+1}^{m} \frac{u_{k}}{2}-\sum_{\substack{k=1 \\ k \neq j}}^{m} \frac{u_{k}}{2 \pi}\arctan \frac{\sqrt{3}x_{k}^{2/3}}{x_{k}^{2/3}+2x_{j}^{2/3}} \bigg), \label{def of mathcal A j} \\
& \vartheta_{j}(r) = -\frac{3\sqrt{3}}{8}(rx_{j})^{\frac{4}{3}}+\frac{\sqrt{3}\rho}{4}(rx_{j})^{\frac{2}{3}} + \arg \Gamma(1-\tfrac{u_{j}}{2\pi i}) \nonumber \\
& \hspace{1.2cm} - \frac{u_{j}}{2\pi} \bigg( \frac{4}{3}\log (rx_{j}) + \log \frac{9}{2} \bigg) - \sum_{\substack{k=1 \\ k \neq j}}^{m} \frac{u_{k}}{2\pi}\log \frac{|x_{j}^{2/3}-\omega x_{k}^{2/3}|}{|x_{j}^{2/3}- x_{k}^{2/3}|}, \label{def of vartheta j}
\end{align}
$\omega:=e^{\frac{2\pi i}{3}}$ and $\Gamma$ is Euler's Gamma function. As $r \to 0$, we have
\begin{subequations}\label{all asymp as r to 0}
\begin{align}
& p_{0} = \frac{1}{\sqrt{2}}\bigg( \frac{\rho^{3}}{54}+\frac{\rho}{2} \bigg) + \bigO(r), \qquad q_{0} = \frac{1}{\sqrt{2}}\bigg( -\frac{\rho^{3}}{54}+\frac{\rho}{2} \bigg) + \bigO(r) \label{asymp of p0 and q0 as r to 0} \\
& p_{j,1}(r) = \bigO(r), \qquad p_{j,2}(r) = \bigO(1), \qquad p_{j,3}(r) = \bigO(r), \qquad j=1,\ldots,m, \label{asymp of pjk as r to 0} \\
& q_{j,1}(r) = \bigO(1), \qquad \hspace{0.017cm} q_{j,2}(r) = \bigO(r), \qquad \hspace{0.017cm} q_{j,3}(r) = \bigO(1), \qquad \hspace{0.017cm} j=1,\ldots,m. \label{asymp of qjk as r to 0}
\end{align}
\end{subequations}
\end{theorem}
\begin{remark}
Interestingly, for $m=1$, the functions $p_{0}$ and $q_{0}$ satisfy a system of coupled differential equations, see \cite[equations (2.19)--(2.20)]{DXZ2020 thinning} (see also \cite[equations (3.25)--(3.26)]{BreHik1} for $\rho=0$). However, it is unclear to us if this result admits a natural analogue for $m\geq 2$.
\end{remark}
It is well-known since the works of Jimbo, Miwa, M\^{o}ri and Sato \cite{JMMS1980} and of Tracy and Widom \cite{TraWidAiry, TraWidBessel} that the $1$-point generating functions of the universal sine, Airy and Bessel point processes are naturally related to the Painlev\'{e} theory. We also refer to \cite{HarnardTracyWidom, AGZ2010} for some Hamiltonian structures associated with the $m$-point functions of the sine and Airy processes, and to \cite{ClDoe, ChDoe} for some representations of the $m$-point functions of the Airy and Bessel processes in terms of the solution to a system of $m$ coupled Painlev\'{e} equations. It was shown in \cite[Theorem 2.2]{DXZ2020 thinning} that the $1$-point function of the Pearcey process admits an elegant representation in terms of a Hamiltonian associated to a system of 8 coupled differential equations. Our first main result generalizes \cite[Theorem 2.2]{DXZ2020 thinning} to an arbitrary $m$.
\begin{theorem}\label{thm: hamiltonian representation of the Pearcey determinant}
Let $\rho \in \mathbb{R}$,
\begin{align*}
r>0, \qquad m \in \mathbb{N}_{>0}, \qquad \vec{u}=(u_{1},\ldots,u_{m}) \in \mathbb{R}^{m}, \quad \mbox{ and } \quad  \vec{x}=(x_{1},\ldots,x_{m}) \in \mathbb{R}^{+,m}_{\mathrm{ord}}.
\end{align*}
The following relation holds
\begin{align}\label{integral representation of F}
F(r\vec{x},\vec{u}) = \exp \bigg( 2 \int_{0}^{r}H(\tau)d \tau \bigg),
\end{align}
with $H$ given by \eqref{def of Hamiltonian}, and where $(p_{0},q_{0},\{p_{j,1},q_{j,1},p_{j,2},q_{j,2},p_{j,3},q_{j,3}\}_{j=1}^{m})$ is a solution to the system of equations \eqref{big system of ODEs} and \eqref{sum relation between qj and pj} which satisfies the asymptotic formulas \eqref{all asymp as r to inf} and \eqref{all asymp as r to 0}. Furthermore,
\begin{align}\label{asymp of H as r to 0 in thm}
H(r) = \bigO(1), \qquad \mbox{as } r \to 0,
\end{align}
and as $r \to +\infty$, 
\begin{align}\label{asymp of H as r to inf in thm}
& H(r) = \sum_{j=1}^{m} \bigg( \frac{\sqrt{3}}{2\pi}u_{j}x_{j}^{\frac{4}{3}}r^{\frac{1}{3}} - \frac{\rho}{2\sqrt{3}\pi} u_{j} x_{j}^{\frac{2}{3}}r^{-\frac{1}{3}} + \frac{u_{j}^{2}}{3\pi^{2}r} - \frac{u_{j}}{3\sqrt{3}\pi r}\cos (2\vartheta_{j}(r)) \bigg) +\bigO(r^{-\frac{5}{3}}),
\end{align}
where $\vartheta_{j}(r)$ is defined in \eqref{def of vartheta j}.
\end{theorem}
Since the functions $(p_{0},q_{0},\{p_{j,1},q_{j,1},p_{j,2},q_{j,2},p_{j,3},q_{j,3}\}_{j=1}^{m})$ appearing in the integral representation \eqref{integral representation of F} are rather complicated objects, it is natural to try to approximate $F(r\vec{x},\vec{u})$ for small and large values of $r$ with some explicit asymptotic formulas. Because $F(r\vec{x},\vec{u})$ is a Fredholm determinant (see \eqref{fredholm det piecewise constant thinning} below), the asymptotics of $F(r\vec{x},\vec{u})$ as $r \to 0$ can be easily obtained from an analysis of the kernel $K_{\rho}^{\mathrm{Pe}}(x,y)$ near $(x,y)=(0,0)$. A much more complicated question is to approximate $F(r\vec{x},\vec{u})$ for large values of $r$.  In the case of the sine, Airy and Bessel processes, the asymptotics for the $m$-point generating functions are known, see \cite{BW1983, ChSine} for sine, \cite{BothnerBuckingham, ChCl3} for Airy, \cite{BIP2019, ChBessel} for Bessel and \cite{ChLen Bessel} for the transition between Bessel and Airy.  The asymptotics for the $1$-point generating function of the Pearcey process, up to and including the notoriously difficult constant term, have recently been established in \cite[Theorem 2.3]{DXZ2020 thinning}. We provide here the generalization of \cite[Theorem 2.3]{DXZ2020 thinning} to an arbitrary $m$.
\begin{theorem}\label{thm: asymp of fredholm determinant}
Let 
\begin{align*}
\rho \in \mathbb{R}, \qquad m \in \mathbb{N}_{>0}, \qquad \vec{u}=(u_{1},\ldots,u_{m}) \in \mathbb{R}^{m}, \quad \mbox{ and } \quad  \vec{x}=(x_{1},\ldots,x_{m}) \in \mathbb{R}^{+,m}_{\mathrm{ord}}.
\end{align*}
As $r \to +\infty$, we have
\begin{multline}
F(r\vec{x},\vec{u})=\exp\bigg(   \sum_{j=1}^{m}u_{j}\mu_{\rho}(rx_{j}) + \sum_{j=1}^{m}\frac{u_{j}^{2}}{2}\sigma^{2}(rx_{j}) + \sum_{1\leq j < k \leq m} u_{j}u_{k}\Sigma(x_{k},x_{j})  \\
 + \sum_{j=1}^{m} 2 \log \big( G(1-\tfrac{u_{j}}{2\pi i})G(1+\tfrac{u_{j}}{2\pi i}) \big) + \bigO (r^{-\frac{2}{3}}) \bigg), \label{asymptotics in main thm}
\end{multline}
where $G$ is Barnes' $G$-function, and $\mu_{\rho}$, $\sigma^{2}$ and $\Sigma$ are given by
\begin{align*}
& \mu_{\rho}(x) = \frac{3\sqrt{3}}{4\pi}x^{\frac{4}{3}}- \frac{\sqrt{3}\rho}{2\pi} x^{\frac{2}{3}}, \\
& \sigma^{2}(x) = \frac{4}{3\pi^{2}}  \log x+\frac{1}{\pi^{2}} \log \frac{9}{2}, \\
& \Sigma(x_{k},x_{j}) = \frac{1}{\pi^{2}} \log  \frac{|x_{j}^{2/3}-\omega x_{k}^{2/3}|}{|x_{j}^{2/3}- x_{k}^{2/3}|},
\end{align*}
where $\omega =e^{\frac{2\pi i}{3}}$. Furthermore, \eqref{asymptotics in main thm} holds uniformly for $\rho$ in compact subsets of $\mathbb{R}$, for $\vec{u}$ in compact subsets of $\mathbb{R}^{m}$, and for $\vec{x}$ in compact subsets of $\mathbb{R}^{+,m}_{\mathrm{ord}}$. The asymptotic formula \eqref{asymptotics in main thm} can also be differentiated any number of times with respect to $u_{1},\ldots,u_{m}$ at the expense of increasing the error term in the following way. Let $\widetilde{F}(r\vec{x},\vec{u})$ be the right-hand side of \eqref{asymptotics in main thm} without the error term, and denote the error term by $\mathcal{E}=\log F(r\vec{x},\vec{u}) - \log \widetilde{F}(r\vec{x},\vec{u})$. For any $k_{1},\ldots,k_{m} \in \mathbb{N}_{\geq 0}$, we have
\begin{align}\label{derivative of error term}
\partial_{u_{1}}^{k_{1}}\ldots \partial_{u_{m}}^{k_{m}} \mathcal{E} = \bigO \bigg( \frac{(\log r)^{k_{1}+\ldots+k_{m}}}{r^{2/3}} \bigg), \qquad \mbox{as } r \to + \infty.
\end{align}
\end{theorem}
We end this section by providing several new applications of Theorem \ref{thm: asymp of fredholm determinant}, and we also discuss its relevance in future studies of the rigidity of the Pearcey process.
\paragraph{Applications of Theorem \ref{thm: asymp of fredholm determinant}.} Using \eqref{asymptotics in main thm} with $m=1$, \eqref{derivative of error term},
\begin{align*}
\partial_{u} \log F(r,u)|_{u=0} = \mathbb{E}[N(r)] \quad \mbox{ and } \quad \partial_{u}^{2} \log F(r,u)|_{u=0} = \mbox{Var}[N(r)],
\end{align*}
it readily follows that
\begin{align}
& \mathbb{E}[N(r)] = \mu_{\rho}(r) + \bigO\bigg( \frac{\log r}{r^{2/3}} \bigg), & & \mbox{as } r \to + \infty, \label{asymp for expectation} \\
& \mbox{Var}[N(r)] = \sigma^{2}(r) + \frac{1+\gamma_{\mathrm{E}}}{\pi^{2}} + \bigO\bigg( \frac{(\log r)^{2}}{r^{2/3}} \bigg), & & \mbox{as } r \to + \infty, \label{asymp for variance}
\end{align}
where $\gamma_{\mathrm{E}}$ is Euler's gamma constant. These asymptotic formulas for the expectation and variance of $N(r)$ are not new and were already obtained in \cite[equations (2.30) and (2.31)]{DXZ2020 thinning}. We see from \eqref{asymp for variance} that the large $r$ asymptotics of $\mbox{Var}[N(r)]$ are of the form $c_{1}\log r + c_{2} + o(1)$ for some explicit constants $c_{1}$ and $c_{2}$. The asymptotics for the variance of the counting functions of other classical point processes such as the sine, Airy and Bessel point processes are also of the same form \cite{SoshnikovSineAiryBessel, ChSine, ChCl3, ChBessel} (with different values for $c_{1}$ and $c_{2}$). This phenomena is expected to be universal, in the sense that it is expected to hold for many point processes in random matrix theory and other related fields, see the very general predictions \cite{SDMS2020, SDMS2021}. We emphasize that the proof of \eqref{asymp for expectation} and \eqref{asymp for variance} only relies on Theorem \ref{thm: asymp of fredholm determinant} with $m=1$. Using Theorem \ref{thm: asymp of fredholm determinant} with $m=2$ allows to obtain new results on the correlation structure of the Pearcey process. More precisely, using \eqref{asymptotics in main thm} with $m=2$, \eqref{derivative of error term}, and
\begin{align*}
\partial_{u}^{2} \log \bigg( \frac{F((x_{1},x_{2}),(u,u))}{F(x_{1},u)F(x_{2},u)} \bigg) \bigg|_{u=0} = 2 \, \mbox{Cov}\big(N(x_{1}),N(x_{2})\big),
\end{align*}
we directly obtain the following result.
\begin{corollary}\label{coro:correlation structure}
Let $x_{2}>x_{1}>0$ be fixed. As $r \to + \infty$, 
\begin{align*}
\mathrm{Cov}\big(N(rx_{1}),N(rx_{2})\big) = \Sigma(x_{k},x_{j}) + \bigO \bigg( \frac{(\log r)^{2}}{r^{2/3}} \bigg).
\end{align*}
\end{corollary}
In \cite[Corollary 2.4]{DXZ2020 thinning}, the authors also proved that the random variable $(N(r)-\mu_{\rho}(r))/\sqrt{\sigma^{2}(r)}$ converges in distribution as $r \to + \infty$ to a normal random variable with mean $0$ and variance $1$. Using Theorem \ref{thm: asymp of fredholm determinant}, we obtain the following generalization of this result.
\begin{corollary}\label{coro: central limit theorem}
Let $0<x_{1}<\ldots<x_{m}<+\infty$ be fixed and consider the random variables $N_{j}^{(r)}$ defined by
\begin{align*}
N_{j}^{(r)} = \frac{N(rx_{j})-\mu_{\rho}(rx_{j})}{\sqrt{\sigma^{2}(rx_{j})}}, \qquad j=1,\ldots,m.
\end{align*}
As $r \to +\infty$, we have
\begin{align}\label{convergence in distribution 1}
\big( N_{1}^{(r)},N_{2}^{(r)},\ldots,N_{m}^{(r)}\big) \quad \overset{d}{\longrightarrow} \quad \mathcal{N}(\vec{0},I_{m}),
\end{align}
where $I_{m}$ is the $m \times m$ identity matrix, and $\mathcal{N}(\vec{0},I_{m})$ is a multivariate normal random variable of mean $\vec{0}=(0,\ldots,0)$ and covariance matrix $I_{m}$.
\end{corollary}
\begin{proof}
Let $a_{1},\ldots,a_{m} \in \mathbb{R}$ be arbitrary and fixed (i.e. independent of $r$). It directly follows from \eqref{F as expectation intro} and \eqref{asymptotics in main thm} with $u_{j}=\frac{\sqrt{3}\pi}{2}\frac{a_{j}}{\sqrt{\log r}}$, $j=1,\ldots,m$, that
\begin{align*}
\mathbb{E}\bigg[ \prod_{j=1}^{m}e^{a_{j}N_{j}^{(r)}} \bigg] = \exp \bigg( \sum_{j=1}^{m} \frac{a_{j}^{2}}{2} + \bigO\bigg( \frac{1}{\sqrt{\log r}} \bigg) \bigg), \qquad \mbox{as } r \to + \infty.
\end{align*}
In other words, the moment generating function of $\big( N_{1}^{(r)},N_{2}^{(r)},\ldots,N_{m}^{(r)}\big)$ converges as $r \to + \infty$ pointwise in $\mathbb{R}^{m}$ to the moment generating function of $\mathcal{N}(\vec{0},I_{m})$. This implies the convergence in distribution \eqref{convergence in distribution 1} by standard probability theorems, see e.g. \cite[Corollary of Theorem 25.10]{Billingsley}.
\end{proof}
\paragraph{Possible future applications of Theorem \ref{thm: asymp of fredholm determinant}.} Corollary \ref{coro: central limit theorem} gives information about the joint fluctuations of the counting function at $m$ well-separated points $rx_{1},\ldots,rx_{m}$. A more difficult question is to understand the \textit{global rigidity} of the Pearcey process, that is, to understand the \textit{maximum} fluctuation of the counting function. In recent years in random matrix theory, there has been a lot of progress in the study of ridigity of various point processes, see \cite{ErdosYauYin, ArguinBeliusBourgade} for important early works, \cite{HolcombPaquette} for the sine process, and \cite{ChCl4} for the Airy and Bessel point processes. Of particular interest for us is the following result on the rigidity of the Pearcey process, which was proved in \cite{CharlierRigidityPearcey} (by combining results from \cite{ChCl4} and \cite{DXZ2020 thinning}): for any $\epsilon > 0$, the probability that
\begin{align}\label{rewriting of thm 1}
\mu_{\rho}(x)- \bigg( \frac{4\sqrt{2}}{3\pi}+\epsilon \bigg) \log x \leq N(x) \leq \mu_{\rho}(x) + \bigg( \frac{4\sqrt{2}}{3\pi}+\epsilon \bigg) \log x \qquad \mbox{for all }x>r
\end{align}
tends to $1$ as $r \to + \infty$. Roughly speaking, this means that with high probability and for all large $x$, $N(x)$ lies in a tube centered at $\mu_{\rho}(x)$ and of width $\big( \frac{8\sqrt{2}}{3\pi}+2\epsilon \big) \log x$, see Figure \ref{fig:rigidity of the counting function} (left). Equivalently, \eqref{rewriting of thm 1} can be rewritten for the normalized counting function as follows
\begin{align}
& \lim_{r \to \infty}\mathbb P\left(\sup_{x> r}\left|\frac{N(x)- \mu_{\rho}(x)}{\log x}\right| \leq  \frac{4\sqrt{2}}{3\pi} + \epsilon \right) = 1, \label{upper bound rigidity 1}
\end{align}
see Figure \ref{fig:rigidity of the counting function} (right). It has also been conjectured in \cite{CharlierRigidityPearcey} that the upper bound \eqref{upper bound rigidity 1} is sharp, in the sense that the following complementary lower bound is expected to hold: for any $\epsilon > 0$,
\begin{align}\label{completementary lower bound}
& \lim_{r \to \infty}\mathbb P\left(\sup_{x> r}\left|\frac{N(x)- \mu_{\rho}(x)}{\log x}\right| \geq  \frac{4\sqrt{2}}{3\pi} - \epsilon \right) = 1,
\end{align}
and this is also supported by Figure \ref{fig:rigidity of the counting function} (right). 
\begin{figure}[h]
\begin{center}
\begin{tikzpicture}
\node at (0,0) {\includegraphics[scale=0.3]{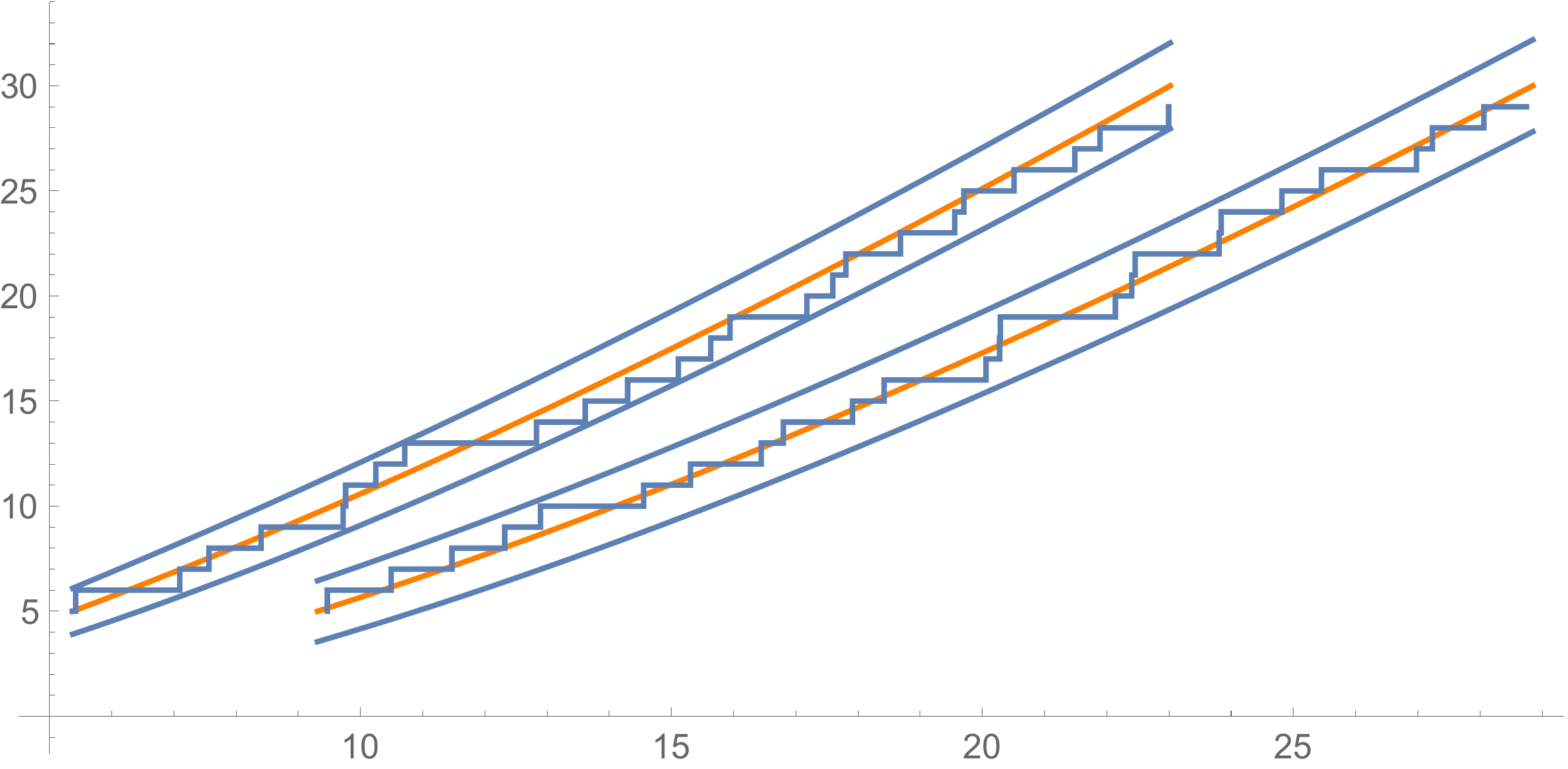}};
\node at (0.2,1.3) {\small $\rho=-1.31$};
\node at (2.3,0.2) {\small $\rho=2.54$};
\end{tikzpicture}
\begin{tikzpicture}
\node at (0,0) {\includegraphics[scale=0.3]{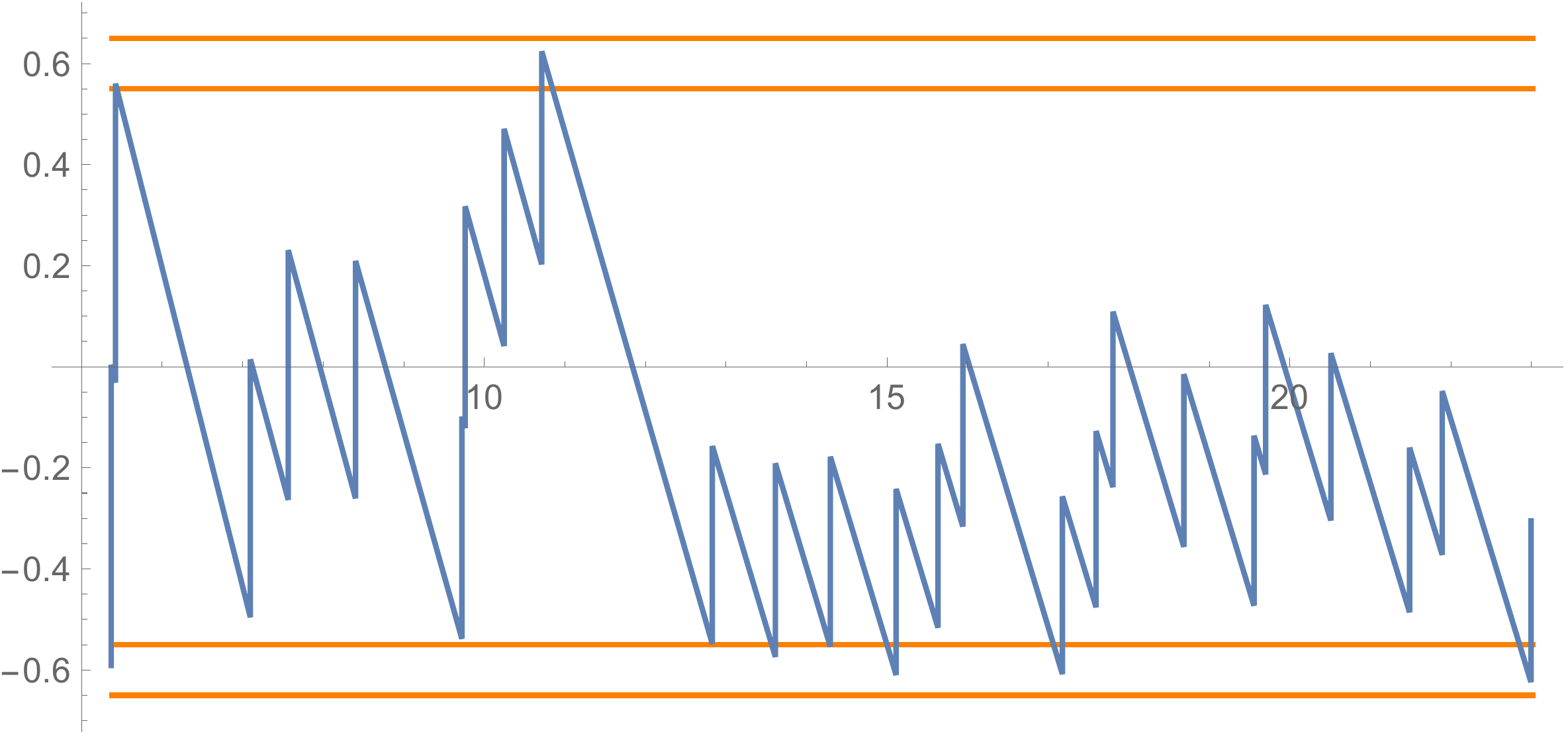}};
\node at (2.5,1) {\small $\rho=-1.31$};
\end{tikzpicture}
\end{center}
\vspace{-0.5cm}\caption{\label{fig:rigidity of the counting function}Rigidity of the Pearcey process (the pictures are taken from \cite{CharlierRigidityPearcey}). Left: the smooth blue lines correspond to the upper and lower bounds in \eqref{rewriting of thm 1} with $\epsilon=0.05$, and $N(x)$ is the discontinuous blue line. Right: the blue line is $\frac{N(x)- \mu_{\rho}(x)}{\log x}$, and the four orange lines correspond to the constants $\pm \frac{4\sqrt{2}}{3\pi} + \epsilon$, $\pm\frac{4\sqrt{2}}{3\pi} - \epsilon$ with $\epsilon=0.05$.} 
\end{figure}
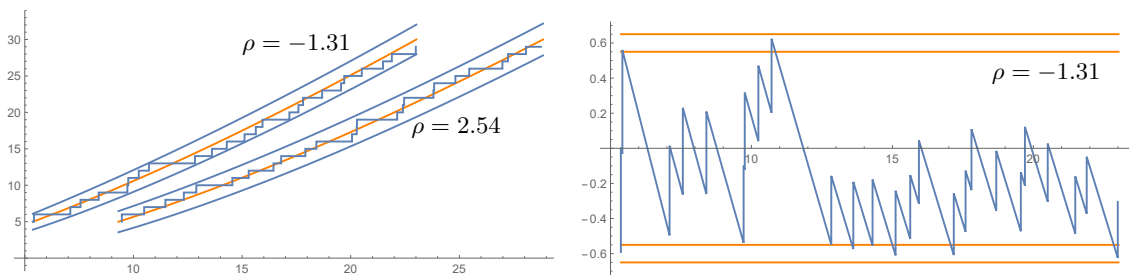
Such lower bounds are notoriously difficult to prove. 

By analogy with the method developed in \cite{CFLW}, we expect that Theorem \ref{thm: asymp of fredholm determinant} with $m=2$ will be useful to establish \eqref{completementary lower bound}. However, we also expect that proving \eqref{completementary lower bound} will also require other estimates that are not provided in this paper, such as the large $r$ asymptotics for $\mathbb{E}[e^{u_{1}N(rx_{1})+u_{2}N(rx_{2})}]$ when simultaneously $|x_{1}-x_{2}|\to 0$. This regime requires a completely different analysis than the one of Theorem \ref{thm: asymp of fredholm determinant}, and we shall not pursue this here.

\paragraph{Outline.} Relying on the fact that the Pearcey kernel $K_{\rho}^{\mathrm{Pe}}$ is known to be integrable (of size $3$) in the sense of \cite{IIKS}, we use in Section \ref{section: differential identity} the general method of \cite{DIZ} to express $\partial_{r}\log F(r\vec{x},\vec{u})$ in terms of the solution $\Phi$ to a $3\times 3$ RH problem. In Section \ref{section: Lax pair}, we derive the system of equations \eqref{big system of ODEs} and establish the formula $\partial_{r}\log F(r\vec{x},\vec{u})=2H(r)$ by analyzing a natural Lax pair associated to $\Phi$. In Section \ref{section: steepest descent for large r}, we use the Deift--Zhou \cite{DeiftZhou} steepest descent method to obtain the large $r$ asymptotics of $\Phi$. A main technical challenge here is to analyze the behavior of the global parametrix at certain points, which becomes particularly delicate for $m \geq 2$, see e.g. the asymptotic formulas of Subsection \ref{subsubsection: asymptotics of global param near xj}. The steepest descent analysis of $\Phi$ for small $r$ is simpler and is performed in Section \ref{section: steepest descent for small r}. In Section \ref{section: proof of main results}, we use the small and large $r$ asymptotics of $\Phi$ together with some remarkable identities for the Hamiltonian to complete the proofs of Theorems \ref{thm:asymp of pkj as r to inf}, \ref{thm: hamiltonian representation of the Pearcey determinant} and \ref{thm: asymp of fredholm determinant}.


\section{Differential identity}\label{section: differential identity}
The main result of this section is a differential identity which expresses $\partial_{r}\log F(r\vec{x},\vec{u})$ in terms of the solution $\Phi$ to a $3\times 3$ RH problem.

\medskip It is well-known, see e.g. \cite[Theorem 2]{Soshnikov}, that the moment generating function \eqref{F as expectation intro} is equal to the Fredholm determinant
\begin{align}\label{fredholm det piecewise constant thinning}
F(r\vec{x},\vec{u}) = \det \big(1-\widetilde{\mathcal{K}}^{\mathrm{Pe}}_{\rho}\big), \qquad \widetilde{\mathcal{K}}^{\mathrm{Pe}}_{\rho} := \sum_{j=1}^{m}(1-s_{j})\mathcal{K}^{\mathrm{Pe}}_{\rho}|_{rA_{j}},
\end{align}
where $s_{j} := e^{u_{j}+\ldots+u_{m}} \in (0,+\infty)$, $j=1,\ldots,m$,
\begin{align*}
A_{1} = (-x_{1},x_{1}), \qquad A_{j} = (-x_{j},-x_{j-1})\cup (x_{j-1},x_{j}), \quad j=2,\ldots,m,
\end{align*}
and $\mathcal{K}^{\mathrm{Pe}}_{\rho}|_{rA_{j}}$ is the trace-class operator acting on $L^{2}(rA_{j})$ whose kernel is $K^{\mathrm{Pe}}_{\rho}$. We now recall a formula from \cite{BleherKuijlaarsIII} which expresses $K_{\rho}^{\mathrm{Pe}}$ in terms of the solution $\Psi$ of a $3\times 3$ RH problem.

\subsection{Background from \cite{BleherKuijlaarsIII}}\label{subsection: model RH problem for Psi}
\subsubsection*{RH problem for $\Psi$}
\begin{itemize}
\item[(a)] $\Psi: \mathbb{C}\setminus \{\cup_{j=0}^{5}\Sigma_{j}\cup\{0\}\}\to \mathbb{C}^{3\times 3}$ is analytic, where
\begin{align}
& \Sigma_{0} = (0,+\infty), & & \Sigma_{1} = e^{\frac{\pi i}{4}}(0,+\infty), & & \Sigma_{2} = e^{\frac{3\pi i}{4}}(+\infty,0), \nonumber \\
& \Sigma_{3} = (-\infty,0), & & \Sigma_{4} = e^{-\frac{3\pi i}{4}}(+\infty,0), & & \Sigma_{5} = e^{-\frac{\pi i}{4}}(0,+\infty).  \label{contour of Psi}
\end{align}
\item[(b)] For $z \in \cup_{j=0}^{5}\Sigma_{j}$, we denote $\Psi_{+}(z)$ (resp. $\Psi_{-}(z)$) for the limit of $\Psi(s)$ as $s \to z$ from the left (resp. right) of  $\cup_{j=0}^{5}\Sigma_{j}$ (here ``left" and ``right" refer to the orientation of $\cup_{j=0}^{5}\Sigma_{j}$ as indicated in \eqref{contour of Psi}). For $z \in \Sigma_{j}$, we have $\Psi_{+}(z) = \Psi_{-}(z)J_{j}$, $j=0,\ldots,5$, where $J_{0},J_{1},J_{2},J_{3},J_{4}$ and $J_{5}$ are respectively given by
\begin{align}\label{def of Jj}
& \begin{pmatrix}
0 & 1 & 0 \\
-1 & 0 & 0 \\
0 & 0 & 1
\end{pmatrix}, \begin{pmatrix}
1 & 0 & 0 \\
1 & 1 & 1 \\
0 & 0 & 1
\end{pmatrix}, \begin{pmatrix}
1 & 0 & 0 \\
0 & 1 & 0 \\
1 & 1 & 1
\end{pmatrix}, \begin{pmatrix}
0 & 0 & 1 \\
0 & 1 & 0 \\
-1 & 0 & 0
\end{pmatrix}, \begin{pmatrix}
1 & 0 & 0 \\
0 & 1 & 0 \\
1 & -1 & 1
\end{pmatrix}, \begin{pmatrix}
1 & 0 & 0 \\
1 & 1 & -1 \\
0 & 0 & 1
\end{pmatrix}.
\end{align}
\item[(c)] As $z \to \infty$, $\pm \im z > 0$, 
\begin{align}\label{asymp of Psi at infty}
\Psi(z) = \sqrt{\frac{2\pi}{3}}e^{\frac{\rho^{2}}{6}}i \Psi_{0} \bigg( I + \frac{\Psi_{1}}{z} + \bigO(z^{-2}) \bigg)\diag(z^{-\frac{1}{3}},1,z^{\frac{1}{3}})L_{\pm}e^{\Theta(z)},
\end{align}
where
\begin{align}
& \Psi_{0} = \begin{pmatrix}
1 & 0 & 0 \\
0 & 1 & 0 \\
\kappa_{3}(\rho)+\frac{2\rho}{3} & 0 & 1
\end{pmatrix}, & & \Psi_{1} = \begin{pmatrix}
0 & \kappa_{3}(\rho) & 0 \\
\widetilde{\kappa}_{6}(\rho) & 0 & \kappa_{3}(\rho)+\frac{\rho}{3} \\
0 & \widehat{\kappa}_{6}(\rho) & 0
\end{pmatrix}, \label{def of Psi0 and Psi1} \\
& \kappa_{3}(\rho) = \frac{\rho^{3}}{54}-\frac{\rho}{6}, & & \kappa_{6}(\rho) = \frac{\rho^{6}}{5832} - \frac{\rho^{4}}{162} - \frac{\rho^{2}}{72} + \frac{7}{32}, \nonumber \\
& \widetilde{\kappa}_{6}(\rho) = \kappa_{6}(\rho) + \frac{\rho}{3}\kappa_{3}(\rho) - \frac{1}{3}, & & \widehat{\kappa}_{6}(\rho) = \kappa_{6}(\rho)-\kappa_{3}(\rho)^{2}+\frac{\rho^{2}}{9}-\frac{1}{3}, \nonumber \\
& L_{+} = \begin{pmatrix}
-\omega & \omega^{2} & 1 \\
-1 & 1 & 1 \\
-\omega^{2} & \omega & 1
\end{pmatrix}, & & L_{-} = \begin{pmatrix}
\omega^{2} & \omega & 1 \\
1 & 1 & 1 \\
\omega & \omega^{2} & 1
\end{pmatrix}, \nonumber \\
& \Theta(z) = \begin{cases}
\diag(\theta_{1}(z),\theta_{2}(z),\theta_{3}(z)), & \im z > 0, \\
\diag(\theta_{2}(z),\theta_{1}(z),\theta_{3}(z)), & \im z < 0,
\end{cases} & & \theta_{k}(z) = \frac{3}{4}\omega^{2k}z^{\frac{4}{3}}+\frac{\rho}{2}\omega^{k}z^{\frac{2}{3}}, \quad k=1,2,3, \label{def of theta}
\end{align}
and $\omega = e^{\frac{2\pi i}{3}}$.
\item[(d)] $\Psi(z)$ remains bounded as $z \to 0$.
\end{itemize}
Consider the following functions
\begin{align}\label{def of mathcal Pj}
\mathcal{P}_{j}(z) = \int_{\Gamma_{j}}e^{-\frac{1}{4}t^{4}-\frac{\rho}{2}t^{2}+itz}dt, \qquad j=0,1,\ldots,5,
\end{align}
where
\begin{align*}
& \Gamma_{0}=(-\infty,+\infty), & & \Gamma_{1} = (i\infty,0]\cup[0,\infty), & & \Gamma_{2} = (i\infty,0]\cup[0,-\infty), \\
& \Gamma_{3} = (-i\infty,0]\cup [0,-\infty), & & \Gamma_{4} = (-i\infty,0]\cup [0,+\infty), & & \Gamma_{5} = (-i\infty,i\infty),
\end{align*}
and define
\begin{align}\label{def of Psi tilde}
\widetilde{\Psi}(z) = \begin{pmatrix}
\mathcal{P}_{0}(z) & \mathcal{P}_{1}(z) & \mathcal{P}_{4}(z) \\
\mathcal{P}_{0}'(z) & \mathcal{P}_{1}'(z) & \mathcal{P}_{4}'(z) \\
\mathcal{P}_{0}''(z) & \mathcal{P}_{1}''(z) & \mathcal{P}_{4}''(z)
\end{pmatrix}, \qquad z \in \mathbb{C}.
\end{align}
It was shown in \cite[Section 8.1]{BleherKuijlaarsIII} that the RH problem for $\Psi$ admits a unique solution which can be explicitly written in terms of $\mathcal{P}_{j}$, $j=0,\ldots,5$. For example, for $\arg z \in (\frac{\pi}{4},\frac{3\pi}{4})$, we have $\Psi(z) = \widetilde{\Psi}(z)$. The explicit expression of $\Psi$ in the other sectors is not needed for us so we do not write it down, but we refer the interested reader to \cite[equations (8.12)--(8.17)]{BleherKuijlaarsIII}. The Pearcey kernel can be written as follows (see \cite[equation (10.19)]{BleherKuijlaarsIII}):
\begin{align}\label{Pearcey kernel in terms of widetilde Psi}
K_{\rho}^{\mathrm{Pe}}(x,y) = \frac{1}{2\pi i(x-y)}\begin{pmatrix}
0 & 1 & 1
\end{pmatrix}\widetilde{\Psi}(y)^{-1}\widetilde{\Psi}(x) \begin{pmatrix}
1 & 0 & 0
\end{pmatrix}^{t}, \qquad x,y \in \mathbb{R},
\end{align}
where $(\cdot)^t$ denotes the transpose operation. Let $\widetilde{K}^{\mathrm{Pe}}_{\rho}$ be the kernel of the operator $\widetilde{\mathcal{K}}^{\mathrm{Pe}}_{\rho}$ appearing in \eqref{fredholm det piecewise constant thinning}:
\begin{align*}
\widetilde{K}^{\mathrm{Pe}}_{\rho}(x,y):=\sum_{j=1}^{m}(1-s_{j})K_{\rho}^{\mathrm{Pe}}(x,y) \chi_{rA_{j}}(y),
\end{align*}
where $\chi_{A_{j}}$ denotes the characteristic function of $A_{j}$, i.e. $\chi_{A_{j}}(x)=1$ if $x \in A_{j}$ and $0$ otherwise. From \eqref{Pearcey kernel in terms of widetilde Psi}, it is easy to see that $\widetilde{K}^{\mathrm{Pe}}_{\rho}$ can be written as
\begin{align}\label{integrable kernel}
\widetilde{K}^{\mathrm{Pe}}_{\rho}(x,y) = \frac{\mathbf{f}(x)^{t} \mathbf{h}(y)}{x-y},
\end{align}
where
\begin{align}\label{def of f and h}
\mathbf{f}(x) = \widetilde{\Psi}(x) \begin{pmatrix}
1 \\ 0 \\ 0
\end{pmatrix}, \quad \mathbf{h}(y) = \frac{\sum_{j=1}^{m}(1-s_{j})\chi_{rA_{j}}(y)}{2\pi i}\widetilde{\Psi}(y)^{-t} \begin{pmatrix}
0 \\ 1 \\ 1
\end{pmatrix}= \sum_{j=1}^{m}\mathfrak{s}_{j}\chi_{rB_{j}}(y)\widetilde{\Psi}(y)^{-t} \begin{pmatrix}
0 \\ 1 \\ 1
\end{pmatrix}.
\end{align}
with $\mathfrak{s}_{j}:=\frac{s_{j+1}-s_{j}}{2\pi i}$ and $B_{j}=(-x_{j},x_{j})$, $j=1,\ldots,m$. Since $\vec{u}\in \mathbb{R}^{m}$, it follows from \eqref{F as expectation intro} that $F(r\vec{x},\vec{u})\in (0,+\infty)$. Thus, by \eqref{fredholm det piecewise constant thinning}, we have $\det (1-\widetilde{\mathcal{K}}^{\mathrm{Pe}}_{\rho})>0$ and in particular $1-\widetilde{\mathcal{K}}^{\mathrm{Pe}}_{\rho}$ is invertible. Using now standard identities for trace-class operators, we obtain
\begin{align}
&  \partial_{r} \log F(r\vec{x},\vec{u}) = \partial_{r} \log \det (1-\widetilde{\mathcal{K}}^{\mathrm{Pe}}_{\rho}) = -\mbox{Tr} \Big( (1-\widetilde{\mathcal{K}}^{\mathrm{Pe}}_{\rho})^{-1} \partial_{r}\widetilde{\mathcal{K}}^{\mathrm{Pe}}_{\rho} \Big) \nonumber \\
& = - \sum_{j=1}^{m} x_{j} \bigg( \lim_{v \nearrow rx_{j}}\mathrm{R}(v,v)+\lim_{v \searrow -rx_{j}}\mathrm{R}(v,v) \bigg) + \sum_{j=1}^{m-1} x_{j} \bigg( \lim_{v \searrow rx_{j}}\mathrm{R}(v,v)+\lim_{v \nearrow -rx_{j}}\mathrm{R}(v,v) \bigg), \label{lol10}
\end{align}
where $\mathrm{R}$ is the kernel of the resolvent operator $(1-\widetilde{\mathcal{K}}^{\mathrm{Pe}}_{\rho})^{-1} \widetilde{\mathcal{K}}^{\mathrm{Pe}}_{\rho}$. Formula \eqref{integrable kernel} shows in particular that $\widetilde{K}^{\mathrm{Pe}}_{\rho}$ is integrable (of size 3) in the sense of \cite{IIKS}. Hence, by \cite[Lemma 2.12]{DIZ}, we have
\begin{align}\label{resolvent in terms of F and H}
\mathrm{R}(u,v) = \frac{\mathbf{F}(u)^{t}\mathbf{H}(v)}{u-v}, \qquad u,v  \in \mathbb{R}, 
\end{align}
where
\begin{align}\label{def of F and H}
\mathbf{F}(u) = (1-\widetilde{\mathcal{K}}^{\mathrm{Pe}}_{\rho})^{-1}\mathbf{f}(u) = Y_{+}(u)\mathbf{f}(u), \qquad \mathbf{H}(v) = Y_{+}^{-t}(v)\mathbf{h}(v), \qquad u,v  \in \mathbb{R}, 
\end{align}
and $Y$ is given by
\begin{align}\label{def of Y}
Y(z) = I - \int_{-rx_{m}}^{rx_{m}}\frac{\mathbf{F}(w)\mathbf{h}(w)^{t}}{w-z}dw.
\end{align}
Furthermore, $Y$ is the unique solution to the following RH problem.
\subsubsection*{RH problem for $Y$}
\begin{itemize}
\item[(a)] $Y : \mathbb{C}\setminus [-rx_{m},rx_{m}] \to \mathbb{C}^{3\times 3}$ is analytic.
\item[(b)] $Y$ satisfies the jumps
\begin{align}\label{jumps of Y}
Y_{+}(x) = Y_{-}(x) (I-2\pi i \mathbf{f}(x)\mathbf{h}(x)^{t}), \qquad x \in (-rx_{m},rx_{m})\setminus \cup_{j=1}^{m-1} \{-rx_{j},rx_{j}\}.
\end{align}
\item[(c)] As $z \to \infty$, $Y(z) = I + \frac{Y_{1}}{z} + \bigO(z^{-2}).$
\item[(d)] As $z \to z_{*}\in \cup_{j=1}^{m-1}\{-rx_{j},rx_{j}\}$, we have $Y(z) = \bigO(\log (z-z_{*}))$.
\end{itemize}
\begin{figure}
\centering
\begin{tikzpicture}
\node at (3,0) {};
\draw (0,0) -- (11,0);
\draw (3,0) -- ($(3,0)+(135:3)$);
\draw (3,0) -- ($(3,0)+(-135:3)$);
\draw (8,0) -- ($(8,0)+(45:3)$);
\draw (8,0) -- ($(8,0)+(-45:3)$);

\draw[fill] (3,0) circle (0.05);
\draw[fill] (8,0) circle (0.05);

\node at (3,-0.3) {$-rx_{m}$};
\node at (8,-0.3) {$rx_{m}$};

\node at ($(8,0)+(22.5:2.5)$) {$\mathrm{I}$};
\node at ($(8,0)+(-22.5:2.5)$) {$\mathrm{VI}$};
\node at ($(5.5,0)+(90:1.5)$) {$\mathrm{II}$};
\node at ($(5.5,0)+(-90:1.5)$) {$\mathrm{V}$};
\node at ($(3,0)+(180-22.5:2.5)$) {$\mathrm{III}$};
\node at ($(3,0)+(-180+22.5:2.5)$) {$\mathrm{IV}$};

\draw[black,arrows={-Triangle[length=0.18cm,width=0.12cm]}]
($(3,0)+(135:1.5)$) --  ++(-45:0.001);
\draw[black,arrows={-Triangle[length=0.18cm,width=0.12cm]}]
($(3,0)+(-135:1.5)$) --  ++(45:0.001);
\draw[black,arrows={-Triangle[length=0.18cm,width=0.12cm]}]
(1.5,0) --  ++(0:0.001);

\draw[black,arrows={-Triangle[length=0.18cm,width=0.12cm]}]
(5.5,0) --  ++(0:0.001);

\draw[black,arrows={-Triangle[length=0.18cm,width=0.12cm]}]
($(8,0)+(45:1.5)$) --  ++(45:0.001);
\draw[black,arrows={-Triangle[length=0.18cm,width=0.12cm]}]
($(8,0)+(-45:1.5)$) --  ++(-45:0.001);
\draw[black,arrows={-Triangle[length=0.18cm,width=0.12cm]}]
(9.5,0) --  ++(0:0.001);

\draw[black,dashed,line width=0.15 mm] ([shift=(0:1cm)]8,0) arc (0:45:1cm);
\node at ($(8,0)+(22.5:1.25)$) {$\frac{\pi}{4}$};
\end{tikzpicture}
\caption{Jump contours $\Sigma_{k}^{(r)}$, $k=0,1,\ldots,6$.}
\label{fig:contour for Phi Sine}
\end{figure}
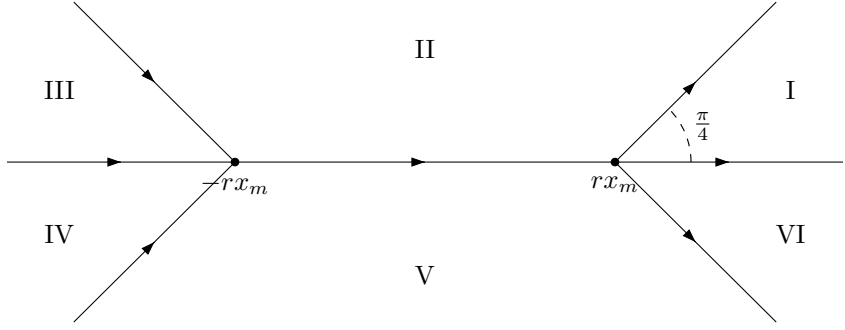
Now, we apply a transformation which changes $Y$ into another function $\Phi$ whose jump matrices are piecewise constant. Following \cite[eq (3.24)]{DXZ2020 thinning}, we define $\Phi(z) = \Phi(z;r)$ as 
\begin{align}\label{Phi in terms of Y and Psi}
\Phi(z) = \frac{\Psi_{0}^{-1}}{\sqrt{\frac{2\pi}{3}}e^{\frac{\rho^{2}}{6}}i} \begin{cases}
Y(z)\Psi(z), & z \in \mathrm{I} \cup \mathrm{III} \cup \mathrm{IV} \cup \mathrm{VI}, \\
Y(z)\widetilde{\Psi}(z), & z \in \mathrm{II}, \\
Y(z) \widetilde{\Psi}(z) \begin{pmatrix}
1 & -1 & -1 \\
0 & 1 & 0 \\
0 & 0 & 1
\end{pmatrix}, & z \in \mathrm{V},
\end{cases}
\end{align}
where the regions $\mathrm{I}$, $\mathrm{II}$ $\mathrm{III}$, $\mathrm{IV}$, $\mathrm{V}$ and $\mathrm{VI}$ are shown in Figure \ref{fig:contour for Phi Sine}. It is easily verified from the RH problems for $Y$ and $\Psi$ that $\Phi$ satisfies the following RH problem.
\subsubsection*{RH problem for $\Phi$}
\begin{itemize}
\item[(a)] $\Phi: \mathbb{C}\setminus \{\cup_{j=0}^{6}\Sigma_{j}^{(r)}\cup\{-rx_{m},rx_{m}\}\}\to \mathbb{C}^{3\times 3}$ is analytic, where
\begin{align}
& \Sigma_{0}^{(r)} = (rx_{m},+\infty), & & \Sigma_{1}^{(r)} = rx_{m}+e^{\frac{\pi i}{4}}(0,+\infty), & & \Sigma_{2}^{(r)} = -rx_{m}+e^{\frac{3\pi i}{4}}(+\infty,0), \nonumber \\
& \Sigma_{3}^{(r)} = (-\infty,-rx_{m}), & & \Sigma_{4}^{(r)} = -rx_{m}+e^{-\frac{3\pi i}{4}}(+\infty,0), & & \Sigma_{5}^{(r)} = rx_{m}+e^{-\frac{\pi i}{4}}(0,+\infty), \label{contour of Phi}
\end{align}
and $\Sigma_{6}^{(r)}=(-rx_{m},rx_{m})$, see also Figure \ref{fig:contour for Phi Sine}.
\item[(b)] For $z \in \Sigma_{j}^{(r)}$, we have $\Phi_{+}(z) = \Phi_{-}(z)J_{j}$, $j=0,\ldots,5$, where $J_{0},J_{1},J_{2},J_{3},J_{4}$ and $J_{5}$ are given by \eqref{def of Jj}. For $z \in (-rx_{m},rx_{m})\setminus \cup_{j=1}^{m-1}\{-rx_{j},rx_{j}\}$, we have $\Phi_{+}(z) = \Phi_{-}(z)J_{6}(z)$, where
\begin{align*}
J_{6}(z)=\begin{pmatrix}
1 & s_{j} & s_{j} \\
0 & 1 & 0 \\
0 & 0 & 1
\end{pmatrix}, \qquad z \in rA_{j}, \qquad j=1,\ldots,m.
\end{align*}
\item[(c)] As $z \to \infty$, $\pm \im z > 0$, 
\begin{align}\label{asymp of Phi at infty}
\Phi(z) = \bigg( I+\frac{\Phi_{1}}{z}+\frac{\Phi_{2}}{z^{2}} +\bigO(z^{-3}) \bigg) \diag\Big( z^{-\frac{1}{3}},1,z^{\frac{1}{3}} \Big)L_{\pm}e^{\Theta(z)},
\end{align}
where $\Phi_{1}$,$\Phi_{2}$ are independent of $z$ and $\Phi_{1}=\Psi_{1}+\Psi_{0}^{-1}Y_{1}\Psi_{0}$.
\item[(d)] As $z \to rx_{j}$, $j=1,\ldots,m$, we have
\begin{align}\label{asymp of Phi near rxj 1}
\hspace{-0.35cm} \Phi(z) = \widehat{\Phi}_{j}(z) \begin{pmatrix}
1 & - \mathfrak{s}_{j}\log(z-rx_{j}) & - \mathfrak{s}_{j}\log(z-rx_{j}) \\
0 & 1 & 0 \\
0 & 0 & 1
\end{pmatrix} \begin{cases} 
I, & \hspace{-0.15cm} z \in \mathrm{II}, \\
\begin{pmatrix}
1 & -s_{j+1} & -s_{j+1} \\
0 & 1 & 0 \\
0 & 0 & 1
\end{pmatrix}, & \hspace{-0.15cm} z \in \mathrm{V},
\end{cases}
\end{align}
where we recall that $\mathfrak{s}_{j}:=\frac{s_{j+1}-s_{j}}{2\pi i}$. The matrix $\widehat{\Phi}_{j}$ is analytic at $rx_{j}$ and satisfies
\begin{align}\label{asymp of Phi near rxj}
\widehat{\Phi}_{j}(z) = \Phi_{j}^{(0)}(r) \Big( I + \Phi_{j}^{(1)}(r)(z-rx_{j}) + \bigO((z-rx_{j})^{2}) \Big), \qquad z \to rx_{j},
\end{align}
for some matrices $\Phi_{j}^{(0)}(r)$ and $\Phi_{j}^{(1)}(r)$. 
\item[(e)] $\Phi$ satisfies the symmetry 
\begin{align}\label{symmetry of Phi}
\Phi(z) = -\diag(1,-1,1)\Phi(-z)\mathcal{B}, \qquad \mathcal{B} = \begin{pmatrix}
-1 & 0 & 0 \\
0 & 0 & 1 \\
0 & 1 & 0
\end{pmatrix}.
\end{align}
\end{itemize}

\begin{proposition}
We have
\begin{align}\label{der of log F diff identity}
\partial_{r} \log F(r\vec{x},\vec{u}) = -\sum_{j=1}^{m} 2\, \mathfrak{s}_{j} x_{j} \Big[ \Big( \Phi_{j}^{(1)}(r)  \Big)_{21} + \Big( \Phi_{j}^{(1)}(r) \Big)_{31} \Big].
\end{align}
\end{proposition}
\begin{proof}
The proof is a minor adaptation of \cite[Proposition 3.5]{DXZ2020 thinning}. For $v \in \mathbb{R}$, by \eqref{def of f and h}, \eqref{def of F and H} and \eqref{Phi in terms of Y and Psi}, we have
\begin{align}\label{F and H in terms of Phi}
\mathbf{F}(v) = \sqrt{\frac{2\pi}{3}}e^{\frac{\rho^{2}}{6}}i \Psi_{0}\Phi_{+}(v) \begin{pmatrix}
1 \\ 0 \\ 0
\end{pmatrix}, \qquad \mathbf{H}(v) = \sum_{j=1}^{m}\mathfrak{s}_{j}\chi_{rB_{j}}(v) \frac{\Psi_{0}^{-t}}{\sqrt{\frac{2\pi}{3}}e^{\frac{\rho^{2}}{6}}i}\Phi_{+}(v)^{-t} \begin{pmatrix}
0 \\ 1 \\ 1
\end{pmatrix}.
\end{align}
Using \eqref{resolvent in terms of F and H}, \eqref{Phi in terms of Y and Psi}, \eqref{symmetry of Phi} and \eqref{F and H in terms of Phi}, we find
\begin{align*}
\mathrm{R}(v,v) = \sum_{j=1}^{m}\mathfrak{s}_{j}\chi_{rB_{j}}(v) \bigg( \big[ \Phi_{+}(v)^{-1}\Phi_{+}'(v) \big]_{21} + \big[ \Phi_{+}(v)^{-1}\Phi_{+}'(v) \big]_{31} \bigg) = \mathrm{R}(-v,-v), \qquad v \in \mathbb{R},
\end{align*}
which allows us to rewrite \eqref{lol10} as
\begin{align*}
\partial_{r} \log F(r\vec{x},\vec{u}) & = - 2 \sum_{j=1}^{m} x_{j} \lim_{v \nearrow rx_{j}}\mathrm{R}(v,v) + 2\sum_{j=1}^{m-1} x_{j} \lim_{v \searrow rx_{j}}\mathrm{R}(v,v) \\
& = - 2 \sum_{j=1}^{m} x_{j} \mathfrak{s}_{j} \bigg( \big[ \Phi_{+}(x_{j})^{-1}\Phi_{+}'(x_{j}) \big]_{21} + \big[ \Phi_{+}(x_{j})^{-1}\Phi_{+}'(x_{j}) \big]_{31} \bigg).
\end{align*}
By \eqref{asymp of Phi near rxj 1} and \eqref{asymp of Phi near rxj}, $\big[ \Phi_{+}(x_{j})^{-1}\Phi_{+}'(x_{j}) \big]_{k1} = ( \Phi_{j}^{(1)}(r)  )_{k1}$, for all $j=1,\ldots,m$ and $k=2,3$, which finishes the proof.
\end{proof}

\section{Lax pair}\label{section: Lax pair}
In this section,
\begin{itemize}
\vspace{-0.2cm}\item we find an explicit solution to \eqref{big system of ODEs}--\eqref{sum relation between qj and pj} in terms of $\Phi$,
\vspace{-0.2cm}\item we prove the relation $\partial_{r} \log F(r\vec{x},\vec{u}) = 2 H(r)$,
\vspace{-0.2cm}\item we derive some further identities for $H$ which will be useful in Section \ref{section: proof of main results}.
\end{itemize}
\vspace{-0.2cm}These results generalize part of the content of \cite[Section 4]{DXZ2020 thinning} to an arbitrary $m$.
\begin{proposition}\label{prop:solution of the system}
The functions $(p_{0},q_{0},\{p_{j,1},q_{j,1},p_{j,2},q_{j,2},p_{j,3},q_{j,3}\}_{j=1}^{m})$ defined by
\begin{align}
& p_{0}(r) := \frac{1}{\sqrt{2}}\Big( \frac{\rho}{3}+\Phi_{1,23} \Big), & & q_{0}(r) := \frac{1}{\sqrt{2}}\Big( \frac{\rho}{3} - \Phi_{1,12} \Big), \label{def of p0 and q0 in proof} \\
& \begin{pmatrix}
p_{j,1}(r) \\ p_{j,2}(r) \\ p_{j,3}(r)
\end{pmatrix} := - \mathfrak{s}_{j}\Phi_{j}^{(0)}(r)^{-t}\begin{pmatrix}
0 \\ 1 \\ 1
\end{pmatrix}, & & \begin{pmatrix}
q_{j,1}(r) \\ q_{j,2}(r) \\ q_{j,3}(r)
\end{pmatrix} := \Phi_{j}^{(0)}(r)\begin{pmatrix}
1 \\ 0 \\ 0
\end{pmatrix}, & & j=1,\ldots,m, \label{def of qj and pj in terms of Phijp0p}
\end{align}
satisfy \eqref{sum relation between qj and pj} and the system of coupled equations \eqref{big system of ODEs}.
\end{proposition}
\begin{remark}
Since $\Phi$ exists by \eqref{def of F and H}, \eqref{def of Y} and \eqref{Phi in terms of Y and Psi}, Proposition \ref{prop:solution of the system} implies that there exists at least one solution to \eqref{big system of ODEs}--\eqref{sum relation between qj and pj}.
\end{remark} 
\begin{proof}
Following \cite{DXZ2020 thinning}, we will proceed by analyzing the following Lax pair
\begin{align*}
L(z;r) := \partial_{z}\Phi(z;r) \cdot \Phi(z;r)^{-1}, \qquad U(z;r) := \partial_{r}\Phi(z;r) \cdot \Phi(z;r)^{-1}.
\end{align*}
Since the jump matrices of $\Phi$ are independent of $z$ and $r$, $L(z)=L(z;r)$ and $U(z)=U(z;r)$ are analytic in $\mathbb{C}\setminus \{-rx_{m},\ldots,-rx_{1},$ $0,rx_{1},\ldots,rx_{m}\}$. Furthermore, using \eqref{symmetry of Phi}, we infer that they satisfy
\begin{align}\label{symmetry of L and U}
L(z) = -\diag(1,-1,1)L(-z)\diag(1,-1,1), \qquad U(z) = \diag(1,-1,1)U(-z)\diag(1,-1,1).
\end{align}
By \eqref{asymp of Phi at infty}, \eqref{def of p0 and q0 in proof} and \eqref{symmetry of L and U}, we find
\begin{align}\label{expansion of L in proof}
L(z) = \begin{pmatrix}
0 & 0 & 0 \\
0 & 0 & 0 \\
1 & 0 & 0 
\end{pmatrix}z + A_{0}(r) + \frac{L_{1}}{z} + \bigO(z^{-2}), \qquad \mbox{as } z \to \infty,
\end{align}
where
\begin{align}
& A_{0}(r) = \begin{pmatrix}
0 & 1 & 0 \\
\frac{\rho}{3} & 0 & 1 \\
0 & \frac{\rho}{3} & 0
\end{pmatrix} + \left[ \Phi_{1},\begin{pmatrix}
0 & 0 & 0 \\ 0 & 0 & 0 \\ 1 & 0 & 0
\end{pmatrix} \right] = \begin{pmatrix}
0 & 1 & 0 \\
\sqrt{2}p_{0}(r) & 0 & 1 \\
0 & \sqrt{2}q_{0}(r) & 0
\end{pmatrix}, \label{final expr for A0} \\
& L_{1} = \begin{pmatrix}
-\frac{1}{3} & 0 & \frac{\rho}{3} \\
0 & 0 & 0 \\
0 & 0 & \frac{1}{3}
\end{pmatrix} + \left[ \Phi_{2},\begin{pmatrix}
0 & 0 & 0 \\
0 & 0 & 0 \\
1 & 0 & 0 
\end{pmatrix} \right] + \left[ \begin{pmatrix}
0 & 0 & 0 \\
0 & 0 & 0 \\
1 & 0 & 0 
\end{pmatrix}\Phi_{1},\Phi_{1} \right] + \left[ \Phi_{1},\begin{pmatrix}
0 & 1 & 0 \\
\frac{\rho}{3} & 0 & 1 \\
0 & \frac{\rho}{3} & 0
\end{pmatrix} \right], \nonumber
\end{align}
and where we have used the notation $[\mathcal{B}_{1},\mathcal{B}_{2}]:=\mathcal{B}_{1}\mathcal{B}_{2}-\mathcal{B}_{2}\mathcal{B}_{1}$.
Also, by \eqref{asymp of Phi near rxj 1}--\eqref{asymp of Phi near rxj} and \eqref{def of qj and pj in terms of Phijp0p}, we have
\begin{align}\label{expansion of L near rxj}
L(z) = \frac{A_{j}(r)}{z-rx_{j}} + \bigO(1), \qquad \mbox{as } z \to rx_{j}, \quad j=1,\ldots,m
\end{align}
with
\begin{align}\label{lol1}
A_{j}(r) = -\mathfrak{s}_{j}\Phi_{j}^{(0)}(r)\begin{pmatrix}
0 & 1 & 1 \\
0 & 0 & 0 \\
0 & 0 & 0
\end{pmatrix}\Phi_{j}^{(0)}(r)^{-1} = \begin{pmatrix}
q_{j,1} \\ q_{j,2} \\ q_{j,3}
\end{pmatrix}\begin{pmatrix}
p_{j,1} & p_{j,2} & p_{j,3}
\end{pmatrix}.
\end{align}
Since $\det \Phi(z)$ is constant, we have $\mbox{Tr} L(z) = \mbox{Tr} A_{j}(r) = \sum_{k=1}^{3} p_{j,k}(r)q_{j,k}(r) = 0$, which already proves \eqref{sum relation between qj and pj}. Combining \eqref{symmetry of L and U}, \eqref{expansion of L in proof} and \eqref{expansion of L near rxj}, we have shown that
\begin{align}\label{def of L}
L(z) = \begin{pmatrix}
0 & 0 & 0 \\
0 & 0 & 0 \\
z & 0 & 0 
\end{pmatrix}+A_{0}(r) + \sum_{j=1}^{m}\bigg( \frac{A_{j}(r)}{z-rx_{j}} + \frac{A_{-j}(r)}{z+rx_{j}}\bigg),
\end{align}
where
\begin{align}\label{def of Amj}
A_{-j}(r) = \diag(1,-1,1)A_{j}(r)\diag(1,-1,1) = \begin{pmatrix}
q_{j,1} \\ -q_{j,2} \\ q_{j,3}
\end{pmatrix}\begin{pmatrix}
p_{j,1} & -p_{j,2} & p_{j,3}
\end{pmatrix}.
\end{align}
For the computation of $U$, we use \eqref{asymp of Phi at infty} and \eqref{asymp of Phi near rxj 1}--\eqref{asymp of Phi near rxj} to obtain
\begin{align*}
& U(z) = \bigO(z^{-1}), \qquad \mbox{as } z \to \infty, & & U(z) = -x_{j} \frac{A_{j}(r)}{z-rx_{j}} + \bigO(1), \qquad \mbox{as } z \to rx_{j}.
\end{align*}
Using also \eqref{symmetry of L and U}, we conclude that
\begin{align}\label{expr for U}
U(z) = \sum_{j=1}^{m}\bigg( -x_{j}\frac{A_{j}(r)}{z-rx_{j}} + x_{j}\frac{A_{-j}(r)}{z+rx_{j}}\bigg).
\end{align}
It remains to show that the functions $(p_{0},q_{0},\{p_{j,1},q_{j,1},p_{j,2},q_{j,2},p_{j,3},q_{j,3}\}_{j=1}^{m})$ satisfy the system of equations \eqref{big system of ODEs}. For this, we note that the compatibility condition $\partial_{z}\partial_{r}\Phi(z) = \partial_{r}\partial_{z}\Phi(z)$ is equivalent to the relation
\begin{align}\label{lol2}
\partial_{r}L(z) - \partial_{z}U(z) = [U(z),L(z)].
\end{align}
On the other hand, by \eqref{def of L} and \eqref{expr for U}, we have
\begin{align}\label{lol3}
\partial_{r}L(z) - \partial_{z}U(z) = A_{0}'(r) + \sum_{j=1}^{m} \bigg(\frac{A_{j}'(r)}{z-rx_{j}} + \frac{A_{-j}'(r)}{z+rx_{j}}\bigg).
\end{align}
Substituting \eqref{def of L} and \eqref{expr for U} in the above two equations, and then taking $z \to \infty$, we get
\begin{align*}
A_{0}'(r) & = \sum_{j=1}^{m}x_{j} \left[ A_{-j}(r)-A_{j}(r),\begin{pmatrix}
0 & 0 & 0 \\
0 & 0 & 0 \\
1 & 0 & 0 
\end{pmatrix} \right]  = \begin{pmatrix}
0 & 0 & 0 \\
-2\sum_{j=1}^{m}x_{j}p_{j,3}(r)q_{j,2}(r) & 0 & 0 \\
0 & 2\sum_{j=1}^{m}x_{j}p_{j,2}(r)q_{j,1}(r) & 0
\end{pmatrix},
\end{align*}
which yields the first two equations in \eqref{big system of ODEs}. We now prove the last six equations of \eqref{big system of ODEs}. A direct computation using \eqref{asymp of Phi near rxj 1} and \eqref{asymp of Phi near rxj} shows that
\begin{align}\label{lol11}
x_{j}L(z)+U(z) = \big( x_{j}\partial_{z}\Phi(z) + \partial_{r} \Phi(z) \big) \Phi(z)^{-1} = \partial_{r} \Phi_{j}^{(0)}(r) \cdot \Phi_{j}^{(0)}(r)^{-1} + o(1) \qquad \mbox{as } z \to rx_{j},
\end{align}
and using \eqref{def of L} and \eqref{expr for U}, we get
\begin{align}\label{lol12}
x_{j}L(z)+U(z) = M_{j}(r) - \frac{1}{r}A_{j}(r) + o(1) \qquad \mbox{as } z \to rx_{j},
\end{align}
with
\begin{align}\label{def of Mj}
M_{j} = \begin{pmatrix}
0 & 0 & 0 \\
0 & 0 & 0 \\
rx_{j}^{2} & 0 & 0
\end{pmatrix} + x_{j} A_{0} + \frac{1}{r} \sum_{\substack{k=-m \\ k \neq 0}}^{m} A_{k} = \begin{pmatrix}
\frac{2}{r}S_{11} & x_{j} & \frac{2}{r}S_{31} \\
\sqrt{2}p_{0}x_{j} & \frac{2}{r}S_{22} & x_{j} \\
rx_{j}^{2} + \frac{2}{r}S_{13} & \sqrt{2}q_{0}x_{j} & \frac{2}{r}S_{33}
\end{pmatrix}.
\end{align}
In \eqref{def of Mj}, we have omitted the $r$-dependence of various functions for notational convenience. Combining \eqref{lol11} and \eqref{lol12} yields
\begin{align*}
\partial_{r} \Phi_{j}^{(0)}(r) = \bigg( M_{j}(r) - \frac{1}{r}A_{j}(r) \bigg) \Phi_{j}^{(0)}(r).
\end{align*}
Taking the first column of the above equation and using \eqref{sum relation between qj and pj}, \eqref{final expr for A0} and \eqref{def of qj and pj in terms of Phijp0p}, we get
\begin{align}\label{lol4}
\begin{pmatrix}
q_{j,1}'(r) & q_{j,2}'(r) & q_{j,3}'(r)
\end{pmatrix}^{t} = M(r) \begin{pmatrix}
q_{j,1}(r) & q_{j,2}(r) & q_{j,3} (r)
\end{pmatrix}^{t},
\end{align}
and it is a direct computation to verify that \eqref{lol4} is equivalent to the third, fourth and fifth equations of \eqref{big system of ODEs}. Finally, using \eqref{def of L}, \eqref{expr for U}, \eqref{lol2} and \eqref{lol3}, and letting $z \to rx_{j}$,
we get
\begin{align}\label{Aj' commutator}
A_{j}'(r) = -[A_{j}(r),M_{j}(r)].
\end{align}
Combining \eqref{Aj' commutator} with \eqref{final expr for A0} and \eqref{lol4}, we get
\begin{align*}
\begin{pmatrix}
p_{j,1}'(r) & p_{j,2}'(r) & p_{j,3}'(r)
\end{pmatrix} = -\begin{pmatrix}
p_{j,1}(r) & p_{j,2}(r) & p_{j,3}(r)
\end{pmatrix} M(r),
\end{align*}
which yields the last three equations in \eqref{big system of ODEs}.
\end{proof}

\medskip For later use, we also note that by taking $z \to \infty$ in \eqref{expansion of L in proof} and then by reading the $z^{-1}$ term of the (1,3) entry, we get
\begin{align}\label{useful relation}
\sum_{j=1}^{m} (A_{j,13}(r)+A_{-j,13}(r)) = 2S_{31}(r) = \frac{\rho}{3} + \Phi_{1,12}(r)-\Phi_{1,23}(r) = \rho-\sqrt{2}(p_{0}(r)+q_{0}(r)),
\end{align}
where we have also used \eqref{def of L} and \eqref{def of p0 and q0 in proof}.
In the rest of this section, we prove some identities for $H$ which will be useful in Section \ref{section: proof of main results}.
\begin{proposition}
Let $H$ be the Hamiltonian given in \eqref{def of Hamiltonian} with $(p_{0},q_{0},\{p_{j,1},q_{j,1},p_{j,2},q_{j,2},p_{j,3},q_{j,3}\}_{j=1}^{m})$ defined as in \eqref{def of p0 and q0 in proof}--\eqref{def of qj and pj in terms of Phijp0p}. We have
\begin{align}\label{der of integral representation}
\partial_{r} \log F(r\vec{x},\vec{u}) = 2 H(r).
\end{align}
\end{proposition}
\begin{proof}
By \eqref{der of log F diff identity}, the claim \eqref{der of integral representation} is equivalent to
\begin{align}\label{Hamiltonian as a trace}
H(r) = -\sum_{j=1}^{m}\mathfrak{s}_{j}x_{j} \mathrm{Tr}\left( \Phi_{j}^{(1)}(r)\begin{pmatrix}
0 & 1 & 1 \\
0 & 0 & 0 \\
0 & 0 & 0
\end{pmatrix} \right) = -\sum_{j=1}^{m}\mathfrak{s}_{j}x_{j} \bigg[ \Big( \Phi_{j}^{(1)}(r)\Big)_{21} + \Phi_{j}^{(1)}(r)\Big)_{31} \bigg].
\end{align}
Let $\mathcal{H}(r)$ be the right-hand side of \eqref{Hamiltonian as a trace}. We must show that $H(r) = \mathcal{H}(r)$. By reading the $\bigO(1)$ term in the expansion of $\partial_{z}\Phi(z) = L(z)\Phi(z)$ as $z \to rx_{j}$ (using \eqref{def of L} and \eqref{asymp of Phi near rxj 1}), we obtain
\begin{align}
& \Phi_{j}^{(1)}(r) = \mathfrak{s}_{j} \left[ \Phi_{j}^{(1)}(r),\begin{pmatrix}
0 & 1 & 1 \\
0 & 0 & 0 \\
0 & 0 & 0
\end{pmatrix} \right] \nonumber \\
& + \Phi_{j}^{(0)}(r)^{-1}\left( \begin{pmatrix}
0 & 0 & 0 \\
0 & 0 & 0 \\
rx_{j} & 0 & 0 
\end{pmatrix}+A_{0}(r) + \sum_{\substack{\ell=1 \\ \ell \neq j}}^{m}\bigg( \frac{A_{\ell}(r)}{rx_{j}-rx_{\ell}} + \frac{A_{-\ell}(r)}{rx_{j}+rx_{\ell}}\bigg) + \frac{A_{-j}(r)}{2rx_{j}} \right)\Phi_{j}^{(0)}(r). \label{lol13}
\end{align}
Substituting \eqref{lol13} in the right-hand side of \eqref{Hamiltonian as a trace} leads to
\begin{align}
\mathcal{H}(r) = & -\sum_{j=1}^{m}\mathfrak{s}_{j}x_{j} \begin{pmatrix}
0 & 1 & 1
\end{pmatrix} \Phi_{j}^{(0)}(r)^{-1}\left[ \begin{pmatrix}
0 & 0 & 0 \\
0 & 0 & 0 \\
rx_{j} & 0 & 0 
\end{pmatrix}+A_{0}(r) \right. \nonumber \\
& \left. + \sum_{\substack{\ell=1 \\ \ell \neq j}}^{m}\bigg( \frac{A_{\ell}(r)}{rx_{j}-rx_{\ell}} + \frac{A_{-\ell}(r)}{rx_{j}+rx_{\ell}}\bigg) + \frac{A_{-j}(r)}{2rx_{j}} \right]\Phi_{j}^{(0)}(r) \begin{pmatrix}
1 \\ 0 \\ 0
\end{pmatrix}. \label{lol14}
\end{align}
Inserting \eqref{def of qj and pj in terms of Phijp0p} and \eqref{final expr for A0} in \eqref{lol14}, we get
\begin{align*}
\mathcal{H}(r) = \sum_{j=1}^{m}x_{j} \begin{pmatrix}
p_{j,1} \\ p_{j,2} \\ p_{j,3}
\end{pmatrix}^{t}\left( \begin{pmatrix}
0 & 1 & 0 \\
\sqrt{2}p_{0} & 0 & 1 \\
rx_{j} & \sqrt{2}q_{0} & 0
\end{pmatrix} + \sum_{\substack{\ell=1 \\ \ell \neq j}}^{m}\bigg( \frac{A_{\ell}}{rx_{j}-rx_{\ell}} + \frac{A_{-\ell}}{rx_{j}+rx_{\ell}}\bigg) + \frac{A_{-j}}{2rx_{j}} \right)\begin{pmatrix}
q_{j,1} \\ q_{j,2} \\ q_{j,3}
\end{pmatrix}.
\end{align*}
The double sum in the above expression can be simplified as
\begin{align*}
& \sum_{j=1}^{m}x_{j} \begin{pmatrix}
p_{j,1} \\ p_{j,2} \\ p_{j,3}
\end{pmatrix}^{t} \sum_{\substack{\ell=1 \\ \ell \neq j}}^{m}\bigg( \frac{A_{\ell}}{rx_{j}-rx_{\ell}} + \frac{A_{-\ell}}{rx_{j}+rx_{\ell}}\bigg) \begin{pmatrix}
q_{j,1} \\ q_{j,2} \\ q_{j,3}
\end{pmatrix} = \frac{1}{2} \sum_{j=1}^{m} \begin{pmatrix}
p_{j,1} \\ p_{j,2} \\ p_{j,3}
\end{pmatrix}^{t} \sum_{\substack{\ell=1 \\ \ell \neq j}}^{m}\Big( A_{\ell} + A_{-\ell}\bigg) \begin{pmatrix}
q_{j,1} \\ q_{j,2} \\ q_{j,3}
\end{pmatrix}.
\end{align*}
Using \eqref{lol1} and \eqref{def of Amj}, it is now a direct computation to check that indeed $\mathcal{H}(r) = H(r)$, which concludes the proof. 
\end{proof}
\begin{proposition}\label{prop: diff identity for Hamiltonian}
Let $H$ be the Hamiltonian given in \eqref{def of Hamiltonian} with $(p_{0},q_{0},\{p_{j,1},q_{j,1},p_{j,2},q_{j,2},p_{j,3},q_{j,3}\}_{j=1}^{m})$ defined as in \eqref{def of p0 and q0 in proof}--\eqref{def of qj and pj in terms of Phijp0p}. We have
\begin{align}
& p_{0}(r)q_{0}'(r)+\sum_{j=1}^{m}\sum_{k=1}^{3}p_{j,k}(r)q_{j,k}'(r) - H(r) \nonumber \\
& = H(r) + \frac{1}{4}\frac{d}{dr}\Big( 2p_{0}(r)q_{0}(r) + \sum_{j=1}^{m}\big[p_{j,2}(r)q_{j,2}(r)+2p_{j,3}(r)q_{j,3}(r)\big] -3rH(r) \Big). \label{magical identity}
\end{align}
Furthermore,
\begin{align}\label{differential identity for big sum}
\partial_{\gamma}\bigg( p_{0}(r)q_{0}'(r) + \sum_{k=1}^{3}\sum_{j=1}^{m}p_{j,k}(r)q_{j,k}'(r) - H(r) \bigg) = \frac{d}{dr}\bigg( \sum_{k=1}^{3}\sum_{j=1}^{m}p_{j,k}(r)\partial_{\gamma}q_{j,k}(r)+p_{0}(r)\partial_{\gamma} q_{0}(r) \bigg)
\end{align}
where $\gamma$ is any parameter among $u_{1},\ldots,u_{m}$.
\end{proposition}
\begin{proof}
Formula \eqref{magical identity} follows directly from \eqref{big system of ODEs} and \eqref{sum relation between qj and pj}, and formula \eqref{differential identity for big sum} follows from \eqref{Hamiltonian relation} together with
\begin{align*}
\partial_{\gamma}H(r) = \frac{\partial H}{\partial p_{0}}(r)\partial_{\gamma}p_{0}(r) + \frac{\partial H}{\partial q_{0}}(r)\partial_{\gamma}q_{0}(r) + \sum_{j=1}^{m}\sum_{k=1}^{3} \bigg( \frac{\partial H}{\partial p_{j,k}}(r)\partial_{\gamma}p_{j,k}(r) + \frac{\partial H}{\partial q_{j,k}}(r)\partial_{\gamma}q_{j,k}(r) \bigg).
\end{align*}
\end{proof}
\vspace{-0.5cm}\section{Asymptotic analysis of $\Phi(z;r)$ as $r \to + \infty$}\label{section: steepest descent for large r}
In this section, we perform a Deift-Zhou steepest descent analysis to obtain the large $r$ asymptotics of $\Phi$. The case $m=1$ of this analysis was previously done in \cite{DXZ2020 thinning}.
\subsection{First transformation: $\Phi \to T$}
Define 
\begin{align}\label{def of T}
T(z) = \diag\big( r^{\frac{1}{3}},1,r^{-\frac{1}{3}} \big) \Phi(rz;r)e^{-\Theta(rz)}.
\end{align}
The jumps for $T$ on $(-x_{m},x_{m})$ are given by
\begin{align*}
& T_{+}(z) = T_{-}(z) \begin{pmatrix}
e^{\theta_{2}(rz)-\theta_{1}(rz)} & s_{j} & s_{j}e^{\theta_{2}(rz)-\theta_{3}(rz)} \\
0 & e^{\theta_{1}(rz)-\theta_{2}(rz)} & 0 \\
0 & 0 & 1
\end{pmatrix}, & & z \in (x_{j-1},x_{j}),  \\
& T_{+}(z) = T_{-}(z) \begin{pmatrix}
e^{\theta_{3,+}(rz)-\theta_{3,-}(rz)} & s_{j}e^{\theta_{2,-}(rz)-\theta_{2,+}(rz)} & s_{j} \\
0 & 1 & 0 \\
0 & 0 & e^{\theta_{3,-}(rz)-\theta_{3,+}(rz)}
\end{pmatrix}, & & z \in (-x_{j},-x_{j-1}),
\end{align*}
where $j=1,\ldots,m$, $x_{0}:=0$, and where we have used
\begin{align*}
\theta_{3,+}(z) = \theta_{2,-}(z), \quad \theta_{1,+}(z) = \theta_{3,-}(z), \quad \theta_{1,-}(z) = \theta_{2,+}(z), \qquad z<0.
\end{align*}
As $z \to \infty$, $\pm \im z > 0$, we have
\begin{align*}
T(z) = \bigg( I+\frac{T_{1}}{z} +\bigO(z^{-2}) \bigg) \diag\Big( z^{-\frac{1}{3}},1,z^{\frac{1}{3}} \Big)L_{\pm}
\end{align*}
where $T_{1} = \frac{1}{r} \diag\big( r^{\frac{1}{3}},1,r^{-\frac{1}{3}} \big) \Phi_{1} \diag\big( r^{-\frac{1}{3}},1,r^{\frac{1}{3}} \big)$.
\subsection{Second transformation: $T \to S$}

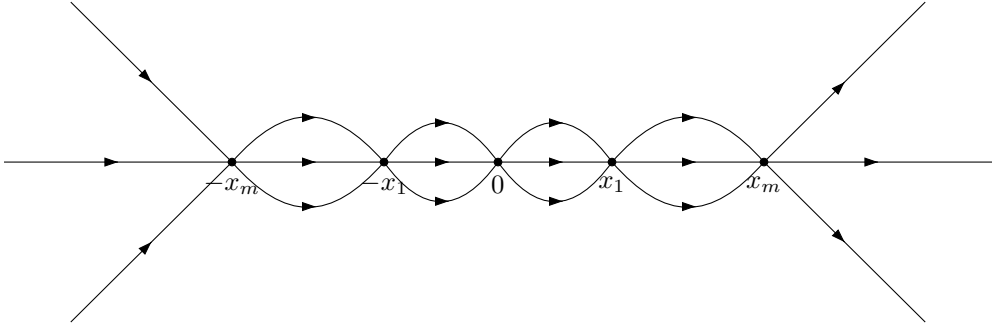
\begin{figure}
\centering
\begin{tikzpicture}
\node at (3,0) {};
\draw (0,0) -- (13,0);
\draw (3,0) -- ($(3,0)+(135:3)$);
\draw (3,0) -- ($(3,0)+(-135:3)$);
\draw (10,0) -- ($(10,0)+(45:3)$);
\draw (10,0) -- ($(10,0)+(-45:3)$);

\draw[fill] (3,0) circle (0.05);
\draw[fill] (5,0) circle (0.05);
\draw[fill] (6.5,0) circle (0.05);
\draw[fill] (8,0) circle (0.05);
\draw[fill] (10,0) circle (0.05);

\node at (3,-0.3) {$-x_{m}$};
\node at (5,-0.3) {$-x_{1}$};
\node at (6.5,-0.3) {$0$};
\node at (8,-0.3) {$x_{1}$};
\node at (10,-0.3) {$x_{m}$};

\draw[black,arrows={-Triangle[length=0.18cm,width=0.12cm]}]
($(3,0)+(135:1.5)$) --  ++(-45:0.001);
\draw[black,arrows={-Triangle[length=0.18cm,width=0.12cm]}]
($(3,0)+(-135:1.5)$) --  ++(45:0.001);
\draw[black,arrows={-Triangle[length=0.18cm,width=0.12cm]}]
(1.5,0) --  ++(0:0.001);

\draw[black,arrows={-Triangle[length=0.18cm,width=0.12cm]}]
(5.85,0) --  ++(0:0.001);
\draw[black,arrows={-Triangle[length=0.18cm,width=0.12cm]}]
(4.1,0) --  ++(0:0.001);
\draw[black,arrows={-Triangle[length=0.18cm,width=0.12cm]}]
(7.35,0) --  ++(0:0.001);
\draw[black,arrows={-Triangle[length=0.18cm,width=0.12cm]}]
(9.1,0) --  ++(0:0.001);

\draw[black,arrows={-Triangle[length=0.18cm,width=0.12cm]}]
($(10,0)+(45:1.5)$) --  ++(45:0.001);
\draw[black,arrows={-Triangle[length=0.18cm,width=0.12cm]}]
($(10,0)+(-45:1.5)$) --  ++(-45:0.001);
\draw[black,arrows={-Triangle[length=0.18cm,width=0.12cm]}]
(11.5,0) --  ++(0:0.001);

\draw (3,0) .. controls (3.7,0.8) and (4.3,0.8) .. (5,0);
\draw (3,0) .. controls (3.7,-0.8) and (4.3,-0.8) .. (5,0);
\draw (5,0) .. controls (5.5,0.7) and (6,0.7) .. (6.5,0);
\draw (5,0) .. controls (5.5,-0.7) and (6,-0.7) .. (6.5,0);
\draw (6.5,0) .. controls (7,0.7) and (7.5,0.7) .. (8,0);
\draw (6.5,0) .. controls (7,-0.7) and (7.5,-0.7) .. (8,0);
\draw (8,0) .. controls (8.7,0.8) and (9.3,0.8) .. (10,0);
\draw (8,0) .. controls (8.7,-0.8) and (9.3,-0.8) .. (10,0);

\draw[black,arrows={-Triangle[length=0.18cm,width=0.12cm]}]
(4.1,0.59) --  ++(0:0.001);
\draw[black,arrows={-Triangle[length=0.18cm,width=0.12cm]}]
(4.1,-0.59) --  ++(0:0.001);
\draw[black,arrows={-Triangle[length=0.18cm,width=0.12cm]}]
(5.85,0.52) --  ++(0:0.001);
\draw[black,arrows={-Triangle[length=0.18cm,width=0.12cm]}]
(5.85,-0.52) --  ++(0:0.001);
\draw[black,arrows={-Triangle[length=0.18cm,width=0.12cm]}]
(7.35,0.52) --  ++(0:0.001);
\draw[black,arrows={-Triangle[length=0.18cm,width=0.12cm]}]
(7.35,-0.52) --  ++(0:0.001);
\draw[black,arrows={-Triangle[length=0.18cm,width=0.12cm]}]
(9.1,0.59) --  ++(0:0.001);
\draw[black,arrows={-Triangle[length=0.18cm,width=0.12cm]}]
(9.1,-0.59) --  ++(0:0.001);
\end{tikzpicture}
\caption{Jump contours for the RH problem for $S$ with $m=2$.}
\label{fig:contour for S s1 neq 0}
\end{figure}

For each $j=1,\ldots,m$, let $\gamma_{j,+}$ and $\gamma_{j,-}$ be open curves, lying in the upper and lower half planes respectively, starting at $x_{j-1}$ and ending at $x_{j}$. We also orient the open curves $\gamma_{-j,+}:=-\gamma_{j,-}$ and $\gamma_{-j,-}:=-\gamma_{j,+}$ from $-x_{j}$ to $-x_{j-1}$. Let
\begin{align*}
\gamma_{m+1,+}:=\Sigma_{1}^{(1)}, \quad \gamma_{m+1,-}:=\Sigma_{5}^{(1)}, \quad \gamma_{-m-1,+}:=\Sigma_{2}^{(1)}, \quad \gamma_{-m-1,-}:=\Sigma_{4}^{(1)},
\end{align*}
$s_{m+1}:=1$, $x_{m+1}:=+\infty$ and for $j=1,\ldots,m,m+1$, define
\begin{align*}
& J_{\gamma_{j,-}}(z) = \begin{pmatrix}
1 & 0 & 0 \\
\frac{e^{\theta_{1}(rz)-\theta_{2}(rz)}}{s_{j}} & 1 & -e^{\theta_{1}(rz)-\theta_{3}(rz)} \\
0 & 0 & 1
\end{pmatrix}, & & J_{\gamma_{j,+}}(z) = \begin{pmatrix}
1 & 0 & 0 \\
\frac{e^{\theta_{2}(rz)-\theta_{1}(rz)}}{s_{j}} & 1 & e^{\theta_{2}(rz)-\theta_{3}(rz)} \\
0 & 0 & 1
\end{pmatrix}, \\
& J_{\gamma_{-j,+}}(z) = \begin{pmatrix}
1 & 0 & 0 \\
0 & 1 & 0 \\
\frac{e^{\theta_{3}(rz)-\theta_{1}(rz)}}{s_{j}} & e^{\theta_{3}(rz)-\theta_{2}(rz)} & 1
\end{pmatrix}, & & J_{\gamma_{-j,-}}(z) = \begin{pmatrix}
1 & 0 & 0 \\
0 & 1 & 0 \\
\frac{e^{\theta_{3}(rz)-\theta_{2}(rz)}}{s_{j}} & -e^{\theta_{3}(rz)-\theta_{1}(rz)} & 1
\end{pmatrix}.
\end{align*}
The next transformation is defined by
\begin{align}\label{def of S}
S(z) = T(z)\begin{cases}
J_{\gamma_{j,-}}(z), & \im z < 0 \mbox{ and } z \mbox{ above } \gamma_{j,-}, \quad \hspace{0.2cm} j \in \{1,\ldots,m\}, \\
J_{\gamma_{j,+}}(z)^{-1}, & \im z > 0 \mbox{ and } z \mbox{ below } \gamma_{j,+}, \quad \hspace{0.2cm} j \in \{1,\ldots,m\}, \\
J_{\gamma_{-j,+}}(z)^{-1}, & \im z > 0 \mbox{ and } z \mbox{ below } \gamma_{-j,+}, \quad j \in \{1,\ldots,m\}, \\
J_{\gamma_{-j,-}}(z), & \im z < 0 \mbox{ and } z \mbox{ above } \gamma_{-j,-}, \quad j \in \{1,\ldots,m\}, \\
I, & \mbox{otherwise}.
\end{cases}
\end{align}
$S$ satisfies the following RH problem.
\subsubsection*{RH problem for $S$}
\begin{itemize}
\item[(a)] $S: \mathbb{C}\setminus \Sigma_{S} \to \mathbb{C}^{3\times 3}$ is analytic, where $\Sigma_{S}:=(-\infty,+\infty) \cup \bigcup_{j=1}^{m+1} \big( \gamma_{j,-} \cup \gamma_{j,+} \cup \gamma_{-j,+} \cup \gamma_{-j,-} \big)$.
\item[(b)] For $z \in \Sigma_{S}\setminus \cup_{j=0}^{m}\{-x_{j},x_{j}\}$, $S_{+}(z) = S_{-}(z)J_{S}(z)$, where
\begin{align*}
& J_{S}(z)=J_{\gamma_{j,-}}(z), & & z \in \gamma_{j,-}, & & J_{S}(z)=J_{\gamma_{j,+}}(z), & & z \in  \gamma_{j,+}, \\
& J_{S}(z)=J_{\gamma_{-j,+}}(z), & & z \in \gamma_{-j,+}, & & J_{S}(z)=J_{\gamma_{-j,-}}(z), & & z \in \gamma_{-j,-} \\
& J_{S}(z)= \begin{pmatrix}
0 & s_{j} & 0 \\
-s_{j}^{-1} & 0 & 0 \\
0 & 0 & 1
\end{pmatrix}, & & z \in (x_{j-1},x_{j}), & & J_{S}(z)= \begin{pmatrix}
0 & 0 & s_{j} \\
0 & 1 & 0 \\
-s_{j}^{-1} & 0 & 0 
\end{pmatrix}, & & z \in (-x_{j},-x_{j-1}),
\end{align*}
where $j=1,\ldots,m,m+1$.
\item[(c)] As $z \to \infty$, $\pm \im z > 0$, we have
\begin{align*}
S(z) = \bigg( I+\frac{T_{1}}{z} +\bigO(z^{-2}) \bigg) \diag\Big( z^{-\frac{1}{3}},1,z^{\frac{1}{3}} \Big)L_{\pm}.
\end{align*}
\item[(d)] As $z \to z_{\star}\in \cup_{j=1}^{m}\{-x_{j},x_{j}\}$, we have $S(z) = \bigO(\log (z-z_{\star}))$.

As $z \to 0$, $S(z) = \bigO(1)$. 
\item[(e)] $S$ satisfies the symmetry $S(z) = -\diag(1,-1,1)S(-z)\mathcal{B}$, where $\mathcal{B}$ is defined in \eqref{symmetry of Phi}.
\end{itemize}
\subsection{Global parametrix}
Using the definitions of $\theta_{1},\theta_{2},\theta_{3}$ given in \eqref{def of theta}, it is easily checked that $J_{S}(z) \to I$ as $r \to \infty$ for each $z \in \cup_{j=1}^{m+1} \big( \gamma_{j,-} \cup \gamma_{j,+} \cup \gamma_{-j,+} \cup \gamma_{-j,-} \big)$. The following RH problem, whose solution is denoted $N$ and called the global parametrix, has the same jump conditions on $(-\infty,+\infty)$ than the RH problem for $S$, and no other jumps. We will show in Subsection \ref{subsection:small norm} that $N$ is a good approximation to $S$ outside small neighborhoods of $\cup_{j=0}^{m}\{-x_{j},x_{j}\}$. The RH problem for $N$ is as follows.
\subsubsection*{RH problem for $N$}
\begin{itemize}
\item[(a)] $N:\mathbb{C}\setminus (-\infty,+\infty) \to \mathbb{C}^{3\times 3}$ is analytic.
\item[(b)] $N$ satisfies the following jump relations:
\begin{align*}
& N_{+}(z) = N_{-}(z) \begin{pmatrix}
0 & 0 & s_{j} \\
0 & 1 & 0 \\
-s_{j}^{-1} & 0 & 0 
\end{pmatrix}, & & z \in (-x_{j},-x_{j-1}), \quad j=1,\ldots,m,m+1 \\
& N_{+}(z) = N_{-}(z) \begin{pmatrix}
0 & s_{j} & 0 \\
-s_{j}^{-1} & 0 & 0 \\
0 & 0 & 1
\end{pmatrix}, & & z \in (x_{j-1},x_{j}), \hspace{0.9cm} j=1,\ldots,m,m+1.
\end{align*}
\item[(c)] As $z \to \infty$, $\pm \im z >0$, we have
\begin{align}\label{asymp of N at inf}
N(z) = \bigg( I + \frac{1}{z}N_{1} + \bigO(z^{-2}) \bigg) \diag (z^{-\frac{1}{3}}, 1, z^{\frac{1}{3}})L_{\pm},
\end{align}
for a certain matrix $N_{1}$.
\item[(d)] As $z \to z_{\star}\in \cup_{j=1}^{m}\{-x_{j},x_{j}\}$, $N(z) = \bigO(1)$.

As $z \to 0$, $N(z)=\bigO(1)\diag(z^{-\frac{1}{3}},1,z^{\frac{1}{3}})\bigO(1)$.
\item[(e)] $N$ satisfies the symmetry $N(z) = -\diag(1,-1,1)N(-z)\mathcal{B}$.
\end{itemize}
For $m=1$, the above RH problem was solved explicitly in \cite[Section 5.3]{DXZ2020 thinning}. Let us define
\begin{align}\label{betaj sj relation}
\beta_{j} := \frac{1}{2\pi i}u_{j} = \frac{1}{2\pi i}\log \frac{s_{j}}{s_{j+1}}, \qquad j=1,\ldots,m.
\end{align}
Inspired by \cite{DXZ2020 thinning}, we consider three functions $d_{1},d_{2},d_{3}$ defined by
\begin{align*}
& d_{1}(z) = \begin{cases}
\lambda(z^{\frac{1}{3}}), & \im z >0, \\
\lambda(\omega^{-1}z^{\frac{1}{3}}), & \im z <0,
\end{cases} & & d_{2}(z) = \begin{cases}
\lambda(\omega^{-1}z^{\frac{1}{3}}), & \im z >0, \\
\lambda(z^{\frac{1}{3}}), & \im z <0,
\end{cases} & & d_{3}(z) = \lambda(\omega z^{\frac{1}{3}}),
\end{align*}
where 
\begin{align}\label{def of lambda}
\lambda(z) = \prod_{j=1}^{m}\bigg( \frac{z^{2}-\omega x_{j}^{2/3}}{z^{2}-x_{j}^{2/3}} \bigg)^{\beta_{j}}, \qquad z \in \mathbb{C}\setminus \big( (-x_{m}^{\frac{1}{3}},x_{m}^{\frac{1}{3}})\cup \omega^{-1}(-x_{m}^{\frac{1}{3}},x_{m}^{\frac{1}{3}})\big).
\end{align}
The branch structure for $\lambda$ is such that $\lambda(z) = 1+\bigO(z^{-1})$ as $z \to \infty$, and
\begin{align}\label{jumps of lambda}
\lambda_{+}(z) = \lambda_{-}(z) \begin{cases}
s_{j}, & z \in (-x_{j}^{\frac{1}{3}},-x_{j-1}^{\frac{1}{3}})\cup e^{-\frac{2\pi i}{3}}(x_{j-1}^{\frac{1}{3}},x_{j}^{\frac{1}{3}}), \quad j \in \{1,\ldots,m\}, \\
s_{j}^{-1}, & z \in e^{\frac{\pi i}{3}}(x_{j-1}^{\frac{1}{3}},x_{j}^{\frac{1}{3}})\cup (x_{j-1}^{\frac{1}{3}},x_{j}^{\frac{1}{3}}), \hspace{1.25cm} j \in \{1,\ldots,m\},
\end{cases}
\end{align}
where the boundary values $\lambda_{+}$ and $\lambda_{-}$ are taken with respect to the orientation of the contour as stated in \eqref{def of lambda}. Using \eqref{def of lambda}--\eqref{jumps of lambda}, it can be verified that
\begin{align}
& N(z) := C_{N} \diag(z^{-\frac{1}{3}},1,z^{\frac{1}{3}}) L_{\pm} \diag(d_{1}(z),d_{2}(z),d_{3}(z)), \qquad \pm \im z >0, \label{def of N} 
\end{align}
is the unique solution to the RH problem for $N$, where 
\begin{align}
& C_{N} := \begin{pmatrix}
1 & 0 & 0 \\
0 & 1 & 0 \\
-i\sqrt{3} \sum_{j=1}^{m} \beta_{j}x_{j}^{2/3} & 0 & 1
\end{pmatrix}. \label{def of CN}
\end{align}
It the next subsections we compute more detailed asymptotic expansions than those stated in conditions (c) and (d) of the RH problem for $N$.
\subsubsection{Asymptotics of $N(z)$ as $z \to \infty$}
As $z \to \infty$, $\pm \im z >0$, we have
\begin{align}\label{asymp of N at inf}
N(z) = \bigg( I + \frac{1}{z}N_{1} + \bigO(z^{-2}) \bigg) \diag (z^{-\frac{1}{3}}, 1, z^{\frac{1}{3}})L_{\pm}
\end{align}
where $N_{1}$ is of the form
\begin{align}\label{def of N1}
N_{1}=\begin{pmatrix}
0 & i \sqrt{3} \sum_{j=1}^{m} \beta_{j}x_{j}^{2/3} & 0 \\
\star & 0 & i \sqrt{3} \sum_{j=1}^{m} \beta_{j}x_{j}^{2/3} \\
0 & \star & 0
\end{pmatrix}.
\end{align}
The entries $(N_{1})_{21}$ and $(N_{1})_{32}$ can also be computed explicitly, but their expressions are longer and not important for us.
\subsubsection{Asymptotics of $N(z)$ as $z \to x_{j}$, $j=1,\ldots,m$}\label{subsubsection: asymptotics of global param near xj}
As $z \to x_{j}$, $\im z > 0$, $j=1,\ldots,m$,
\begin{align*}
& d_{1}(z) = d_{1,x_{j}}^{(0)} (z-x_{j})^{-\beta_{j}}\big( 1 + d_{1,x_{j}}^{(1)}(z-x_{j}) + \bigO((z-x_{j})^{2}) \big),  \\
& d_{1,x_{j}}^{(0)} = \bigg( \frac{3\sqrt{3}x_{j}}{2} \bigg)^{\beta_{j}}e^{-\frac{\pi i \beta_{j}}{6}} \prod_{\substack{k=1\\k \neq j}}^{m} \frac{(x_{j}^{2/3}-\omega x_{k}^{2/3})^{\beta_{k}}}{(x_{j}^{2/3}-x_{k}^{2/3})_{+}^{\beta_{k}}},  \\
& d_{1,x_{j}}^{(1)} = \frac{(\omega-5)\beta_{j}}{6(\omega-1)x_{j}} + \sum_{\substack{k=1 \\ k \neq j}}^{m} \frac{2(\omega-1)x_{k}^{2/3}\beta_{k}}{3x_{j}^{1/3}(x_{j}^{2/3}-x_{k}^{2/3})(x_{j}^{2/3}-\omega x_{k}^{2/3})}, \\
& d_{2}(z) = d_{2,x_{j}}^{(0)} (z-x_{j})^{\beta_{j}}\big( 1 + d_{2,x_{j}}^{(1)}(z-x_{j}) + \bigO((z-x_{j})^{2}) \big), \\
& d_{2,x_{j}}^{(0)} = \bigg( \frac{2}{3\sqrt{3}x_{j}} \bigg)^{\beta_{j}}e^{-\frac{\pi i \beta_{j}}{6}}\prod_{\substack{k=1\\k \neq j}}^{m}\frac{(x_{j}^{2/3}-x_{k}^{2/3})_{+}^{\beta_{k}}}{(x_{j}^{2/3}-\omega^{2} x_{k}^{2/3})^{\beta_{k}}}, \\
& d_{2,x_{j}}^{(1)} = \frac{(1-5\omega)\beta_{j}}{6(\omega-1)x_{j}} + \sum_{\substack{k=1 \\ k \neq j}}^{m} \frac{2(1-\omega^{2})x_{k}^{2/3}\beta_{k}}{3x_{j}^{1/3}(x_{j}^{2/3}-x_{k}^{2/3})(x_{j}^{2/3}-\omega^{2} x_{k}^{2/3})}, \\
& d_{3}(z) = d_{3,x_{j}}^{(0)} \big( 1 + d_{3,x_{j}}^{(1)}(z-x_{j}) + \bigO((z-x_{j})^{2}) \big), \\
& d_{3,x_{j}}^{(0)} = e^{\frac{\pi i \beta_{j}}{3}} \prod_{\substack{k=1\\k \neq j}}^{m}\frac{(x_{j}^{2/3}-\omega^{2}x_{k}^{2/3})^{\beta_{k}}}{(x_{j}^{2/3}-\omega x_{k}^{2/3})^{\beta_{k}}}, \\
& d_{3,x_{j}}^{(1)} = \frac{2(\omega + 1)\beta_{j}}{3(\omega-1)x_{j}} + \sum_{\substack{k=1 \\ k \neq j}}^{m} \frac{2(\omega^{2}-\omega)x_{k}^{2/3}\beta_{k}}{3x_{j}^{1/3}(x_{j}^{2/3}-\omega x_{k}^{2/3})(x_{j}^{2/3}-\omega^{2} x_{k}^{2/3})}.
\end{align*}
In the above asymptotic expansions, all branches are the principal ones: for example, for the product appearing in $d_{1,x_{j}}^{(0)}$, we have
\begin{align*}
& \Bigg|\prod_{\substack{k=1\\k \neq j}}^{m} \frac{(x_{j}^{2/3}-\omega x_{k}^{2/3})^{\beta_{k}}}{(x_{j}^{2/3}-x_{k}^{2/3})_{+}^{\beta_{k}}} \Bigg| =  \exp \Bigg( -\sum_{\substack{k=1\\ k \neq j}}^{m} i\beta_{k} \arctan \frac{\sqrt{3}x_{k}^{2/3}}{x_{k}^{2/3}+2x_{j}^{2/3}} - \sum_{k=j+1}^{m} \pi i \beta_{k} \Bigg), \\
& \arg \prod_{\substack{k=1\\k \neq j}}^{m} \frac{(x_{j}^{2/3}-\omega x_{k}^{2/3})^{\beta_{k}}}{(x_{j}^{2/3}-x_{k}^{2/3})_{+}^{\beta_{k}}} = -\sum_{\substack{k=1\\ k \neq j}}^{m} i \beta_{k} \log \frac{|x_{j}^{2/3}-\omega x_{k}^{2/3}|}{|x_{j}^{2/3}- x_{k}^{2/3}|} \mod 2\pi.
\end{align*}
Hence, as $z \to x_{j}$, $\im z > 0$, $j=1,\ldots,m$, 
\begin{align}\label{asymptotics of N near xj}
N(z) = (N_{x_{j}}^{(0)} + (z-x_{j})N_{x_{j}}^{(1)} + \bigO\big( (z-x_{j})^{2} \big))(z-x_{j})^{-\beta_{j}\sigma_{3,1}},
\end{align}
where $\sigma_{3,1} = \diag(1,-1,0)$ and
\begin{align}
& N_{x_{j}}^{(0)} = C_{N} \diag(x_{j}^{-1/3},1,x_{j}^{1/3}) L_{+} \diag(d_{1,x_{j}}^{(0)},d_{2,x_{j}}^{(0)},d_{3,x_{j}}^{(0)}), \label{def of Nxjp0p} \\
& N_{x_{j}}^{(1)} = C_{N} \Big( \diag(x_{j}^{-1/3},1,x_{j}^{1/3}) L_{+} \diag(d_{1,x_{j}}^{(0)}d_{1,x_{j}}^{(1)},d_{2,x_{j}}^{(0)}d_{2,x_{j}}^{(1)},d_{3,x_{j}}^{(0)}d_{3,x_{j}}^{(1)}) \nonumber \\
& \hspace{1.7cm} + \diag(-\tfrac{1}{3}x_{j}^{-4/3},0,\tfrac{1}{3}x_{j}^{-2/3}) L_{+} \diag(d_{1,x_{j}}^{(0)},d_{2,x_{j}}^{(0)},d_{3,x_{j}}^{(0)}) \Big). \nonumber
\end{align}

\subsubsection{Asymptotics of $N(z)$ as $z \to 0$}
As $z \to 0$, $\im z>0$, 
\begin{align*}
& d_{\ell}(z) = d_{\ell,0}^{(0)}(1+d_{\ell,0}^{(1)}z^{\frac{2}{3}} + \bigO(z^{\frac{4}{3}})), \qquad \ell = 1,2,3, \\
& d_{1,0}^{(0)} = e^{-\pi i \sum_{k=1}^{m}\beta_{k}}(-\omega)^{\sum_{k=1}^{m}\beta_{k}} = s_{1}^{-\frac{2}{3}}, \qquad d_{2,0}^{(0)} = d_{3,0}^{(0)} = \omega^{\sum_{k=1}^{m}\beta_{k}} = s_{1}^{\frac{1}{3}}, \\
& d_{1,0}^{(1)} = (1-\omega^{2})\sum_{k=1}^{m} \beta_{k}x_{k}^{-\frac{2}{3}}, \quad d_{2,0}^{(1)} = (\omega-1)\sum_{k=1}^{m} \beta_{k}x_{k}^{-\frac{2}{3}}, \quad d_{1,0}^{(1)} = (\omega^{2}-\omega)\sum_{k=1}^{m} \beta_{k}x_{k}^{-\frac{2}{3}},
\end{align*}
where the branches are the principal ones. As $z \to 0$, $\im z>0$,
\begin{align}\label{asymptotics of N near 0}
N(z) = C_{N} \diag(z^{-\frac{1}{3}},1,z^{\frac{1}{3}}) L_{+}  \diag(d_{1,0}^{(0)},d_{2,0}^{(0)},d_{3,0}^{(0)}) \Big( I + z^{\frac{2}{3}}\diag(d_{1,0}^{(1)},d_{2,0}^{(1)},d_{3,0}^{(1)}) + \bigO(z^{\frac{4}{3}}) \Big).
\end{align}

\paragraph{Local parametrices.} For each $p \in \{-x_{m},\ldots,-x_{1},0,x_{1},\ldots,x_{m}\}$, we let $\mathcal{D}_{p}$ be a small open disk centered at $p$. The local parametrix $P^{(p)}$ is defined inside $\mathcal{D}_{p}$, has the same jumps as $S$ in $\mathcal{D}_{p}$, and satisfies $S(z)P^{(p)}(z)^{-1}=\bigO(1)$ as $z \to p$. Furthermore, we require $P^{(p)}$ to satisfy the following matching condition with $P^{(\infty)}$ on $\partial \mathcal{D}_{p}$:
\begin{align}\label{matching condition general}
P^{(p)}(z) = (I+o(1))P^{(\infty)}(z), \qquad \mbox{as } r \to + \infty,
\end{align}
uniformly for $z \in \partial \mathcal{D}_{p}$. 
\subsection{Local parametrix near $x_{j}$, $j=1,\ldots,m$}\label{subsection: local param near xj}
Since the (1,2) and (2,1) entries of $J_{S}(z)$ have each a discontinuity at $z=x_{j}$, we can follow \cite{ItsKrasovsky} and build $P^{(x_{j})}$ using the model RH problem $\Phi_{\mathrm{HG}}$ which is presented in Appendix \ref{subsection: model RHP with HG functions}. We also refer to \cite[Section 5.5]{DXZ2020 thinning} for more details about this construction. The local parametrix $P^{(x_{j})}$ is of the form
\begin{align}\label{def of Pxj}
& P^{(x_{j})}(z) = E_{x_{j}}(z) \begin{pmatrix}
\Phi_{\mathrm{HG},11}(r^{\frac{4}{3}}f_{x_{j}}(z);\beta_{j}) & \Phi_{\mathrm{HG}, 12}(r^{\frac{4}{3}}f_{x_{j}}(z);\beta_{j}) & 0 \\
\Phi_{\mathrm{HG},21}(r^{\frac{4}{3}}f_{x_{j}}(z);\beta_{j}) & \Phi_{\mathrm{HG}, 22}(r^{\frac{4}{3}}f_{x_{j}}(z);\beta_{j}) & 0 \\
0 & 0 & 1
\end{pmatrix} \\
& \hspace{1.8cm} \times (s_{j}s_{j+1})^{-\frac{\sigma_{3,1}}{4}}e^{\pm \frac{1}{2}(\theta_{2}(rz)-\theta_{1}(rz))\sigma_{3,1}}A_{x_{j}}(z), \nonumber \\
& A_{x_{j}}(z)= \begin{cases}
\begin{pmatrix}
1 & 0 & 0 \\
0 & 1 & e^{\theta_{2}(rz)-\theta_{3}(rz)} \\
0 & 0 & 1
\end{pmatrix}, & z \in \{z:\im z > 0\}\cap \mathcal{D}_{x_{j}} \setminus (\Omega_{j,+}\cup \Omega_{j+1,+}), \\
\begin{pmatrix}
1 & 0 & 0 \\
0 & 1 & e^{\theta_{1}(rz)-\theta_{3}(rz)} \\
0 & 0 & 1
\end{pmatrix}, & z \in \{z:\im z < 0\}\cap \mathcal{D}_{x_{j}} \setminus (\Omega_{j,-}\cup \Omega_{j+1,-}), \\
I, & \mbox{otherwise},
\end{cases} \nonumber
\end{align}
where $\sigma_{3,1}=\diag(1,-1,0)$, $\pm$ stands for $\pm \im z > 0$, $E_{x_{j}}$ is analytic in $\mathcal{D}_{x_{j}}$ and given by
\begin{align*}
E_{x_{j}}(z) = N(z) (s_{j}s_{j+1})^{\frac{\sigma_{3,1}}{4}} \left\{ \begin{array}{l l}
\sqrt{\frac{s_{j+1}}{s_{j}}}^{\sigma_{3,1}}, & \im z > 0 \\
\begin{pmatrix}
0 & 1 & 0 \\
-1 & 0 & 0 \\
0 & 0 & 1
\end{pmatrix}, & \im z <0
\end{array} \right\} 
e^{-\frac{1}{2}(\theta_{2}(rx_{j})-\theta_{1}(rx_{j}))\sigma_{3,1}} (r^{\frac{4}{3}}f_{x_{j}}(z))^{\beta_{j}\sigma_{3,1}},
\end{align*}
and $f_{x_{j}}$ is given by
\begin{align*}
f_{x_{j}}(z) & = r^{-\frac{4}{3}} \big[ (\theta_{2}(rz)-\theta_{1}(rz)) - (\theta_{2}(rx_{j})-\theta_{1}(rx_{j})) \big] = \frac{i \sqrt{3}}{4} \Big( 3(z^{\frac{4}{3}}-x_{j}^{\frac{4}{3}}) - \frac{2\rho}{r^{\frac{2}{3}}}(z^{\frac{2}{3}} - x_{j}^{\frac{2}{3}}) \Big).
\end{align*}
It is easily checked that $A_{x_{j}}(z)$ is exponentially small as $r \to + \infty$ uniformly for $z \in \mathcal{D}_{x_{j}}$, and that
\begin{align}\label{def of coeff expansion fxj}
f_{x_{j}}(x_{j}) = 0, \quad f_{x_{j}}'(x_{j}) = i \bigg( \sqrt{3}x_{j}^{1/3} - \frac{\rho}{\sqrt{3} x_{j}^{1/3} r^{2/3}} \bigg), \quad f_{x_{j}}''(x_{j}) = i \bigg( \frac{1}{\sqrt{3}x_{j}^{2/3}} + \frac{\rho}{3\sqrt{3} x_{j}^{4/3}r^{2/3}} \bigg).
\end{align}
Using \eqref{asymptotics of N near xj}, we obtain 
\begin{align}
& E_{x_{j}}(x_{j}) = N_{x_{j}}^{(0)} \sqrt{s_{j+1}}^{\sigma_{3,1}} e^{-\frac{1}{2}(\theta_{2}(rx_{j})-\theta_{1}(rx_{j}))\sigma_{3,1}} (r^{\frac{4}{3}}|f'(x_{j})|)^{\beta_{j}\sigma_{3,1}}, \label{Exjxj} \\
& E_{x_{j}}(x_{j})^{-1}E_{x_{j}}'(x_{j}) = \begin{pmatrix}
\frac{f_{x_{j}}''(x_{j})}{2f_{x_{j}}'(x_{j})}\beta_{j} + d_{1,x_{j}}^{(1)} & \frac{i\mathfrak{c}_{j}^{-2}}{3\sqrt{3}x_{j}}\frac{d_{2,x_{j}}^{(0)}}{d_{1,x_{j}}^{(0)}} & -\frac{i\mathfrak{c}_{j}^{-1}}{3\sqrt{3}x_{j}}\frac{d_{3,x_{j}}^{(0)}}{d_{1,x_{j}}^{(0)}} \\
-\frac{i\mathfrak{c}_{j}^{2}}{3\sqrt{3}x_{j}}\frac{d_{1,x_{j}}^{(0)}}{d_{2,x_{j}}^{(0)}} & -\frac{f_{x_{j}}''(x_{j})}{2f_{x_{j}}'(x_{j})}\beta_{j} + d_{2,x_{j}}^{(1)} & -\frac{i\mathfrak{c}_{j}}{3\sqrt{3}x_{j}}\frac{d_{3,x_{j}}^{(0)}}{d_{2,x_{j}}^{(0)}} \\
\frac{i\mathfrak{c}_{j}}{3\sqrt{3}x_{j}}\frac{d_{1,x_{j}}^{(0)}}{d_{3,x_{j}}^{(0)}} & \frac{i\mathfrak{c}_{j}^{-1}}{3\sqrt{3}x_{j}}\frac{d_{2,x_{j}}^{(0)}}{d_{3,x_{j}}^{(0)}} & d_{3,x_{j}}^{(1)}
\end{pmatrix}, \label{Exjxj der} \\
& \mathfrak{c}_{j} = \sqrt{s_{j+1}}e^{-\frac{1}{2}(\theta_{2}(rx_{j})-\theta_{1}(rx_{j}))}\big( r^{\frac{4}{3}}|f_{x_{j}}'(x_{j})| \big)^{\beta_{j}}. \label{def of cfrak j}
\end{align}
Using \eqref{Asymptotics HG}, we obtain
\begin{align}\label{matching of Pxj}
P^{(x_{j})}(z)N(z)^{-1} = I + \frac{1}{r^{\frac{4}{3}}f_{x_{j}}(z)} E_{x_{j}}(z)\begin{pmatrix}
\Phi_{\mathrm{HG},1}(\beta_{j})_{11} & \Phi_{\mathrm{HG},1}(\beta_{j})_{12}  & 0 \\
\Phi_{\mathrm{HG},1}(\beta_{j})_{21}  & \Phi_{\mathrm{HG},1}(\beta_{j})_{22}  & 0 \\
0 & 0 & 1
\end{pmatrix}E_{x_{j}}(z)^{-1} + \bigO(r^{-\frac{8}{3}}),
\end{align}
as $r \to + \infty$ uniformly for $z \in \partial \mathcal{D}_{x_{j}}$.

\subsection{Local parametrix near $-x_{j}$, $j=1,\ldots,m$}\label{subsection: local param near -xj}
$P^{(-x_{j})}$ can be constructed in terms of $\Phi_{\mathrm{HG}}$ in a similar way as $P^{(x_{j})}$.    Alternatively, we can use the symmetry stated in condition (e) of the RH problem for $S$. This observation saves us some effort and allows us to see immediately that
\begin{align}\label{symmetry of local param}
P^{(-x_{j})}(z) = -\diag(1,-1,1)P^{(x_{j})}(z)
\mathcal{B}, \qquad z \in \mathcal{D}_{-x_{j}}.
\end{align}

\subsection{Local parametrix near $0$}\label{subsection: local param near 0}
For the local parametrix $P^{(0)}$, we need to use the model RH problem for $\Psi$ from \cite{BleherKuijlaarsIII} (and which is recalled in Subsection \ref{subsection: model RH problem for Psi} for the convenience of the reader). Define
\begin{align}\label{def of Pp0p}
P^{(0)}(z) = E_{0}(z)\Psi(rz) e^{-\Theta(rz)}\diag(s_{1}^{-\frac{2}{3}},s_{1}^{\frac{1}{3}},s_{1}^{\frac{1}{3}}),
\end{align}
where $E_{0}$ is analytic inside $\mathcal{D}_{0}$ and given by
\begin{align*}
E_{0}(z) = -\sqrt{\frac{3}{2\pi}}e^{-\frac{\rho^{2}}{6}}i N(z) \diag(s_{1}^{\frac{2}{3}},s_{1}^{-\frac{1}{3}},s_{1}^{-\frac{1}{3}}) L_{\pm}^{-1}\diag((rz)^{\frac{1}{3}},1,(rz)^{-\frac{1}{3}})\Psi_{0}^{-1}, \quad \pm \im z > 0,
\end{align*}
and 
\begin{align*}
\Psi_{0} = \begin{pmatrix}
1 & 0 & 0 \\
0 & 1 & 0 \\
\kappa_{3}+\frac{2\rho}{3} & 0 & 1
\end{pmatrix}, \qquad \kappa_{3} = \frac{\rho^{3}}{54}-\frac{\rho}{6}.
\end{align*}
In a similar way as in \cite[Proposition 5.13]{DXZ2020 thinning}, we verify that $P^{(0)}$ has the same jumps as $S$ inside $\mathcal{D}_{0}$. It is also clear from condition (d) of the RH problem for $\Psi$ that $P^{(0)}(z)$ remains bounded as $z \to 0$. Finally, using \eqref{def of Pp0p} and \eqref{asymp of Psi at infty}, we obtain
\begin{align}\label{matching of Pp0p}
P^{(0)}(z)N(z)^{-1} = I + \frac{1}{r^{\frac{2}{3}}z} \widehat{E}_{0}(z)\widehat{\Psi}_{1}\widehat{E}_{0}(z)^{-1} + \bigO(r^{-\frac{4}{3}}), \qquad \widehat{\Psi}_{1} = \begin{pmatrix}
0 & \kappa_{3} & 0 \\
0 & 0 & \kappa_{3} + \frac{\rho}{3} \\
0 & 0 & 0
\end{pmatrix},
\end{align}
as $r \to + \infty$ uniformly for $z \in \partial \mathcal{D}_{0}$, where 
\begin{align*}
\widehat{E}_{0}(z) := N(z) \diag(s_{1}^{\frac{2}{3}},s_{1}^{-\frac{1}{3}},s_{1}^{-\frac{1}{3}}) L_{\pm}^{-1}\diag(z^{\frac{1}{3}},1,z^{-\frac{1}{3}}).
\end{align*}
For future reference, using \eqref{asymptotics of N near 0} we note that $\widehat{E}_{0}(0) = C_{N}$.

\subsection{Final transformation}\label{subsection:small norm}
Define
\begin{align}\label{def of R}
R(z) = \begin{cases}
S(z)N(z)^{-1}, & z \in \mathbb{C}\setminus \big( \cup_{j=1}^{m}(\mathcal{D}_{x_{j}}\cup \mathcal{D}_{-x_{j}}) \cup \mathcal{D}_{0}  \big), \\
S(z)P^{(p)}(z)^{-1}, & z \in \mathcal{D}_{p}, \; p \in \{-x_{m},\ldots,-x_{1},0,x_{1},\ldots,x_{m}\}.
\end{cases}
\end{align}
Using the analysis of Subsections \ref{subsection: local param near xj}--\ref{subsection: local param near 0}, we conclude that $R$ is analytic in $\cup_{j=1}^{m}(\mathcal{D}_{x_{j}}\cup \mathcal{D}_{-x_{j}}) \cup \mathcal{D}_{0}$, and therefore $R$ is analytic in $\mathbb{C}\setminus \Sigma_{R}$, where 
\begin{align*}
\Sigma_{R}:= \cup_{j=1}^{m}(\partial \mathcal{D}_{x_{j}}\cup \partial \mathcal{D}_{-x_{j}})\cup \partial \mathcal{D}_{0} \cup \Sigma_{S} \setminus \big( \cup_{j=1}^{m}(\mathcal{D}_{x_{j}}\cup \mathcal{D}_{-x_{j}}) \cup \mathcal{D}_{0}\cup \mathbb{R} \big).
\end{align*}
For convenience, we orient the boundaries of the $2m+1$ disks in the \textit{clockwise} direction, and for $z \in \Sigma_{R}$, we define $J_{R}(z):=R_{-}^{-1}(z)R_{+}(z)$. Using the definitions \eqref{def of theta} of $\theta_{1},\theta_{2},\theta_{3}$, condition (b) of the RH problem for $S$ and \eqref{def of R}, we verify that 
\begin{align*}
J_{R}(z) = I + \bigO(e^{-cr}), \qquad \mbox{as } r \to + \infty \mbox{ uniformly for } z \in \Sigma_{R}\setminus \big(  \cup_{j=1}^{m}(\mathcal{D}_{x_{j}}\cup \mathcal{D}_{-x_{j}}) \cup \mathcal{D}_{0} \big)
\end{align*}
for a certain $c>0$. Also, by \eqref{matching of Pxj}, \eqref{symmetry of local param} and \eqref{matching of Pp0p}, we have 
\begin{align*}
& J_{R}(z) = I + \bigO(r^{-\frac{4}{3}}), & & \mbox{as } r \to + \infty, \mbox{ uniformly for } z \in \cup_{j=1}^{m}(\partial \mathcal{D}_{x_{j}}\cup \partial \mathcal{D}_{-x_{j}}), \\
& J_{R}(z) = I + \tfrac{J^{(1)}(z)}{r^{2/3}} + \bigO(r^{-\frac{4}{3}}), & & \mbox{as } r \to + \infty, \mbox{ uniformly for } z \in \partial \mathcal{D}_{0},
\end{align*}
where $J^{(1)}(z) = \frac{1}{z} \widehat{E}_{0}(z)\widehat{\Psi}_{1}\widehat{E}_{0}(z)^{-1}$. Hence, $R$ satisfies a small norm RH problem \cite{DeiftZhou}, and we have
\begin{align}\label{asymp of R as r to infty}
R(z) = I + R^{(1)}(z)r^{-\frac{2}{3}} + \bigO(r^{-\frac{4}{3}}), \qquad \mbox{as } r \to + \infty,
\end{align}
uniformly for $z \in \mathbb{C}\setminus \Sigma_{R}$, with
\begin{align}\label{Rp1p def}
R^{(1)}(z) = \frac{1}{2\pi i}\oint_{\partial \mathcal{D}_{0}} \frac{J^{(1)}(x)dx}{x-z} = \begin{cases}
\frac{C_{N}\widehat{\Psi}_{1}C_{N}^{-1}}{z}, & z \notin \mathcal{D}_{0}, \\
\frac{C_{N}\widehat{\Psi}_{1}C_{N}^{-1}}{z}-J^{(1)}(z), & z \in \mathcal{D}_{0},
\end{cases}
\end{align}
and where we recall that $\partial \mathcal{D}_{0}$ is oriented in the clockwise direction. (the method of \cite{DeiftZhou} also implies that $R$ exists for all sufficiently large $r$; note however that here we already know from \eqref{def of Y} and from $\det (1-\widetilde{\mathcal{K}}^{\mathrm{Pe}}_{\rho})>0$ that $Y$ exists for \textit{all} $r>0$, which implies by the transformations $Y\to \Phi \to T \to S \to R$ that $R$ also exists for \textit{all} $r>0$). Finally, the same analysis as in \cite[Section 3.5]{ChLen Bessel} shows that for any $k_{1},\ldots,k_{m} \in \mathbb{N}_{\geq 0}$ with $k_{1}+\ldots+k_{m}\geq 1$, we have
\begin{align}\label{der of asymp of R as r to infty}
\partial_{u_{1}}^{k_{1}}\ldots \partial_{u_{m}}^{k_{m}}R(z) = \partial_{u_{1}}^{k_{1}}\ldots \partial_{u_{m}}^{k_{m}}R^{(1)}(z)r^{-\frac{2}{3}} + \bigO((\log r)^{k_{1}+\ldots+k_{m}}r^{-\frac{4}{3}}), \qquad \mbox{as } r \to + \infty,
\end{align}
uniformly for $z \in \mathbb{C}\setminus \Sigma_{R}$.

\section{Asymptotic analysis of $\Phi$ as $r \to 0$}\label{section: steepest descent for small r}
The analysis of $\Phi$ as $r \to 0$ is much simpler than the large $r$ analysis of Section \ref{section: steepest descent for large r}. Here we generalize \cite[Section 6]{DXZ2020 thinning} to an arbitrary $m$. Let $\delta>0$ be fixed. Define
\begin{align}\label{def of N tilde}
\widetilde{N}(z) := -\sqrt{\tfrac{3}{2\pi}}e^{-\frac{\rho^{2}}{6}}i \Psi_{0}^{-1} \Psi(z) \times \begin{cases}
J_{1}, & \mbox{for } \arg z < \frac{\pi}{4} \mbox{ and } \arg (z-rx_{m})>\frac{\pi}{4}, \\
J_{2}, & \mbox{for } \arg z > \frac{3\pi}{4} \mbox{ and } \arg (z+rx_{m})<\frac{3\pi}{4}, \\
J_{4}^{-1}, & \mbox{for } \arg z < -\frac{3\pi}{4} \mbox{ and } \arg (z+rx_{m})>-\frac{3\pi}{4}, \\
J_{5}^{-1}, & \mbox{for } \arg z > -\frac{\pi}{4} \mbox{ and } \arg (z-rx_{m})<-\frac{\pi}{4}, 
\end{cases}
\end{align}
and for $|z|<\delta$, define
\begin{align}
& \widetilde{P}^{(0)}(z) := -\sqrt{\tfrac{3}{2\pi}}e^{-\frac{\rho^{2}}{6}}i \Psi_{0}^{-1} \widetilde{\Psi}(z) \nonumber \\
& \times \left( I + \sum_{j=1}^{m} \frac{s_{j}-s_{j+1}}{2\pi i} \Big( \log(z-rx_{j})-\log(z+rx_{j}) \Big) \begin{pmatrix}
0 & 1 & 1 \\
0 & 0 & 0 \\
0 & 0 & 0
\end{pmatrix} \right) \begin{cases}
J_{1}^{-1}, & z \in \mathrm{I}, \\
I, & z \in \mathrm{II}, \\
J_{2}^{-1}, & z \in \mathrm{III}, \\
J_{2}^{-1}J_{3}^{-1}, & z \in \mathrm{IV}, \\
J_{1}^{-1}J_{0}^{-1}J_{5}^{-1}, & z \in \mathrm{V}, \\
J_{1}^{-1}J_{0}^{-1}, & z \in \mathrm{VI},
\end{cases} \label{def of Ptp0p}
\end{align}
where the principal branches are chosen for the log, and we recall that $s_{m+1}:=1$, the regions $\mathrm{I}$, $\mathrm{II}$, $\mathrm{III}$, $\mathrm{IV}$, $\mathrm{V}$ and $\mathrm{VI}$ are shown in Figure \ref{fig:contour for Phi Sine}, $\widetilde{\Psi}$ is defined in \eqref{def of Psi tilde}, and the matrices $J_{j}$, $j=0,\ldots,5$ are defined in \eqref{def of Jj}. Let
\begin{align}\label{def of Rtilde}
\widetilde{R}(z) := \begin{cases}
\Phi(z)\widetilde{N}(z)^{-1}, & |z| > \delta, \\
\Phi(z)\widetilde{P}^{(0)}(z)^{-1}, & |z|<\delta.
\end{cases}
\end{align}
The definitions \eqref{def of N tilde} and \eqref{def of Ptp0p} also ensure that $\widetilde{R}$ has no jumps on $\cup_{j=0}^{5}\Sigma_{j}^{(r)}$. Since 
\begin{align}\label{lol15}
\widetilde{P}_{-}^{(0)}(z)^{-1}\widetilde{P}_{+}^{(0)}(z) = J_{5}J_{0}J_{1} \begin{pmatrix}
1 & s_{j}-1 & s_{j}-1 \\
0 & 1 & 0 \\
0 & 0 & 1
\end{pmatrix} = \begin{pmatrix}
1 & s_{j} & s_{j} \\
0 & 1 & 0 \\
0 & 0 & 1
\end{pmatrix} = \Phi_{-}(z)^{-1}\Phi_{+}(z)
\end{align}
holds for all $z \in (-rx_{j},-rx_{j-1})\cup (rx_{j-1},rx_{j})$, it follows that $\widetilde{R}(z)$ is also analytic in $(-rx_{m},rx_{m})\setminus \cup_{j=1}^{m-1}\{-rx_{j},rx_{j}\}$. Furthermore, from a direct inspection of \eqref{def of Ptp0p} and \eqref{def of Rtilde}, we see that the singularities of $\widetilde{R}(z)$ at $z=-rx_{m},\ldots,-rx_{1},0,rx_{1},\ldots,rx_{m}$ are removable. Hence, $\widetilde{R}(z)$ is analytic for $z \in \mathbb{C}\setminus \{z:|z|=\delta\}$. Let us orient the circle $|z|=\delta$ in the clockwise direction, and define $J_{\widetilde{R}}:=\widetilde{R}_{-}^{-1}\widetilde{R}_{+}$. By \eqref{def of Rtilde}, $J_{\widetilde{R}}(z)=\widetilde{P}^{(0)}(z)\widetilde{N}(z)^{-1}$, and by \eqref{asymp of Psi at infty}, \eqref{asymp of Phi at infty} and \eqref{def of Rtilde}, $\widetilde{R}(z) = I + \bigO(z^{-1})$ as $z \to \infty$. We also check using \eqref{def of N tilde} and \eqref{def of Ptp0p} that $J_{\widetilde{R}}(z)=I+\bigO(r)$ as $r \to 0$ uniformly for $|z|=\delta$. Thus
\begin{align}\label{asymp of Rtilde}
\widetilde{R}(z) = I + \bigO(r), \qquad \mbox{as } r \to 0 \mbox{ uniformly for } z \in \mathbb{C}\setminus \{z:|z|=\delta\}.
\end{align}

\section{Proof of the main results}\label{section: proof of main results}
\subsection{Proof of Theorem \ref{thm:asymp of pkj as r to inf}}
We already proved in Section \ref{section: Lax pair} that the functions $(p_{0},q_{0},\{p_{j,1},q_{j,1},p_{j,2},q_{j,2},p_{j,3},q_{j,3}\}_{j=1}^{m})$ defined by \eqref{def of p0 and q0 in proof}--\eqref{def of qj and pj in terms of Phijp0p} exist and satisfy \eqref{big system of ODEs}--\eqref{sum relation between qj and pj}. In this subsection we complete the proof of Theorem \ref{thm:asymp of pkj as r to inf} by obtaining the asymptotic formulas \eqref{all asymp as r to inf} and \eqref{all asymp as r to 0}.
\subsubsection*{Asymptotics of $p_{0}$, $q_{0}$, $p_{j,k}$ and $q_{j,k}$ as $r \to +\infty$}

We first compute the asymptotics of $p_{0}$ and $q_{0}$ using \eqref{def of p0 and q0 in proof}. Using \eqref{asymp of R as r to infty} and \eqref{Rp1p def}, we see that $R(z) = I + \frac{R_{1}}{z}+\bigO(z^{-2})$ as $z \to \infty$, where
\begin{align}\label{expression for R1}
R_{1} = C_{N}\widehat{\Psi}_{1}C_{N}^{-1}r^{-\frac{2}{3}} + \bigO(r^{-\frac{4}{3}}), \qquad \mbox{as } r \to +\infty.
\end{align}
Inverting the transformations $\Phi \mapsto T \mapsto S \mapsto R$ in the region outside the disks using \eqref{def of T}, \eqref{def of S} and \eqref{def of R}, we get
\begin{align*}
\Phi(rz) = \diag(r^{-\frac{1}{3}}, 1, r^{\frac{1}{3}})R(z)N(z)e^{\Theta(rz)}, \qquad z \in \mathbb{C}\setminus \big( \cup_{j=1}^{m}(\mathcal{D}_{x_{j}}\cup \mathcal{D}_{-x_{j}}) \cup \mathcal{D}_{0}  \big).
\end{align*}
From this expression, \eqref{asymp of Phi at infty} and \eqref{asymp of N at inf}, we deduce that
\begin{align*}
\Phi_{1} = r \; \diag(r^{-\frac{1}{3}}, 1, r^{\frac{1}{3}}) (R_{1}+N_{1}) \diag(r^{\frac{1}{3}}, 1, r^{-\frac{1}{3}}),
\end{align*}
which implies by \eqref{def of N1} and \eqref{expression for R1} that
\begin{align*}
& \Phi_{1,12} = i \sqrt{3} r^{\frac{2}{3}}\sum_{j=1}^{m}\beta_{j}x_{j}^{\frac{2}{3}} + \frac{\rho^{3}}{54}- \frac{\rho}{6} + \bigO(r^{-\frac{2}{3}}), & & \Phi_{1,23} = i \sqrt{3} r^{\frac{2}{3}}\sum_{j=1}^{m}\beta_{j}x_{j}^{\frac{2}{3}} + \frac{\rho^{3}}{54}+ \frac{\rho}{6} + \bigO(r^{-\frac{2}{3}}),
\end{align*}
as $r \to + \infty$. Substituting these asymptotics in \eqref{def of p0 and q0 in proof} gives the large $r$ asymptotics of $p_{0}$ and $q_{0}$ stated in \eqref{p0 large r in thm} and \eqref{q0 large r in thm}.

\medskip We now compute the large $r$ asymptotics of $p_{j,k}$, $j=1,\ldots,m$, $k=1,2,3$ using the definitions \eqref{def of qj and pj in terms of Phijp0p}. It follows from \eqref{asymp of Phi near rxj 1} and \eqref{asymp of Phi near rxj} that
\begin{align}\label{Phip0p as a limit}
\Phi_{j}^{(0)}(r) = \lim_{\substack{z \to x_{j}\\ z \in \mathrm{II}}} \Phi(rz)\begin{pmatrix}
1 & \mathfrak{s}_{j}\log(rz-rx_{j}) & \mathfrak{s}_{j}\log(rz-rx_{j}) \\
0 & 1 & 0 \\
0 & 0 & 1
\end{pmatrix}.
\end{align}
For $z \in \mathcal{D}_{x_{j}}$ and $z$ outside the lenses, by \eqref{def of T}, \eqref{def of S} and \eqref{def of R} we have
\begin{align}\label{Phi in Dxj in terms of R}
\Phi(rz) = \diag(r^{-\frac{1}{3}},1,r^{\frac{1}{3}}) R(z) P^{(x_{j})}(z)e^{\Theta(rz)}.
\end{align}
Hence, using also \eqref{def of Pxj}, we get
\begin{align}
& \Phi_{j}^{(0)}(r) = \diag(r^{-\frac{1}{3}},1,r^{\frac{1}{3}})R(x_{j})E_{x_{j}}(x_{j}) \lim_{\substack{z \to x_{j}\\ z \in \mathrm{II}}} \Bigg[\begin{pmatrix}
\Phi_{\mathrm{HG},11}(r^{\frac{4}{3}}f_{x_{j}}(z);\beta_{j}) & \Phi_{\mathrm{HG}, 12}(r^{\frac{4}{3}}f_{x_{j}}(z);\beta_{j}) & 0 \\
\Phi_{\mathrm{HG},21}(r^{\frac{4}{3}}f_{x_{j}}(z);\beta_{j}) & \Phi_{\mathrm{HG}, 22}(r^{\frac{4}{3}}f_{x_{j}}(z);\beta_{j}) & 0 \\
0 & 0 & 1
\end{pmatrix} \nonumber \\
& \times  (s_{j}s_{j+1})^{-\frac{\sigma_{3,1}}{4}}\widetilde{\Theta}(z) \begin{pmatrix}
1 & \mathfrak{s}_{j}\log(rz-rx_{j}) & \mathfrak{s}_{j}\log(rz-rx_{j}) \\
0 & 1 & 0 \\
0 & 0 & 1
\end{pmatrix} \Bigg] \label{lol16}
\end{align}
where
\begin{align}\label{def of Theta tilde}
\widetilde{\Theta}(z) = e^{\frac{1}{2}(\theta_{2}(rz)-\theta_{1}(rz))\sigma_{3,1}}\begin{pmatrix}
1 & 0 & 0 \\
0 & 1 & e^{\theta_{2}(rz)-\theta_{3}(rz)} \\
0 & 0 & 1
\end{pmatrix}e^{\Theta(rz)} = e^{-\frac{\theta_{3}(rz)}{2}}\begin{pmatrix}
1 & 0 & 0 \\
0 & 1 & 1 \\
0 & 0 & e^{\frac{3}{2}\theta_{3}(rz)}
\end{pmatrix}.
\end{align}
Using \eqref{betaj sj relation}, we note that
\begin{align*}
\frac{\sin(\pi\beta_{j})}{\pi} = \frac{1}{\Gamma(\beta_{j})\Gamma(1-\beta_{j})} = - \frac{\mathfrak{s}_{j}}{\sqrt{s_{j}s_{j+1}}}.
\end{align*}
Therefore, using \eqref{precise asymptotics of Phi HG near 0}, we can rewrite \eqref{lol16} as
\begin{align}
& \Phi_{j}^{(0)}(r) = e^{-\frac{\theta_{3}(rx_{j})}{2}}\diag(r^{-\frac{1}{3}},1,r^{\frac{1}{3}})R(x_{j})E_{x_{j}}(x_{j}) \nonumber \\
& \times \diag(\Upsilon_{j}^{(0)},1)\begin{pmatrix}
1 & \mathfrak{s}_{j}(\log r - \log(r^{\frac{4}{3}}|f_{x_{j}}'(x_{j})|i)) & \mathfrak{s}_{j}(\log r - \log(r^{\frac{4}{3}}|f_{x_{j}}'(x_{j})|i)) \\
0 & 1 & 1 \\
0 & 0 & e^{\frac{3}{2}\theta_{3}(rx_{j})}
\end{pmatrix}, \label{explicit expression for Phijp0p in large r analysis}
\end{align}
where
\begin{align}\label{def of Upsilonj p0p}
\Upsilon_{j}^{(0)} = \begin{pmatrix}
\frac{\Gamma(1-\beta_{j})}{(s_{j}s_{j+1})^{1/4}} & \frac{(s_{j}s_{j+1})^{1/4}}{\Gamma(\beta_{j})} \left( \frac{\Gamma^{\prime}(1-\beta_{j})}{\Gamma(1-\beta_{j})}+2\gamma_{\mathrm{E}} - i \pi \right) \\
\frac{\Gamma(1+\beta_{j})}{(s_{j}s_{j+1})^{1/4}} & \frac{-(s_{j}s_{j+1})^{1/4}}{\Gamma(-\beta_{j})} \left( \frac{\Gamma^{\prime}(-\beta_{j})}{\Gamma(-\beta_{j})} + 2\gamma_{\mathrm{E}} - i \pi \right)
\end{pmatrix},
\end{align}
and $\gamma_{\mathrm{E}}$ is Euler's gamma constant. Combining \eqref{explicit expression for Phijp0p in large r analysis} with \eqref{def of qj and pj in terms of Phijp0p}, we get
\begin{align*}
& \begin{pmatrix}
q_{j,1}(r) \;\; q_{j,2}(r) \;\; q_{j,3}(r)
\end{pmatrix}^{t} = e^{-\frac{\theta_{3}(rx_{j})}{2}}\diag(r^{-\frac{1}{3}},1,r^{\frac{1}{3}})R(x_{j})E_{x_{j}}(x_{j})\begin{pmatrix}
\Upsilon_{j,11}^{(0)} \;\; \Upsilon_{j,21}^{(0)} \;\; 0 
\end{pmatrix}^{t}, \\
& \begin{pmatrix}
p_{j,1}(r) \;\; p_{j,2}(r) \;\; p_{j,3}(r)
\end{pmatrix}^{t} = - e^{\frac{\theta_{3}(rx_{j})}{2}}\mathfrak{s}_{j}\diag(r^{\frac{1}{3}},1,r^{-\frac{1}{3}})R(x_{j})^{-t}E_{x_{j}}(x_{j})^{-t}   \diag((\Upsilon_{j}^{(0)})^{-t},1) \begin{pmatrix}
0 \;\; 1 \;\; 0
\end{pmatrix}^{t}.
\end{align*}
We then obtain \eqref{p1 large r in thm}, \eqref{p2 large r in thm}, \eqref{p3 large r in thm}, \eqref{q1 large r in thm}, \eqref{q2 large r in thm} and \eqref{q3 large r in thm} after a long computation. We omit the details.
\subsubsection*{Asymptotics of $p_{0}$, $q_{0}$, $p_{j,k}$ and $q_{j,k}$ as $r \to 0$}
Using \eqref{def of Rtilde} and \eqref{asymp of Rtilde}, we have, for $|z|>\delta$,
\begin{align*}
\Phi(z) = \big(\tfrac{\sqrt{2\pi}}{\sqrt{3}}e^{\frac{\rho^{2}}{6}}i\big)^{-1}(I+\bigO(r))\Psi_{0}^{-1}\Psi(z),  \qquad \mbox{as } r \to 0.
\end{align*}
Using also \eqref{asymp of Psi at infty} and \eqref{def of p0 and q0 in proof}, we obtain \eqref{asymp of p0 and q0 as r to 0}.

We now compute the asymptotics of $p_{j,k}$ and $q_{j,k}$, $j=1,\ldots,m$, $k=1,2,3$ as $r \to 0$. Recall that $p_{j,k}$ and $q_{j,k}$ are defined in \eqref{def of qj and pj in terms of Phijp0p}. Using \eqref{def of Ptp0p} and \eqref{def of Rtilde}, for $z \in \mathrm{II}\cap \{z:|z|<\delta \, r^{-1}\}$ we get
\begin{align*}
\Phi(rz) = \widetilde{R}(rz) \frac{\Psi_{0}^{-1}}{\frac{\sqrt{2\pi}}{\sqrt{3}}e^{\frac{\rho^{2}}{6}}i} \widetilde{\Psi}(rz) \left( I - \sum_{j=1}^{m} \mathfrak{s}_{j} \Big( \log(rz-rx_{j})-\log(rz+rx_{j}) \Big) \begin{pmatrix}
0 & 1 & 1 \\
0 & 0 & 0 \\
0 & 0 & 0
\end{pmatrix} \right).
\end{align*}
It then follows from \eqref{Phip0p as a limit} that
\begin{align*}
\Phi_{j}^{(0)}(r) = \widetilde{R}(rx_{j}) \frac{\Psi_{0}^{-1}}{\frac{\sqrt{2\pi}}{\sqrt{3}} e^{\frac{\rho^{2}}{6}}i} \widetilde{\Psi}(rx_{j}) \left( I + \sum_{j=1}^{m} \mathfrak{s}_{j} \log(2rx_{j})  \begin{pmatrix}
0 & 1 & 1 \\
0 & 0 & 0 \\
0 & 0 & 0
\end{pmatrix} \right),
\end{align*}
and then by \eqref{def of Psi tilde} and \eqref{asymp of Rtilde}, we get
\begin{align*}
\Phi_{j}^{(0)}(r) = \frac{\Psi_{0}^{-1}}{\frac{\sqrt{2\pi}}{\sqrt{3}}e^{\frac{\rho^{2}}{6}}i} (\widetilde{\Psi}(0)+\bigO(r)) \left( I + \sum_{j=1}^{m} \mathfrak{s}_{j} \log(2rx_{j})  \begin{pmatrix}
0 & 1 & 1 \\
0 & 0 & 0 \\
0 & 0 & 0
\end{pmatrix} \right).
\end{align*}
On the other hand, by \cite[equation (7.22)]{DXZ2020 thinning}, we have
\begin{align*}
\widetilde{\Psi}(0)\begin{pmatrix}
1 \;\; 0 \;\; 0
\end{pmatrix}^{t} = \begin{pmatrix}
\mathcal{P}_{0}(0) \;\; 0 \;\; \mathcal{P}_{0}''(0)
\end{pmatrix}^{t}, \qquad \widetilde{\Psi}(0)^{-t}\begin{pmatrix}
0 \;\; 1 \;\; 1
\end{pmatrix}^{t} = \begin{pmatrix}
0 \;\; \frac{1}{\mathcal{P}_{1}'(0)} \;\; 0
\end{pmatrix}^{t},
\end{align*}
where $\mathcal{P}_{1}'(0)\neq 0$. The asymptotics of \eqref{asymp of pjk as r to 0} and \eqref{asymp of qjk as r to 0} now directly follows from \eqref{def of qj and pj in terms of Phijp0p}. 

\subsection{Proof of Theorem \ref{thm: hamiltonian representation of the Pearcey determinant}}
The asymptotics of $H(r)$ as $r \to 0$ given by \eqref{asymp of H as r to 0 in thm} are directly obtained from \eqref{def of Hamiltonian} and \eqref{asymp of p0 and q0 as r to 0}--\eqref{asymp of qjk as r to 0}. Since $F(0\vec{x},\vec{u})=1$ by \eqref{F as expectation intro}, the integral representation \eqref{integral representation of F} follows by integrating \eqref{der of integral representation} from $0$ to an arbitrary $r>0$. To compute the asymptotics of $H(r)$ as $r \to + \infty$, we follow the method of \cite{DXZ2020 thinning} and rely on \eqref{Hamiltonian as a trace}. Using \eqref{asymp of Phi near rxj 1} and \eqref{asymp of Phi near rxj}, we obtain
\begin{align}\label{Phip1p as a limit}
\Phi_{j}^{(1)}(r) =\frac{1}{r}\Phi_{j}^{(0)}(r)^{-1} \lim_{\substack{z \to x_{j}\\ z \in \mathrm{II}}} \left[ \Phi(rz)\begin{pmatrix}
1 & \mathfrak{s}_{j}\log(rz-rx_{j}) & \mathfrak{s}_{j}\log(rz-rx_{j}) \\
0 & 1 & 0 \\
0 & 0 & 1
\end{pmatrix}\right]',
\end{align}
where $'$ denotes the derivative with respect to $z$. We see from \eqref{explicit expression for Phijp0p in large r analysis} that
\begin{align*}
\begin{pmatrix}
0 & 1 & 1
\end{pmatrix}\Phi_{j}^{(0)}(r)^{-1} = e^{\frac{1}{2}\theta_{3}(rx_{j})} \begin{pmatrix}
0 & 1 & 0
\end{pmatrix} \diag((\Upsilon_{j}^{(0)})^{-1},1) E_{x_{j}}(x_{j})^{-1}R(x_{j})^{-1}\diag(r^{\frac{1}{3}},1,r^{-\frac{1}{3}}).
\end{align*}
Also, by \eqref{Phi in Dxj in terms of R} and \eqref{def of Pxj}, we have
\begin{align*}
& \lim_{\substack{z \to x_{j}\\ z \in \mathrm{II}}} \left[ \Phi(rz)\begin{pmatrix}
1 & \mathfrak{s}_{j}\log(rz-rx_{j}) & \mathfrak{s}_{j}\log(rz-rx_{j}) \\
0 & 1 & 0 \\
0 & 0 & 1
\end{pmatrix}\right]'\begin{pmatrix}
1 \\ 0 \\ 0
\end{pmatrix} = \diag(r^{-\frac{1}{3}},1,r^{\frac{1}{3}}) \\
& \times \lim_{\substack{z \to x_{j}\\ z \in \mathrm{II}}} \left[  R(z) E_{x_{j}}(z) \begin{pmatrix}
\Phi_{\mathrm{HG},11}(r^{\frac{4}{3}}f_{x_{j}}(z);\beta_{j}) & \Phi_{\mathrm{HG}, 12}(r^{\frac{4}{3}}f_{x_{j}}(z);\beta_{j}) & 0 \\
\Phi_{\mathrm{HG},21}(r^{\frac{4}{3}}f_{x_{j}}(z);\beta_{j}) & \Phi_{\mathrm{HG}, 22}(r^{\frac{4}{3}}f_{x_{j}}(z);\beta_{j}) & 0 \\
0 & 0 & 1
\end{pmatrix} (s_{j}s_{j+1})^{-\frac{\sigma_{3,1}}{4}} \widetilde{\Theta}(z)\right]'\begin{pmatrix}
1 \\ 0 \\ 0
\end{pmatrix}
\end{align*}
where $\widetilde{\Theta}$ is defined in \eqref{def of Theta tilde}. A direct computation using \eqref{precise asymptotics of Phi HG near 0} shows that this last limit is given by
\begin{align}
& e^{-\frac{\theta_{3}(rx_{j})}{2}} \diag(r^{-\frac{1}{3}},1,r^{\frac{1}{3}}) \big[ R'(x_{j})E_{x_{j}}(x_{j})\diag(\Upsilon_{j}^{(0)},1) + R(x_{j})E_{x_{j}}'(x_{j})\diag(\Upsilon_{j}^{(0)},1)  \\
& +r^{\frac{4}{3}}f_{x_{j}}'(x_{j}) R(x_{j})E_{x_{j}}(x_{j}) \diag(\Upsilon_{j}^{(0)}\Upsilon_{j}^{(1)},0) - \tfrac{r\theta_{3}'(rx_{j})}{2}R(x_{j})E_{x_{j}}(x_{j})\diag(\Upsilon_{j}^{(0)},1) \big] \begin{pmatrix}
1 & 0 & 0
\end{pmatrix}^{t} \nonumber
\end{align}
where $\Upsilon^{(1)}_{j,21}= \frac{\beta_{j} \pi}{\sqrt{s_{j}s_{j+1}}\sin (\pi \beta_{j})}$. Combining the above equations with \eqref{Hamiltonian as a trace}, we then find
\begin{align}
& H(r) = -\frac{1}{r}\sum_{j=1}^{m}\mathfrak{s}_{j}x_{j} \Big[ \diag((\Upsilon_{j}^{(0)})^{-1},1) E_{x_{j}}(x_{j})^{-1}R(x_{j})^{-1}R'(x_{j})E_{x_{j}}(x_{j})\diag(\Upsilon_{j}^{(0)},1) \nonumber \\
& + \diag((\Upsilon_{j}^{(0)})^{-1},1) E_{x_{j}}(x_{j})^{-1}E_{x_{j}}'(x_{j})\diag(\Upsilon_{j}^{(0)},1)+ r^{\frac{4}{3}}f_{x_{j}}'(x_{j})  \diag(\Upsilon_{j}^{(1)},0) \Big]_{21}. \label{lol17}
\end{align}
The first part in the sum in \eqref{lol17} decays as $r \to + \infty$ by \eqref{asymp of R as r to infty}; more precisely
\begin{align*}
\diag((\Upsilon_{j}^{(0)})^{-1},1) E_{x_{j}}(x_{j})^{-1}R(x_{j})^{-1}R'(x_{j})E_{x_{j}}(x_{j})\diag(\Upsilon_{j}^{(0)},1) = \bigO(r^{-\frac{2}{3}}), \qquad \mbox{as } r \to + \infty.
\end{align*}
For the second part in \eqref{lol17}, we use \eqref{def of Upsilonj p0p} and get
\begin{align*}
& \big[\diag((\Upsilon_{j}^{(0)})^{-1},1) E_{x_{j}}(x_{j})^{-1}E_{x_{j}}'(x_{j})\diag(\Upsilon_{j}^{(0)},1)\big]_{21} = \tfrac{\Gamma(1-\beta_{j})^{2}}{\sqrt{s_{j}s_{j+1}}} \big[E_{x_{j}}(x_{j})^{-1}E_{x_{j}}'(x_{j})\big]_{21} \\
& - \tfrac{\Gamma(1+\beta_{j})^{2}}{\sqrt{s_{j}s_{j+1}}} \big[E_{x_{j}}(x_{j})^{-1}E_{x_{j}}'(x_{j})\big]_{12} + \tfrac{\Gamma(1-\beta_{j})\Gamma(1+\beta_{j})}{\sqrt{s_{j}s_{j+1}}}\Big( \big[E_{x_{j}}(x_{j})^{-1}E_{x_{j}}'(x_{j})\big]_{22}-\big[E_{x_{j}}(x_{j})^{-1}E_{x_{j}}'(x_{j})\big]_{11} \Big).
\end{align*}
This expression can be further simplified using \eqref{Exjxj der} and \eqref{def of cfrak j}. After a rather long computation, as $r \to + \infty$ we obtain
\begin{align*}
& \tfrac{\Gamma(1-\beta_{j})\Gamma(1+\beta_{j})}{\sqrt{s_{j}s_{j+1}}}\Big( \big[E_{x_{j}}(x_{j})^{-1}E_{x_{j}}'(x_{j})\big]_{22}-\big[E_{x_{j}}(x_{j})^{-1}E_{x_{j}}'(x_{j})\big]_{11} \Big) = \frac{1}{\mathfrak{s}_{j}}\bigg( \frac{\beta_{j}^{2}}{3x_{j}} + 2i \beta_{j} \im(d_{1,x_{j}}^{(1)}) \bigg) + \bigO(r^{-\frac{2}{3}}), \\
& \tfrac{\Gamma(1-\beta_{j})^{2}}{\sqrt{s_{j}s_{j+1}}} \big[E_{x_{j}}(x_{j})^{-1}E_{x_{j}}'(x_{j})\big]_{21} - \tfrac{\Gamma(1+\beta_{j})^{2}}{\sqrt{s_{j}s_{j+1}}} \big[E_{x_{j}}(x_{j})^{-1}E_{x_{j}}'(x_{j})\big]_{12} = \frac{2}{3\sqrt{3}x_{j}}\frac{i\beta_{j}}{\mathfrak{s}_{j}} \cos(2\vartheta_{j}(r)) + \bigO(r^{-\frac{2}{3}}),
\end{align*}
where we recall that $\vartheta_{j}(r)$ is defined in \eqref{def of vartheta j}. For the third and last term in \eqref{lol17}, we use \eqref{def of UpGamma 0 and 1} and \eqref{def of coeff expansion fxj} to write
\begin{align*}
\big[r^{\frac{4}{3}}f_{x_{j}}'(x_{j})  \diag(\Upsilon_{j}^{(1)},0)\big]_{21} = r^{\frac{4}{3}}f_{x_{j}}'(x_{j})\Upsilon_{j,21}^{(1)} = \frac{i\beta_{j}}{\mathfrak{s}_{j}}\bigg( -\sqrt{3}x_{j}^{\frac{1}{3}}r^{\frac{4}{3}}+\frac{\rho}{\sqrt{3}x_{j}^{1/3}}r^{\frac{2}{3}} \bigg).
\end{align*}
Substituting the above formulas in \eqref{lol17}, and noting the remarkable simplification
\begin{align*}
\sum_{j=1}^{m}2i \beta_{j}x_{j}\im (d_{1,x_{j}}^{(1)}) & =  \sum_{j=1}^{m} 2i\beta_{j}x_{j}\im \bigg( \frac{(\omega-5)\beta_{j}}{6(\omega-1)x_{j}} \bigg) + 
\sum_{1\leq j \neq k \leq m} \re \bigg( \frac{4(\omega-1)x_{k}^{2/3}\beta_{k}x_{j}^{2/3}\beta_{j}}{3(x_{j}^{2/3}-x_{k}^{2/3})(x_{j}^{2/3}-\omega x_{k}^{2/3})} \bigg) \\
& = \sum_{j=1}^{m} 2i\beta_{j}x_{j}\im \bigg( \frac{(\omega-5)\beta_{j}}{6(\omega-1)x_{j}} \bigg) = \sum_{j=1}^{m}\beta_{j}^{2},
\end{align*}
we obtain \eqref{asymp of H as r to inf in thm}. This finishes the proof of Theorem \ref{thm: hamiltonian representation of the Pearcey determinant}.

\subsection{Proof of Theorem \ref{thm: asymp of fredholm determinant}}
Integrating \eqref{magical identity} from $0$ to an arbitrary $r>0$, we get
\begin{align}
\int_{0}^{r}H(\tau)d\tau = & \int_{0}^{r} \bigg( p_{0}(\tau)q_{0}'(\tau)+\sum_{j=1}^{m}\sum_{k=1}^{3}p_{j,k}(\tau)q_{j,k}'(\tau) - H(\tau) \bigg) d\tau \nonumber \\
& -\frac{1}{4}\bigg[ 2p_{0}(\tau)q_{0}(\tau) + \sum_{j=1}^{m}\Big(p_{j,2}(\tau)q_{j,2}(\tau)+2p_{j,3}(\tau)q_{j,3}(\tau)\Big) -3\tau H(\tau) \bigg]_{\tau=0}^{r}. \label{lol7}
\end{align}
For convenience, we write $\vec{0}=(0,\ldots,0) \in \mathbb{R}^{m}$,
\begin{align*}
\vec{\beta}:=(\beta_{1},\ldots,\beta_{m}), \qquad \vec{\beta}_{\ell}:=(\beta_{1},\ldots,\beta_{\ell},0,\ldots,0), \qquad \vec{\beta}_{\ell}':=(\beta_{1},\ldots,\beta_{\ell-1},\beta_{\ell}',0,\ldots,0).
\end{align*}
We also write explicitly the dependence of $p_{j,k}$, $q_{j,k}$, $p_{0}$, $q_{0}$ and $H$ in $\beta_{1},\ldots,\beta_{m}$ using the notation $p_{j,k}(r;\vec{\beta})$, $q_{j,k}(r;\vec{\beta})$, $p_{0}(r;\vec{\beta})$, $q_{0}(r;\vec{\beta})$ and $H(r;\vec{\beta})$. By \eqref{betaj sj relation}, the parameter $\gamma$ in \eqref{differential identity for big sum} can also be chosen to be any parameter among $\beta_{1},\ldots,\beta_{m}$. Let $\ell \in \{1,\ldots,m\}$. Using \eqref{differential identity for big sum} with $\vec{\beta}= \vec{\beta}_{\ell-1}$ and $\gamma=\beta_{\ell}$, and integrating in $r$ from $0$ to $r$ and integrating in $\beta_{\ell}$ from $0$ to $\beta_{\ell}$, we get
\begin{align}
& \int_{0}^{r} \bigg( p_{0}(\tau;\vec{\beta})q_{0}'(\tau;\vec{\beta}_{\ell})+\sum_{j=1}^{m}\sum_{k=1}^{3}p_{j,k}(\tau;\vec{\beta}_{\ell})q_{j,k}'(\tau;\vec{\beta}_{\ell}) - H(\tau;\vec{\beta}_{\ell}) \bigg) d\tau \nonumber \\
& - \int_{0}^{r} \bigg( p_{0}(\tau;\vec{\beta}_{\ell-1})q_{0}'(\tau;\vec{\beta}_{\ell-1})+\sum_{j=1}^{m}\sum_{k=1}^{3}p_{j,k}(\tau;\vec{\beta}_{\ell-1})q_{j,k}'(\tau;\vec{\beta}_{\ell-1}) - H(\tau;\vec{\beta}_{\ell-1}) \bigg) d\tau \nonumber \\
& = \int_{0}^{\beta_{\ell}} \bigg( \sum_{k=1}^{3}\sum_{j=1}^{m}p_{j,k}(r;\vec{\beta}_{\ell}')\partial_{\beta_{\ell}'}q_{j,k}(r;\vec{\beta}_{\ell}')+p_{0}(r;\vec{\beta}_{\ell}')\partial_{\beta_{\ell}'} q_{0}(r;\vec{\beta}_{\ell}') \bigg)d\beta_{\ell}', \label{lol18} 
\end{align}
where we have used \eqref{asymp of p0 and q0 as r to 0}--\eqref{asymp of qjk as r to 0} to conclude that
\begin{align*}
\int_{0}^{\beta_{\ell}}\bigg( \sum_{k=1}^{3}\sum_{j=1}^{m}p_{j,k}(0;\vec{\beta}_{\ell}')\partial_{\beta_{\ell}'}q_{j,k}(0;\vec{\beta}_{\ell}')+p_{0}(0;\vec{\beta}_{\ell}')\partial_{\beta_{\ell}'} q_{0}(0;\vec{\beta}_{\ell}') \bigg)d\beta_{\ell}' = 0.
\end{align*}
Summing \eqref{lol18} over $\ell=1,\ldots,m$, we then obtain
\begin{align}
& \int_{0}^{r} \bigg( p_{0}(\tau;\vec{\beta})q_{0}'(\tau;\vec{\beta})+\sum_{j=1}^{m}\sum_{k=1}^{3}p_{j,k}(\tau;\vec{\beta})q_{j,k}'(\tau;\vec{\beta}) - H(\tau;\vec{\beta}) \bigg) d\tau \nonumber \\
& - \int_{0}^{r} \bigg( p_{0}(\tau;\vec{0})q_{0}'(\tau;\vec{0})+\sum_{j=1}^{m}\sum_{k=1}^{3}p_{j,k}(\tau;\vec{0})q_{j,k}'(\tau;\vec{0}) - H(\tau;\vec{0}) \bigg) d\tau \nonumber \\
& = \sum_{\ell=1}^{m}\int_{0}^{\beta_{\ell}} \bigg( \sum_{k=1}^{3}\sum_{j=1}^{m}p_{j,k}(r;\vec{\beta}_{\ell}')\partial_{\beta_{\ell}'}q_{j,k}(r;\vec{\beta}_{\ell}')+p_{0}(r;\vec{\beta}_{\ell}')\partial_{\beta_{\ell}'} q_{0}(r;\vec{\beta}_{\ell}') \bigg)d\beta_{\ell}'. \label{lol6}
\end{align}
If $\beta_{\ell}=0$ (or equivalently, if $s_{\ell}=s_{\ell+1}$), it follows from \eqref{def of qj and pj in terms of Phijp0p} that $p_{\ell,1}(r)=p_{\ell,2}(r)=p_{\ell,3}(r)=0$ for all $r> 0$. Also, by \eqref{Hamiltonian as a trace}, we have $H(r;\vec{0})=0$, and by \eqref{def of f and h} and \eqref{jumps of Y}, we have $Y|_{\vec{\beta}=\vec{0}} \equiv I$. Hence, by \eqref{def of p0 and q0 in proof} and the relations $\Phi_{1}=\Psi_{1}+\Psi_{0}^{-1}Y_{1}\Psi_{0}$ and \eqref{def of Psi0 and Psi1}, we have $p_{0}(r;\vec{0})=\frac{1}{\sqrt{2}}(\frac{\rho^{3}}{54}+\frac{\rho}{2})$ and $q_{0}(r;\vec{0})=\frac{1}{\sqrt{2}}(-\frac{\rho^{3}}{54}+\frac{\rho}{2})$. Thus,
\begin{align*}
\int_{0}^{r} \bigg( p_{0}(\tau;\vec{0})q_{0}'(\tau;\vec{0})+\sum_{j=1}^{m}\sum_{k=1}^{3}p_{j,k}(\tau;\vec{0})q_{j,k}'(\tau;\vec{0}) - H(\tau;\vec{0}) \bigg) d\tau = 0
\end{align*}
and \eqref{lol6} can be simplified as
\begin{multline}
\int_{0}^{r} \bigg( p_{0}(\tau;\vec{\beta})q_{0}'(\tau;\vec{\beta})+\sum_{j=1}^{m}\sum_{k=1}^{3}p_{j,k}(\tau;\vec{\beta})q_{j,k}'(\tau;\vec{\beta}) - H(\tau;\vec{\beta}) \bigg) d\tau  \\
 = \sum_{\ell=1}^{m}\int_{0}^{\beta_{\ell}} \bigg( \sum_{k=1}^{3}\sum_{j=1}^{\ell}p_{j,k}(r;\vec{\beta}_{\ell}')\partial_{\beta_{\ell}'}q_{j,k}(r;\vec{\beta}_{\ell}')+p_{0}(r;\vec{\beta}_{\ell}')\partial_{\beta_{\ell}'} q_{0}(r;\vec{\beta}_{\ell}') \bigg)d\beta_{\ell}'. \label{lol8}
\end{multline}
Substituting \eqref{lol8} in \eqref{lol7}, we obtain
\begin{align}
& \int_{0}^{r}H(\tau;\vec{\beta})d\tau =  \sum_{\ell=1}^{m}\int_{0}^{\beta_{\ell}}\bigg(\sum_{k=1}^{3}\sum_{j=1}^{\ell}p_{j,k}(r;\vec{\beta}_{\ell}')\partial_{\beta_{\ell}'}q_{j,k}(r;\vec{\beta}_{\ell}')+p_{0}(r;\vec{\beta}_{\ell}')\partial_{\beta_{\ell}'} q_{0}(r;\vec{\beta}_{\ell}') \bigg)d\beta_{\ell}' \label{lol9} \\
& -\frac{1}{4}\bigg( 2p_{0}(r;\vec{\beta})q_{0}(r;\vec{\beta}) - \sum_{j=1}^{m}\Big( 2p_{j,1}(r;\vec{\beta})q_{j,1}(r;\vec{\beta})+  p_{j,2}(r;\vec{\beta})q_{j,2}(r;\vec{\beta})\Big) -3rH(r;\vec{\beta}) \bigg)+ \frac{1}{2}p_{0}(0;\vec{\beta})q_{0}(0;\vec{\beta}), \nonumber
\end{align}
where we have also used \eqref{sum relation between qj and pj}. Using \eqref{p1 large r in thm}--\eqref{p3 large r in thm}, \eqref{q1 large r in thm}--\eqref{q3 large r in thm} and \eqref{der of asymp of R as r to infty}, we get
\begin{align}
& p_{j,1}(r;\vec{\beta}_{\ell})\partial_{\beta_{\ell}}q_{j,1}(r;\vec{\beta}_{\ell}) = p_{j,1}(r;\vec{\beta}_{\ell})q_{j,1}(r;\vec{\beta}_{\ell}) \partial_{\beta_{\ell}}\log q_{j,1}(r;\vec{\beta}_{\ell}) \nonumber \\
& = -\frac{2i\beta_{j}}{3}\bigg( \sin \big( 2\vartheta_{j}(r;\vec{\beta}_{\ell}) -\tfrac{2\pi}{3} \big) + \Big( \sin (2\vartheta_{j}(r;\vec{\beta}_{\ell})) - \tfrac{\sqrt{3}}{2} \Big) \sum_{n=1}^{\ell}\sqrt{3}i\beta_{n}\frac{x_{n}^{2/3}}{x_{j}^{2/3}} \bigg) \nonumber \\
& \quad \times \bigg( \partial_{\beta_{\ell}}\log \mathcal{A}_{j}(\vec{\beta}_{\ell}) +\cot (\vartheta_{j}(r;\vec{\beta}_{\ell})-\tfrac{\pi}{3}) \partial_{\beta_{\ell}}\vartheta_{j} \bigg)\bigg( 1+\bigO\Big(\frac{\log r}{r^{2/3}}\Big)  \bigg), \qquad \mbox{as } r \to + \infty, \label{lol19}
\end{align}
where we have explicitly written the dependence of $\vartheta_{j}$ and $\mathcal{A}_{j}$ in $r$ and $\vec{\beta}_{\ell}$. For $p_{j,2}(r;\vec{\beta}_{\ell})\partial_{\beta_{\ell}}q_{j,2}(r;\vec{\beta}_{\ell})$ and $p_{j,3}(r;\vec{\beta}_{\ell})\partial_{\beta_{\ell}}q_{j,3}(r;\vec{\beta}_{\ell})$, we obtain after another computation (using again \eqref{p1 large r in thm}--\eqref{p3 large r in thm}, \eqref{q1 large r in thm}--\eqref{q3 large r in thm} and \eqref{der of asymp of R as r to infty}) that
\begin{align}
& p_{j,2}(r;\vec{\beta}_{\ell})\partial_{\beta_{\ell}}q_{j,2}(r;\vec{\beta}_{\ell}) = -\frac{2i\beta_{j}}{3}\sin \big( 2\vartheta_{j}\big)  \bigg( \partial_{\beta_{\ell}}\log \mathcal{A}_{j}(\vec{\beta}_{\ell}) +\cot (\vartheta_{j}) \partial_{\beta_{\ell}}\vartheta_{j} \bigg)\bigg( 1+\bigO\Big(\frac{\log r}{r^{2/3}}\Big)  \bigg), \label{lol20} \\
& p_{j,3}(r;\vec{\beta}_{\ell})\partial_{\beta_{\ell}}q_{j,3}(r;\vec{\beta}_{\ell}) = \bigg[ \tfrac{-2i\beta_{j}}{3}\sin \big( 2\vartheta_{j}+\tfrac{2\pi}{3} \big)\bigg( \partial_{\beta_{\ell}}\log \mathcal{A}_{j}(\vec{\beta}_{\ell}) +\cot (\vartheta_{j}+\tfrac{\pi}{3}) \partial_{\beta_{\ell}}\vartheta_{j} \bigg) \nonumber \\
& + \tfrac{2i\beta_{j}}{3}\Big( \sin(2\vartheta_{j})-\tfrac{\sqrt{3}}{2} \Big) \bigg( \bigg\{\partial_{\beta_{\ell}}\log \mathcal{A}_{j}(\vec{\beta}_{\ell})+\cot(\vartheta_{j}-\tfrac{\pi}{3})\partial_{\beta_{\ell}}\vartheta_{j}\bigg\} \sum_{n=1}^{\ell}\sqrt{3}i\beta_{n}\tfrac{x_{n}^{2/3}}{x_{j}^{2/3}} + \sqrt{3}i\tfrac{x_{\ell}^{2/3}}{x_{j}^{2/3}} \bigg)\bigg] \nonumber \\
& \times \bigg( 1+\bigO\Big(\frac{\log r}{r^{2/3}}\Big) \bigg) \label{lol21}
\end{align}
as $r \to + \infty$, where $\vartheta_{j}=\vartheta_{j}(r;\vec{\beta}_{\ell})$. Furthermore, using \eqref{def of mathcal A j} and \eqref{betaj sj relation}, we find
\begin{align*}
& \partial_{\beta_{\ell}}\log \mathcal{A}_{j}(\vec{\beta}_{\ell}) = \begin{cases}
  -\frac{2\pi i}{3} + \partial_{\beta_{\ell}} \log |\Gamma (1-\beta_{\ell})|, & \mbox{if } j=\ell, \\[0.1cm]
\log \frac{|x_{j}^{2/3}-\omega x_{\ell}^{2/3}|}{|x_{j}^{2/3}- x_{\ell}^{2/3}|}, & \mbox{if } j<\ell.
\end{cases}
\end{align*}
Combining the asymptotics \eqref{lol19}--\eqref{lol21}, we get
\begin{align*}
\sum_{k=1}^{3}\sum_{j=1}^{\ell}p_{j,k}(r;\vec{\beta}_{\ell})\partial_{\beta_{\ell}}q_{j,k}(r;\vec{\beta}_{\ell}) = \sum_{j=1}^{\ell} \beta_{j}\bigg( \frac{x_{\ell}^{2/3}}{x_{j}^{2/3}}-\frac{2}{\sqrt{3}}\sin(2\vartheta_{j})\frac{x_{\ell}^{2/3}}{x_{j}^{2/3}} - 2 i \partial_{\beta_{\ell}}\vartheta_{j} \bigg) + \bigO\bigg( \frac{\log r}{r^{2/3}} \bigg)
\end{align*}
as $r \to + \infty$, where again $\vartheta_{j}=\vartheta_{j}(r;\vec{\beta}_{\ell})$. Using \eqref{useful relation} and the fact that $p_{j,1}(r)=p_{j,2}(r)=p_{j,3}(r)=0$ if $\beta_{j}=0$, we note that
\begin{align*}
p_{0}(r;\vec{\beta}_{\ell})\partial_{\beta_{\ell}}q_{0}(r;\vec{\beta}_{\ell}) = - \sqrt{2}\sum_{j=1}^{\ell}p_{j,3}(r;\vec{\beta}_{\ell})q_{j,1}(r;\vec{\beta}_{\ell})\partial_{\beta_{\ell}}q_{0}(r;\vec{\beta}_{\ell}) - \bigg( q_{0}(r;\vec{\beta}_{\ell})-\frac{\rho}{\sqrt{2}} \bigg)\partial_{\beta_{\ell}}q_{0}(r;\vec{\beta}_{\ell}).
\end{align*}
Integrating this identity in $\beta_{\ell}$ from $0$ to an arbitrary $\beta_{\ell}\in i \mathbb{R}$, we get
\begin{align}
\int_{0}^{\beta_{\ell}}p_{0}(r;\vec{\beta}_{\ell}')\partial_{\beta_{\ell}'}q_{0}(r;\vec{\beta}_{\ell}')d\beta_{\ell}' = & \; - \sqrt{2}\sum_{j=1}^{\ell} \int_{0}^{\beta_{\ell}} p_{j,3}(r;\vec{\beta}_{\ell}')q_{j,1}(r;\vec{\beta}_{\ell}')\partial_{\beta_{\ell}'}q_{0}(r;\vec{\beta}_{\ell}')d\beta_{\ell}' \nonumber \\
& - \bigg[ \frac{1}{2} q_{0}(r;\vec{\beta}_{\ell}')^{2}-\frac{\rho}{\sqrt{2}}q_{0}(r;\vec{\beta}_{\ell}') \bigg]_{\beta_{\ell}'=0}^{\beta_{\ell}}. \label{lol22}
\end{align}
Using \eqref{p3 large r in thm}, \eqref{q0 large r in thm}, \eqref{q1 large r in thm} and \eqref{der of asymp of R as r to infty}, as $r \to + \infty$ we get
\begin{align}\label{lol28}
& - \sqrt{2}\sum_{j=1}^{\ell} p_{j,3}(r;\vec{\beta}_{\ell})q_{j,1}(r;\vec{\beta}_{\ell})\partial_{\beta_{\ell}}q_{0}(r;\vec{\beta}_{\ell}) = \sum_{j=1}^{\ell}  \frac{x_{\ell}^{2/3}}{x_{j}^{2/3}}\beta_{j}\bigg( \frac{2}{\sqrt{3}}\sin(2\vartheta_{j}(r;\vec{\beta}_{\ell}))-1 \bigg)+ \bigO\bigg( \frac{\log r}{r^{2/3}} \bigg).
\end{align}
Hence, substituting \eqref{lol19}--\eqref{lol28} in \eqref{lol9}, we obtain
\begin{align}
& \int_{0}^{r}H(\tau;\vec{\beta})d\tau =  - 2 \sum_{\ell=1}^{m}\int_{0}^{\beta_{\ell}}\bigg\{\sum_{j=1}^{\ell-1} i\beta_{j} \partial_{\beta_{\ell}'}\vartheta_{j}(r;\vec{\beta}_{\ell}')+  i\beta_{\ell}' \partial_{\beta_{\ell}'}\vartheta_{\ell}(r;\vec{\beta}_{\ell}')   \bigg\}d\beta_{\ell}' \nonumber \\
& -\bigg( \frac{1}{2} q_{0}(r;\vec{\beta})^{2}-\frac{\rho}{\sqrt{2}}q_{0}(r;\vec{\beta}) \bigg) + \bigg( \frac{1}{2} q_{0}(r;\vec{0})^{2}-\frac{\rho}{\sqrt{2}}q_{0}(r;\vec{0}) \bigg) + \frac{1}{2}p_{0}(0;\vec{\beta})q_{0}(0;\vec{\beta}) \nonumber \\
& -\frac{1}{4}\bigg( 2p_{0}(r;\vec{\beta})q_{0}(r;\vec{\beta}) - \sum_{j=1}^{m}\Big( 2p_{j,1}(r;\vec{\beta})q_{j,1}(r;\vec{\beta})+  p_{j,2}(r;\vec{\beta})q_{j,2}(r;\vec{\beta})\Big) -3rH(r;\vec{\beta}) \bigg) \nonumber \\
& + \bigO\Big( \frac{\log r}{r^{2/3}} \Big), \qquad \mbox{as } r \to + \infty. \label{lol23}
\end{align}
Since $p_{0}(r;\vec{0})=\frac{1}{\sqrt{2}}(\frac{\rho^{3}}{54}+\frac{\rho}{2})$, $q_{0}(r;\vec{0})=\frac{1}{\sqrt{2}}(-\frac{\rho^{3}}{54}+\frac{\rho}{2})=q_{0}(0;\vec{\beta})$, we see that
\begin{align}
\bigg( \frac{1}{2} q_{0}(r;\vec{0})^{2}-\frac{\rho}{\sqrt{2}}q_{0}(r;\vec{0}) \bigg) + \frac{1}{2}p_{0}(0;\vec{\beta})q_{0}(0;\vec{\beta}) = -\frac{\rho}{2 \sqrt{2}}q_{0}(0;\vec{\beta}) = \frac{\rho^{4}}{216}-\frac{\rho^{2}}{8}. \label{lol24}
\end{align}
Also, using \eqref{all asymp as r to inf}, \eqref{asymp of H as r to inf in thm} and \eqref{useful relation}, we obtain
\begin{align}
& -\bigg( \frac{1}{2} q_{0}(r)^{2}-\frac{\rho}{\sqrt{2}}q_{0}(r) \bigg) -\frac{1}{4}\bigg( 2p_{0}(r)q_{0}(r) - \sum_{j=1}^{m}\Big( 2p_{j,1}(r)q_{j,1}(r)+  p_{j,2}(r)q_{j,2}(r)\Big) -3rH(r) \bigg) \nonumber \\
& = \frac{q_{0}(r)}{\sqrt{2}}\bigg( \frac{\rho}{2}+\sum_{j=1}^{m}p_{j,3}(r)q_{j,1}(r) \bigg) +\frac{3}{4}rH(r) + \frac{1}{4}\sum_{j=1}^{m}\Big( 2p_{j,1}(r)q_{j,1}(r)+ p_{j,2}(r)q_{j,2}(r)\Big) \nonumber \\
& = \sum_{j=1}^{m}\bigg( \frac{3\sqrt{3}}{4}i\beta_{j} (rx_{j})^{\frac{4}{3}} - \frac{\sqrt{3}\rho}{2}i\beta_{j} (rx_{j})^{\frac{2}{3}} - \beta_{j}^{2} \bigg) - \frac{\rho^{4}}{216} + \frac{\rho^{2}}{8} + \bigO (r^{-\frac{2}{3}}), \qquad \mbox{as } r \to + \infty. \label{lol25}
\end{align}
It follows from \cite[equation (7.51)]{DXZ2020 thinning} that
\begin{align}\label{lol26}
& -2\sum_{j=1}^{m}\int_{0}^{\beta_{\ell}} i\beta_{\ell}' \partial_{\beta_{\ell}'}\vartheta_{\ell}(r;\vec{\beta}_{\ell}') d\beta_{\ell}' = \sum_{j=1}^{m} \log\big( G(1+\beta_{j})G(1-\beta_{j}) \big) + \sum_{j=1}^{m}\beta_{j}^{2}\bigg( 1-\frac{4}{3}\log(rx_{j}) - \log\bigg( \frac{9}{2} \bigg) \bigg).
\end{align}
For $m \geq 2$, we also need the relation
\begin{align}\label{lol27}
& - 2 \sum_{\ell=1}^{m}\int_{0}^{\beta_{\ell}}\sum_{j=1}^{\ell-1} i\beta_{j} \partial_{\beta_{\ell}'}\vartheta_{j}(r;\vec{\beta}_{\ell}') d\beta_{\ell}' = -2 \sum_{\ell=1}^{m}\sum_{j=1}^{\ell-1} \beta_{j}\beta_{\ell} \log \frac{|x_{j}^{2/3}-\omega x_{\ell}^{2/3}|}{|x_{j}^{2/3}- x_{\ell}^{2/3}|},
\end{align}
which can be proved directly from \eqref{def of vartheta j} and a direct computation. The asymptotic formula \eqref{asymptotics in main thm} (without the error term) now follows after substituting \eqref{lol24} and \eqref{lol25} in \eqref{lol23} and performing a rather long calculation which uses \eqref{lol26} and \eqref{lol27}. The fact the the error term in \eqref{asymptotics in main thm} is $\bigO(r^{-\frac{2}{3}})$ and not $\bigO(r^{-\frac{2}{3}}\log r)$ follows directly from \eqref{integral representation of F} and \eqref{asymp of H as r to inf in thm}, and \eqref{derivative of error term} follows from \eqref{der of asymp of R as r to infty}. This finishes the proof of Theorem \ref{thm: asymp of fredholm determinant}.

\appendix

\section{Confluent hypergeometric model RH problem}\label{subsection: model RHP with HG functions}
In this appendix we recall a well-known model RH problem, whose solution depends on a parameter $\beta \in i \mathbb{R}$ and is denoted $\Phi_{\mathrm{HG}}(\cdot) = \Phi_{\mathrm{HG}}(\cdot ; \beta)$.
\begin{itemize}
\item[(a)] $\Phi_{\mathrm{HG}} : \mathbb{C} \setminus \Sigma_{\mathrm{HG}} \rightarrow \mathbb{C}^{2 \times 2}$ is analytic, where $\Sigma_{\mathrm{HG}}=e^{\frac{\pi i}{4}}(-\infty,\infty)\cup e^{\frac{\pi i}{2}}(-\infty,\infty) \cup e^{\frac{3\pi i}{4}}(-\infty,\infty)$.
\item[(b)] $\Phi_{\mathrm{HG}}$ satisfies the jump relations
\begin{equation}\label{jumps PHG3}
\Phi_{\mathrm{HG},+}(z) = \Phi_{\mathrm{HG},-}(z)\widetilde{J}_{k}, \qquad z \in \Gamma_{k}^{\mathrm{HG}}, \quad k=1,\ldots,6,
\end{equation}
where 
\begin{align*}
& \Gamma_{1}^{\mathrm{HG}} = (0,i\infty), & & \Gamma_{2}^{\mathrm{HG}} = (0,e^{\frac{3\pi i}{4}}\infty), & & \Gamma_{3}^{\mathrm{HG}} = (e^{-\frac{3\pi i}{4}}\infty,0), \\
& \Gamma_{4}^{\mathrm{HG}} = (-i\infty,0), & & \Gamma_{5}^{\mathrm{HG}} = (e^{-\frac{\pi i}{4}}\infty,0), & & \Gamma_{6}^{\mathrm{HG}} = (0,e^{\frac{\pi i}{4}}\infty),
\end{align*}
and
\begin{align*}
& \widetilde{J}_{1} = \begin{pmatrix}
0 & e^{-i\pi \beta} \\ -e^{i\pi\beta} & 0
\end{pmatrix}, \;\; \widetilde{J}_{4} = \begin{pmatrix}
0 & e^{i\pi\beta} \\ -e^{-i\pi\beta} & 0
\end{pmatrix}, \;\; \widetilde{J}_{2} = \widetilde{J}_{6} = \begin{pmatrix}
1 & 0 \\ e^{i\pi\beta} & 1
\end{pmatrix}, \;\; \widetilde{J}_{3} = \widetilde{J}_{5} = \begin{pmatrix}
1 & 0 \\ e^{-i\pi\beta} & 1
\end{pmatrix}.
\end{align*}
\item[(c)] As $z \to \infty$, $z \notin \Sigma_{\mathrm{HG}}$, we have
\begin{equation}\label{Asymptotics HG}
\Phi_{\mathrm{HG}}(z) = \left( I +  \frac{\Phi_{\mathrm{HG},1}(\beta)}{z} + \bigO(z^{-2}) \right) z^{-\beta\sigma_{3}}e^{-\frac{z}{2}\sigma_{3}}\left\{ \begin{array}{l l}
\displaystyle e^{i\pi\beta  \sigma_{3}}, & \displaystyle \tfrac{\pi}{2} < \arg z <  \tfrac{3\pi}{2}, \\
\begin{pmatrix}
0 & -1 \\ 1 & 0
\end{pmatrix}, & \displaystyle -\tfrac{\pi}{2} < \arg z < \tfrac{\pi}{2},
\end{array} \right.
\end{equation}
where $z^{\beta} = |z|^{\beta}e^{i\beta \arg z}$ with $\arg z \in (-\frac{\pi}{2},\frac{3\pi}{2})$ and
\begin{equation}\label{def of tau}
\Phi_{\mathrm{HG},1}(\beta) = \beta^{2} \begin{pmatrix}
-1 & \tau(\beta) \\ - \tau(-\beta) & 1
\end{pmatrix}, \qquad \tau(\beta) = \frac{- \Gamma\left( -\beta \right)}{\Gamma\left( \beta + 1 \right)}.
\end{equation}
\item[(d)] $\Phi_{\mathrm{HG}}(z) = \bigO(\log z)$ as $z \to 0$.
\end{itemize}
This model RH problem can be solved explicitly using confluent hypergeometric functions \cite{ItsKrasovsky}. By a computation similar to \cite[equation (A.9)]{ChSine} and \cite[equation (A.10)]{DXZ2020 thinning}, we have
\begin{equation}\label{precise asymptotics of Phi HG near 0}
\Phi_{\mathrm{HG}}(z) = \Upsilon^{(0)} (I + \Upsilon^{(1)}z+\bigO(z^{2})) \begin{pmatrix}
1 & \frac{\sin (\pi \beta)}{\pi} \log z \\
0 & 1
\end{pmatrix}, \qquad \mbox{as } z \to 0, \, \arg z \in (\tfrac{3\pi}{4},\tfrac{5\pi}{4}),
\end{equation}
where
\begin{align}\label{def of UpGamma 0 and 1}
\Upsilon^{(0)} = \begin{pmatrix}
\Gamma(1-\beta) & \frac{1}{\Gamma(\beta)} \left( \frac{\Gamma^{\prime}(1-\beta)}{\Gamma(1-\beta)}+2\gamma_{\mathrm{E}} - i \pi \right) \\
\Gamma(1+\beta) & \frac{-1}{\Gamma(-\beta)} \left( \frac{\Gamma^{\prime}(-\beta)}{\Gamma(-\beta)} + 2\gamma_{\mathrm{E}} - i \pi \right)
\end{pmatrix}, \qquad \Upsilon^{(1)}_{21}= \frac{\beta \pi}{\sin (\pi \beta)},
\end{align}
$\gamma_{\mathrm{E}}$ is Euler's gamma constant and
\begin{equation*}
\log z = \log |z| + i \arg z, \qquad \arg z \in \big(-\tfrac{\pi}{2},\tfrac{3\pi}{2}\big),
\end{equation*}

\paragraph{Acknowledgements.} C.C. acknowledges support from the European Research Council, Grant Agreement No. 682537, the Swedish Research Council, Grant No. 2015-05430, and the Ruth and Nils-Erik Stenb\"ack Foundation. P.M. acknowledges support from the Swedish Research Council, Grant No. 2017-05195.

\end{document}